\documentclass[11pt,conference,onecolumn]{IEEEtran}
\usepackage[ruled,vlined]{algorithm2e}
\usepackage{xcolor}
\definecolor{ieeelink}{RGB}{0,0,100}
\usepackage[colorlinks=true,
            linkcolor=ieeelink,
            urlcolor=ieeelink,
            citecolor=ieeelink,
            anchorcolor=ieeelink]{hyperref}
\usepackage{amsfonts,amssymb,amsthm,bm,graphicx,mathtools,physics,thm-restate,tikz,upgreek}
\usepackage[T1]{fontenc}
\theoremstyle{plain}
\newtheorem{theorem}{Theorem}
\newtheorem{proposition}[theorem]{Proposition}
\newtheorem{lemma}[theorem]{Lemma}
\newtheorem{corollary}[theorem]{Corollary}
\theoremstyle{definition}
\newtheorem{definition}[theorem]{Definition}

\newtheorem{question}[theorem]{Question}
\theoremstyle{remark}

\newcommand{\ci}{\ensuremath{\mathrm{i}}}

\newcommand{\cpi}{\ensuremath{\uppi}}

\def\BibTeX{{\rm B\kern-.05em{\sc i\kern-.025em b}\kern-.08em
    T\kern-.1667em\lower.7ex\hbox{E}\kern-.125emX}}

\newcommand{\revref}{}
\newcommand{\change}[1]{{#1}}
\newcommand{\changetwo}[1]{{#1}}

\IEEEoverridecommandlockouts

\begin{document}

\title{\changetwo{Quantum Glassiness From Efficient Learning}
}

\author{\IEEEauthorblockN{Eric R.\ Anschuetz}
\IEEEauthorblockA{\textit{Institute for Quantum Information and Matter \& Walter Burke Institute for Theoretical Physics} \\
\textit{Caltech}\\
Pasadena, CA, USA \\
\href{mailto:eans@caltech.edu}{eans@caltech.edu}}
}

\maketitle

\begin{abstract}
    We show a relation between quantum learning theory and algorithmic hardness. We \change{use the existence of efficient, local learning algorithms for energy estimation---such as the classical shadows algorithm---to prove that finding near-ground states of disordered quantum systems \changetwo{exhibiting a certain topological property} is impossible in the average case for Lipschitz quantum algorithms.} A corollary of our result is that many standard quantum algorithms fail to find near-ground states of these systems, including \change{time-$T$ Lindbladian dynamics from an arbitrary initial state, time-$T$ quantum annealing, phase estimation to $T$ bits of precision, and depth-$T$ variational quantum algorithms, whenever $T$ is less than some universal constant times the logarithm of the system size}.

    To achieve this, we introduce \changetwo{a generalization of the overlap gap property (OGP) for quantum systems that we call the \emph{quantum overlap gap property} (QOGP). This property is defined by a specific topological structure over representations of low-energy quantum states as output by an efficient local learning algorithm. We prove that preparing low-energy states of systems which exhibit the QOGP is intractable for quantum algorithms whose outputs are stable under perturbations of their inputs.} We then prove that the QOGP is satisfied for a sparsified variant of the quantum $p$-spin model, giving the first known algorithmic hardness-of-approximation result for quantum algorithms in finding the ground state of a non-stoquastic, noncommuting quantum system. Our resulting lower bound for quantum algorithms optimizing this model using Lindbladian evolution matches (up to constant factors) the best-known time lower bound for classical Langevin dynamics optimizing classical $p$-spin models. For this reason we suspect that finding ground states of typical \changetwo{instances of these quantum spin models} using quantum algorithms is, in practice, as intractable as the classical $p$-spin model is for classical algorithms. Inversely, we show that the Sachdev--Ye--Kitaev (SYK) model does not exhibit the QOGP, consistent with previous evidence that the model is rapidly mixing at low temperatures.
\end{abstract}

\setcounter{tocdepth}{2}
\tableofcontents

\section{Introduction}

\subsection{Motivation}

Quantum computers holds tremendous promise in revolutionizing our ability to understand and solve complex problems in physics. One especially favorable setting is in the preparation of low-energy states of a given quantum system. For instance, following physics heuristics~\cite{PhysRevD.94.106002}, there are known quantum algorithms for preparing near-ground states of the Sachdev--Ye--Kitaev (SYK) model~\cite{10.1145/3519935.3519960,araz2024thermal,basso2024optimizing}. The existence of such algorithms is surprising for two reasons:
\begin{enumerate}
    \item For worst-case instances of the disorder, the ground state problem for the SYK model is strongly believed to be difficult even for quantum computers (in particular, it is \textsc{QMA}-hard)~\cite{PhysRevLett.98.110503}.
    \item The ground state problem for analogous disordered classical systems is believed to be difficult even in the typical case~\cite{doi:10.1073/pnas.2108492118}.
\end{enumerate}
This poses the natural question, of particular importance in the search for candidate problems which showcase a quantum advantage:
\begin{question}
    When is the ground state problem for a disordered quantum system average-case hard for quantum computers?
\end{question}

One natural way to answer this question is by borrowing techniques used in the analogous classical setting. There, algorithmic hardness is characterized by the existence of a certain topological property of near-optimal states called the \emph{overlap gap property} (OGP)~\cite{doi:10.1073/pnas.2108492118}. In short, an energy function $E\left(\bm{x}\right)$ to be minimized over bit strings $\bm{x}\in\left\{0,1\right\}^{\times n}$ is said to satisfy the OGP if low-energy configurations have an extensive ``gap'' in Hamming distance $d_{\mathrm{H}}$. More concretely, denoting the minimum of $E\left(\bm{x}\right)$ as $E^\ast<0$, and defining the set of near-optimal configurations:
\begin{equation}
    S_\mu:=\left\{\bm{x}:E\left(\bm{x}\right)\leq\mu E^\ast\right\},
\end{equation}
a system is said to satisfy the OGP if there exist $0<\mu<1$ and $0<\nu_1<\nu_2<1$ such that:
\begin{equation}
    \left\{\left(\bm{x},\bm{y}\right)\in S_\mu^{\times 2}:d_{\mathrm{H}}\left(\bm{x},\bm{y}\right)\in\left[\nu_1 n,\nu_2 n\right]\right\}=\varnothing.
\end{equation}
An illustration of the low-energy space for a problem satisfying the overlap gap property is given in Fig.~\ref{fig:ogp}. For a variety of problems, the existence of (generalizations of) the OGP has been used to prove the failure of many classes of algorithms in achieving a given approximation ratio for typical problem instances~\cite{10.1145/2554797.2554831,doi:10.1137/140989728,pmlr-v65-david17a,doi:10.1137/22M150263X,gamarnik2022circuitlowerboundspspin,9996948,10.1214/23-AAP1953,gamarnik2023shatteringisingpurepspin}. Intuitively, this barrier poses a problem for algorithms whose outputs are Lipschitz functions of their inputs and are thus unable to ``jump'' the gap in Hamming distance of low-energy configurations. Obstructed algorithms include those which are state-of-the-art, such as approximate message passing (AMP), algorithms based on low-degree polynomials, $\operatorname{O}\left(\log\left(n\right)\right)$-time Langevin dynamics, and even certain quantum algorithms~\cite{farhi2020stable,PRXQuantum.4.010309,anshu2023concentrationbounds,anschuetz2024combinatorialnlts}.
\begin{figure}
    \begin{center}
        \includegraphics[width=0.5\linewidth]{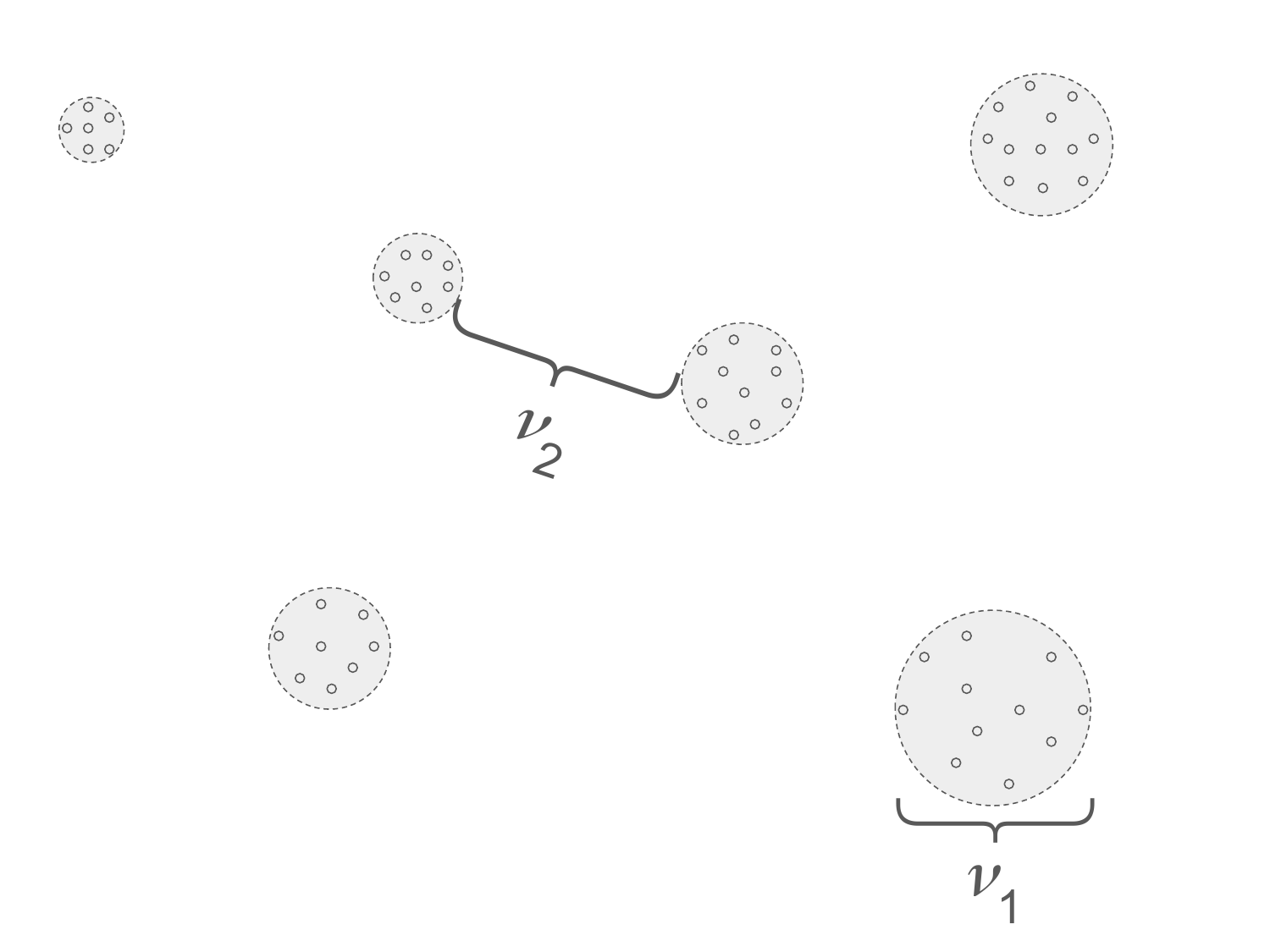}
        \caption{An illustration of the low-energy space of a classical spin system satisfying the overlap gap property. No two configurations (small, solid-bordered circles) achieving an approximation ratio $\mu$ have a normalized Hamming distance in the closed set $\left[\nu_1,\nu_2\right]$. Larger circles with dashed borders are an aid to the eye.\label{fig:ogp}}
    \end{center}
\end{figure}

The ground state problem for quantum systems differs in many ways from that of classical systems. Fundamentally these differences arise due to the state of quantum systems being described by a \emph{state vector} of dimension generally exponential in the system size. Whereas the state of a classical spin system on $n$ spins may be described by an $n$-bit string, the state of a quantum $n$-spin system is described by a \emph{quantum state} on $n$ qubits $\ket{\psi}\in\mathbb{C}^{2^n}$. Where one hopes to find the minimizer $\bm{x}\in\left\{0,1\right\}^{\times n}$ of some energy function $E\left(\bm{x}\right)$ in the classical setting, in the quantum setting one hopes to find the eigenvector $\ket{\psi}\in\mathbb{C}^{2^n}$ associated with the smallest eigenvalue of a Hermitian matrix $\bm{H}\in\mathbb{C}^{2^n\times 2^n}$ (known as the \emph{Hamiltonian}).

Unfortunately, these differences give rise to many barriers in directly proving the existence of some form of OGP for quantum systems. For one, in the classical setting one takes for granted the ability to measure the function value $E\left(\bm{x}\right)$ of a given configuration without disturbing the configuration itself. This is not possible in the quantum setting if one does not exactly prepare an \emph{energy eigenstate}, that is, an eigenstate of $\bm{H}$. More specifically, measuring the energy $E$ of a state $\ket{\psi}$ is equivalent to projecting onto an energy eigenstate $\ket{E}$ of $\bm{H}$, and:
\begin{equation}
    \ket{E}\bra{E}\ket{\psi}\neq\ket{\psi}\text{ in general}.
\end{equation}
Even if $\ket{\psi}$ \emph{were} an energy eigenstate, in some settings it is provably computationally expensive to exactly measure the energy of a quantum state, and an inexact approximation of this measurement may still disturb the state $\ket{\psi}$. These properties make it difficult to even \emph{define} what is meant by a ``quantum OGP,'' let alone prove any implications of hardness with one. Should a quantum OGP be defined over the eigenbasis of $\bm{H}$? What implications would that have on algorithms outputting states \emph{not} in the eigenbasis of $\bm{H}$? If defined over general states instead, how does one handle the fact that mixtures of low-energy states are also low-energy, meaning every low-energy state is connected to other low-energy states via a low-energy subspace?

\subsection{Contributions}

We here bypass these issues by borrowing ideas from a surprising place: quantum learning theory. We use as a tool the \emph{classical shadows} framework~\cite{huang2020predicting}, which was first introduced as a protocol for estimating the expectation values of many quantum observables in a given quantum state with minimal sample complexity. The protocol works by constructing a ``classical shadow'' representation---a classical, bit string representation---of the quantum state using randomized local measurements.

We consider the space of classical shadow representations of quantum states. We show that if a given class of random Hamiltonians:
\begin{enumerate}
    \item has an efficient classical shadows estimator, and
    \item exhibits a generalized version of the OGP over the space of classical shadow representations,
\end{enumerate}
\emph{any} Lipschitz quantum algorithm---even those which do not use the classical shadows protocol in any way---are unable to prepare low-energy states for this class of Hamiltonians. We call this generalization of the OGP the \emph{quantum overlap gap property} (QOGP). \change{Just as in the classical setting, we will later see that its presence inhibits quantum generalizations of $\operatorname{O}\left(\log\left(n\right)\right)$-time Langevin dynamics, as well as other quantum algorithms such as $\operatorname{O}\left(\log\left(n\right)\right)$-time quantum annealing.} Its implication on algorithmic hardness can be informally stated as follows.
\begin{theorem}[The QOGP implies hardness for stable quantum algorithms, informal \change{statement of Theorem~\ref{thm:m_qogp_implies_alg_hardness}}]\label{thm:alg_hard_stable_qas_inf}
    Consider a disordered system:
    \begin{equation}
        \bm{H}_{\bm{X}}=\frac{1}{\sqrt{m}}\sum_{i=1}^m X_i\bm{H}_i,
    \end{equation}
    for $\bm{H}_i$ fixed $k$-local operators, exhibiting the QOGP. Let $\mathfrak{d}$ be the degree of the interaction hypergraph of $\bm{H}_{\bm{X}}$ in expectation over the disorder $\bm{X}$. For any constant $L>0$ and sufficiently large $n$ and $k$, there exists no $\mathfrak{d}L$-Lipschitz quantum algorithm $\bm{\mathcal{A}}\left(\bm{X}\right)$ achieving a constant-factor approximation to the ground state with high probability over $\bm{H}_{\bm{X}}$. An algorithm is said to be $L$-Lipschitz if for any $\bm{X}$ and $\bm{Y}$:
    \begin{equation}
        W_2\left(\bm{\mathcal{A}}\left(\bm{X}\right),\bm{\mathcal{A}}\left(\bm{Y}\right)\right)\leq L\left\lVert\bm{X}-\bm{Y}\right\rVert_1.
    \end{equation}
\end{theorem}
Here, $W_2$ denotes what we call the quantum Wasserstein distance of order $2$, an immediate generalization of the well-known quantum Wasserstein distance of order $1$~\cite{9420734}. Informally, this distance is a quantum ``earth mover's'' metric in that two states which differ only by a channel acting on $\ell$ qubits differ in Wasserstein distance by $\operatorname{O}\left(\ell\right)$. While this metric is not unitarily invariant, it is still nonincreasing under convex combinations of tensor product channels. Furthermore, this Lipschitz condition is satisfied by many standard quantum algorithms used for ground state preparation.
\begin{proposition}[Standard quantum algorithms are Lipschitz, informal \change{statement of Corollaries~\ref{cor:p_trott_qa},~\ref{cor:phase_est_stab}, and~\ref{cor:lind_ev_stab}}]
    \changetwo{Fix a sparse interaction hypergraph in the following. There exist universal constants $C,L>0$ such that for sufficiently large $n$, $C\log\left(n\right)$-depth variational quantum algorithms, phase estimation \change{to $C\log\left(n\right)$ bits of precision}, and quantum annealing or Linbladian evolution for $C\log\left(n\right)$ time\footnote{In units where the associated Hamiltonian operator norm is $n$-independent.} are $L$-Lipschitz.}
\end{proposition}
At a high level, our strategy for proving Theorem~\ref{thm:alg_hard_stable_qas_inf} takes advantage of the fact that the Pauli shadows protocol~\cite{huang2020predicting}---a specific instantiation of the classical shadows framework---only acts locally. In particular, the quantum Wasserstein distance between two states is nonincreasing after the channel $\varPhi$ associated with the protocol is performed:
\begin{equation}
    W_2\left(\varPhi\left(\bm{\mathcal{A}}\left(\bm{X}\right)\right),\varPhi\left(\bm{\mathcal{A}}\left(\bm{Y}\right)\right)\right)\leq W_2\left(\bm{\mathcal{A}}\left(\bm{X}\right),\bm{\mathcal{A}}\left(\bm{Y}\right)\right)\leq L\left\lVert\bm{X}-\bm{Y}\right\rVert_1.
\end{equation}
We then show that this implies that, w.h.p., a sample drawn from the mixed state $\varPhi\left(\bm{\mathcal{A}}\left(\bm{X}\right)\right)$ is close in quantum Wasserstein distance to some state in the support of $\varPhi\left(\bm{\mathcal{A}}\left(\bm{Y}\right)\right)$ whenever $\left\lVert\bm{X}-\bm{Y}\right\rVert_1$ is small. By interpolating between independent problem instances, this ``stability'' allows us to then argue that, if $\bm{\mathcal{A}}$ outputted near-ground states w.h.p., $\varPhi\circ\bm{\mathcal{A}}$ must output states inhibited by the QOGP. Therefore, no such $\bm{\mathcal{A}}$ exists. Informally, stable algorithms cannot ``jump the gap'' in classical shadow representations of low-energy states.

Finally, we consider when the QOGP is satisfied. We examine a variant of the quantum $k$-spin model:
\begin{equation}
    \bm{H}_{k\mathrm{-spin}}=\frac{1}{\sqrt{\binom{n}{k}}}\sum_{\overline{i}\in\binom{\left[n\right]}{k}}\sum_{\bm{b}\in\left\{1,2,3\right\}^{\times k}}J_{\overline{i},\bm{b}}\prod_{j=1}^k\bm{\sigma}_{i_j}^{\left(b_j\right)},
\end{equation}
where $J_{\overline{i},\bm{b}}$ are i.i.d. standard Gaussian variables and the sum is over all $k$-local Pauli operators. The variant we consider sparsifies this model in a certain structured way.\footnote{We additionally show a weaker topological obstruction for the dense quantum $k$-spin model, under which we are able to prove a weaker algorithmic hardness result.} We prove that this sparsified model exhibits the QOGP, thus proving its algorithmic intractability for a large class of quantum algorithms. Our technique uses the first moment method: we show there is a ``sweet spot'' of pairs of configurations\footnote{In actuality we look at $m$-tuples of configurations to get better bounds, but here assume $m=2$ for simplicity.} that are:
\begin{itemize}
    \item Close enough in Wasserstein distance that there are not so many pairs at this distance.
    \item Far enough in Wasserstein distance that the energies of each state in a pair are not highly correlated.
\end{itemize}
Carefully choosing these distances, this suffices to show that the expected number of such pairs is exponentially small in $n$, and by Markov's inequality that the probability of having any such pair is exponentially small in $n$. This culminates in the following hardness result.
\begin{corollary}[Lipschitz quantum algorithms fail to optimize sparse quantum spin glass models, informal \change{statement of Corollary~\ref{cor:stab_algs_fail_sparse_q_spin_glass}}]
    Let $\bm{H}_{k,\mathrm{sparse}}$ be the sparsified quantum $k$-spin model we consider. For any constants $L\geq 0$ and $0<\gamma\leq 1$ and sufficiently large $n$ and $k$, $L$-Lipschitz quantum algorithms fail to achieve an approximation ratio $\gamma$ of the ground state energy for this model.
\end{corollary}
On the other hand, we are able to use known lower-bounds on the sample complexity of classical shadows estimators for fermionic systems~\cite{chen2024optimal,king2025triply} to show that the QOGP is \emph{not} satisfied by the SYK model. This is consistent with prior work indicating that the ground-state problem for the SYK model is typically easy for quantum computers~\cite{PhysRevD.94.106002,10.1145/3519935.3519960,araz2024thermal,anschuetz2024strongly,basso2024optimizing}, suggesting there is a surprisingly deep connection between learning theory and the hardness of finding the ground state of quantum systems.

\subsection{Discussion}\label{sec:discussion}

To the best of our knowledge, our construction is the first known example of average-case algorithmic hardness in finding the ground state of a non-stoquastic, noncommuting quantum system to a constant-factor approximation ratio. Previous constructions of average-case algorithmic hardness have relied on one of these two properties~\cite{gamarnik2024slowmixingquantumgibbs,rakovszky2024bottlenecksquantumchannelsfinite,placke2024topological}, as either property implies the existence of a classical description for the partition function. The systems we study here have no known faithful, classical description as they are \textsc{QMA}-hard for worst-case instances of the disorder~\cite{doi:10.1137/140998287}.

Our results give further evidence as to the nontrivial relation between the sample complexity of learning and quantum glassiness. This connection was first pointed out in \revref\cite{anschuetz2024strongly}, where the authors found that disordered systems for which there is no constant-sample complexity classical shadows estimator---such as the SYK model---exhibited non-glassy behavior. Here, our results can be taken as a sort of converse: Lipschitz quantum algorithms fail to find the ground state of disordered systems with constant-sample complexity energy estimators. This mirrors what is known in the classical setting, where the onset of a glassy phase is known to coincide with the onset of algorithmic hardness in finding low-energy states of the system~\cite{doi:10.1073/pnas.2108492118}.

Furthermore, our resulting lower bound for quantum algorithms optimizing $\bm{H}_{k,\mathrm{sparse}}$ using Lindbladian evolution matches (up to constant factors) the best-known time lower bound for classical Langevin dynamics optimizing classical $p$-spin models~\cite{doi:10.1137/22M150263X}. For this reason we suspect that finding ground states of typical quantum $p$-spin models using quantum algorithms is, in practice, as intractable as the classical $p$-spin model is for classical algorithms. This would suggest that fermionic systems are more promising than spin systems as a problem setting for showcasing a practical quantum advantage in finding near-ground states, at least in a mean-field setting.

There are a couple of natural directions for future work. First, while our technique suffices for demonstrating the algorithmic hardness of the ground state problem for glassy quantum systems, it unfortunately does not give a sense of the topological structure of low-energy states in the full Hilbert space. Analyzing this structure for a given disordered system would seem to require a better handle on the system's eigenbasis than we are able to achieve here. Perhaps a natural first step would be to consider systems believed physically to have exponentially many ``phases''---subspaces of the low-energy space which are separated by high-complexity quantum circuits~\cite{haah2016invariant}. We hope to explore this more in future work.

Second, our technique is loose in the sense that once we reduce to classical shadow representations, we no longer make use of the fact that these representations came from valid quantum states. As a simple example of how this can be problematic, one could imagine performing a randomized measurement of the state $\ket{0}$, measuring in the $Z$ basis half of the time, and in the $X$ basis half of the time. The former measurement leaves the state unchanged, while the second measurement results in a mixed state of the form:
\begin{equation}
    \frac{1}{2}\ket{+}\bra{+}+\frac{1}{2}\ket{-}\bra{-}.
\end{equation}
If one were to estimate an energy expectation value by sampling many times from copies of these resulting states, it is possible---though extremely unlikely---to only measure $\ket{+}$ when measuring in the $X$ basis. Taken in combination with the $Z$ measurements, this would give an energy estimate for a state as if it were maximally polarized in both the $Z$ and $X$ directions. Such a state cannot physically exist due to the uncertainty principle, though we are unable to exclude such ``states'' in our analysis, loosening our bounds. We believe that resolving this issue is a prerequisite for examining the hardness of approximation at approximation ratios $\gamma$ which depend on the locality $k$ of the system---say, $\gamma=1-\operatorname{o}_k\left(1\right)$---which our current techniques are unable to handle. We hope to address this shortcoming in future work.

\section{Main Results}

\subsection{Preliminaries}

\subsubsection{Notation and Quantum Mechanics}

We begin by defining general notation we use throughout. We use the physics notation $\ket{\psi}\in\mathbb{C}^N$ to denote vectors and $\bra{\psi}$ to denote vectors in the adjoint space. We use $\left\lVert\cdot\right\rVert_{\mathrm{op}}$ to denote the operator norm and $\left\lVert\cdot\right\rVert_\ast$ to denote the trace norm (also known as the nuclear norm). For bounded Hermitian operators $\bm{H}\in\mathbb{C}^{N\times N}$, these can be expressed as:
\begin{align}
    \left\lVert\bm{H}\right\rVert_{\mathrm{op}}&=\max_{\ket{\psi}\in\mathbb{C}^N:\left\lvert\bra{\psi}\ket{\psi}\right\rvert=1}\bra{\psi}\bm{H}\ket{\psi},\\
    \left\lVert\bm{H}\right\rVert_\ast&=\Tr\left(\sqrt{\bm{H}^\dagger\bm{H}}\right).
\end{align}
We also use the notation $\odot$ to refer to the Hadamard product of two vectors, i.e., $\left(\bm{x}\odot\bm{y}\right)_i=x_i y_i$. We use the notation $\left[n\right]$ to denote the subset of natural numbers $\left[1,n\right]\cap\mathbb{N}$, the notation $\binom{\mathcal{S}}{k}$ to denote the set of all cardinality-$k$ subsets of a set $\mathcal{S}$, and the notation $2^{\mathcal{S}}$ to denote the power set of a set $\mathcal{S}$. Finally, we use the notation $\operatorname{O}\left(\cdot\right)$ to denote big O notation with respect to the variable $n$, and $\operatorname{O}_k\left(\cdot\right)$ to denote big O notation with respect to the variable $k$ (assumed to take place after the $n\to\infty$ limit).

We now define notation associated with quantum mechanical objects. We use $\mathcal{O}_{d,n}\subset\mathbb{C}^{d^n\times d^n}$ to denote the set of $n$-qudit Hermitian operators, often called \emph{observables}:
\begin{equation}
    \mathcal{O}_{d,n}=\left\{\bm{H}\in\mathbb{C}^{d^n\times d^n}:\bm{H}=\bm{H}^\dagger\right\}.
\end{equation}
When $d=2$---i.e., for $n$-qubit operators---we will often drop the $d$ and use the simpler notation $\mathcal{O}_n$. The space $\mathcal{O}_n$ has basis given by the (generalized) \emph{Pauli matrices}:
\begin{equation}
    \mathcal{P}_n=\left\{\bigotimes_{i=1}^n\bm{\sigma}^{\left(x_i\right)}:\bm{x}\in\left\{0,1,2,3\right\}^{\times n}\right\},
\end{equation}
where here $\bigotimes$ denotes the Kronecker product and:
\begin{equation}
    \bm{\sigma}^{\left(0\right)}:=\begin{pmatrix}1 & 0\\0 & 1\end{pmatrix},\quad
    \bm{\sigma}^{\left(1\right)}:=\begin{pmatrix}0 & 1\\1 & 0\end{pmatrix},\quad
    \bm{\sigma}^{\left(2\right)}:=\begin{pmatrix}0 & -\ci\\\ci & 0\end{pmatrix},\quad
    \bm{\sigma}^{\left(3\right)}:=\begin{pmatrix}1 & 0\\0 & -1\end{pmatrix}.
\end{equation}
If $\bm{P}=\bigotimes_{i=1}^n\bm{\sigma}^{\left(x_i\right)}\in\mathcal{P}_n$ is such that $\left\lVert\bm{x}\right\rVert_0=k$, we call $\bm{P}$ \emph{$k$-local}, and refer to the \emph{support} of $\bm{P}$ as:
\begin{equation}
    \operatorname{supp}\left(\bm{P}\right)=\left\{i\in\left[n\right]:x_i\neq 0\right\}.
\end{equation}
More generally, we will refer to the support of an operator $\operatorname{supp}\left(\bm{O}\right)$ as \changetwo{the complement of} the set of $i\in\left[n\right]$ such that $\bm{O}$ can be written as:
\begin{equation}
    \bm{O}=\bm{I}_d\otimes\bm{O}_i,
\end{equation}
with $\bm{I}_d$ the identity operator acting on the $i$th qudit and $\bm{O}_i$ some operator acting on the other $n-1$ qudits.

We will also consider a special subspace of observables $\mathcal{S}_n\subset\mathcal{O}_n$ known as the $n$-qubit \emph{pure quantum states}:
\begin{equation}
    \mathcal{S}_n=\left\{\bm{\rho}\in\mathcal{O}_n:\bm{\rho}\succeq\bm{0}\wedge\Tr\left(\bm{\rho}\right)=1\wedge\rank\left(\bm{\rho}\right)=1\right\}.
\end{equation}
We also use the notation $\mathcal{S}_n^{\mathrm{m}}:=\operatorname{Conv}\left(\mathcal{S}_n\right)$ as shorthand for the convex hull of $\mathcal{S}_n$, also known as the space of $n$-qubit \emph{mixed states}:
\begin{equation}
    \mathcal{S}_n^{\mathrm{m}}=\left\{\bm{\rho}\in\mathcal{O}_n:\bm{\rho}\succeq\bm{0}\wedge\Tr\left(\bm{\rho}\right)=1\right\},
\end{equation}
which can be interpreted as probability distributions over pure states. We will use the term \emph{expectation value} to describe the Frobenius inner product of a quantum state $\bm{\rho}\in\mathcal{S}_n^{\mathrm{m}}$ with an observable $\bm{H}\in\mathcal{O}_n$:
\begin{equation}
    \left\langle\bm{H}\right\rangle_{\bm{\rho}}=\Tr\left(\bm{\rho}\bm{H}\right).
\end{equation}
We will also use the term \emph{product state} to describe a state $\bm{\rho}\in\mathcal{S}_n^{\mathrm{m}}$ that can be written as a Kronecker product of $2\times 2$ states $\bm{\rho}_i\in\mathcal{S}_1^{\mathrm{m}}$:
\begin{equation}
    \bm{\rho}=\bigotimes_{i=1}^n\bm{\rho}_i.
\end{equation}
If $\left\{\ket{\psi}_i\right\}_{i=1}^{2^n}$ is an orthonormal basis of $\mathbb{C}^{2^n}$, we call $\left\{\ket{\psi}_i\bra{\psi}_i\right\}_{i=1}^{2^n}$ a \emph{product state basis} if each $\ket{\psi}_i\bra{\psi}_i\in\mathcal{S}_n$ is a product state. Relatedly, we will use the term \emph{computational basis} to refer to the basis mutually diagonalizing the product states:
\begin{equation}
    \ket{\bm{x}}\bra{\bm{x}}:=\bigotimes_{i=1}^n\ket{x_i}\bra{x_i},
\end{equation}
where
\begin{equation}
    \ket{0}\bra{0}:=\begin{pmatrix}1 & 0\\0 & 0\end{pmatrix},\quad
    \ket{1}\bra{1}:=\begin{pmatrix}0 & 0\\0 & 1\end{pmatrix},
\end{equation}
and $\bm{x}\in\left\{0,1\right\}^{\times n}$. More generally, we will use the term \emph{Pauli basis state} to refer to a state in an eigenbasis of an $n$-local Pauli operator. Finally, we will use the terminology \emph{quantum channel} (or \emph{channel}) to refer to a completely positive trace preserving linear map $\bm{\varLambda}:\mathcal{S}_n^{\mathrm{m}}\to\mathcal{S}_n^{\mathrm{m}}$. We refer the reader to \revref\cite{Nielsen_Chuang_2010} for more background on these concepts and on quantum mechanics.

\change{Finally, in full generality our main results depend on many parameters. In an attempt at clarity, we have attempted to use the same expression for a given value throughout the text (e.g., $L$ will refer to a Lipschitz constant almost everywhere). We have summarized these parameters in Table~\ref{tab:notation} along with their interpretations.}
\begin{table*}
    \begin{center}
        \caption{\change{Parameters used in the presentation of our results are given in the left column with corresponding meaning given in the right column.}\label{tab:notation}}
        \begin{tabular}{c|c}
            \change{$n$} & \change{System size}\\\hline
            \change{$k$} & \change{System locality}\\\hline
            \change{$d_{\mathrm{max}}$} & \change{Maximal interaction degree}\\\hline
            \change{$p$} & \change{Sparsity parameter}\\\hline
            \change{$E^\ast$} & \change{Limit superior of the maximal energy}\\\hline
            \change{$L$} & \change{Lipschitz constant}\\\hline
            \change{$f$} & \change{Error in Lipschitz property for approximately Lipschitz functions}\\\hline
            \change{$\mathfrak{d}$} & \change{Maximal interaction degree for which an algorithm is Lipschitz}\\\hline
            \change{$\kappa$} & \change{Parameterization of correlation between instances of the disorder}\\\hline
            \change{$p_{\mathrm{st}}$} & \change{Probability of failure of Lipschitz property}\\\hline
            \change{$\gamma$} & \change{Approximation ratio}\\\hline
            \change{$p_{\mathrm{f}}$} & \change{Probability of failing to achieve a large approximation ratio}\\\hline
            \change{$\delta$} & \change{Multiplicative error of classical shadows estimator}\\\hline
            \change{$p_{\mathrm{est}}$} & \change{Probability of failure of classical shadows estimator}\\\hline
            \change{$p_{\mathrm{b}}$} & \change{Probability of operator norm of the Hamiltonian greatly exceeding its mean}\\\hline
            \change{$m$} & \change{Number of replicas considered in the $m$-quantum overlap gap property ($m$-QOGP)}\\\hline
            \change{$\xi$} & \change{Fractional overlap ruled out by the $m$-QOGP}\\\hline
            \change{$\eta$} & \change{Width of fractional overlap ruled out by the $m$-QOGP}\\\hline
            \change{$\eta'$} & \change{Width of fractional overlap ruled out by the quantum chaos property}\\\hline
            \change{$\mathcal{I}$} & \change{Set of pairwise correlations considered by the $m$-QOGP}\\\hline
            \change{$R$} & \change{Number of classical shadow estimator samples considered in the $m$-QOGP}\\\hline
            \change{$r_{\mathrm{max}}$} & \change{Maximal number of terms in the classical shadows estimator with support a given hyperedge in the interaction graph}\\\hline
            \change{$c$} & \change{Fractional logarithmic cardinality of a set $\mathcal{S}$, i.e., $\left\lvert\mathcal{S}\right\rvert=\exp_2\left(cn\right)$}\\\hline
            \change{$F$} & \change{Lower bound on number of resampled disorder coefficients between correlated problem instances}
        \end{tabular}
    \end{center}
\end{table*}

\subsubsection{Problem Setting}

We consider the ground state problem for very general disordered quantum systems. More concretely, we will consider $n$-qubit, randomized systems with Hamiltonians of the form:
\begin{equation}\label{eq:random_ham}
    \bm{H}_{\bm{S};\bm{J}}=\frac{1}{\sqrt{Z\left(p,n\right)}}\sum_{i=1}^D S_i J_i\bm{H}_i\in\mathcal{O}_n,
\end{equation}
where the $S_i\in\left\{0,1\right\}$ are chosen i.i.d. from the Bernoulli distribution with sparsity parameter $\mathbb{E}\left[S_i\right]=p$, the $J_i$ are i.i.d. standard normal random variables, $D$ labels the dimension of $\bm{S}$ and $\bm{J}$, and the $\bm{H}_i$ are observables with $n$-independent operator norm. We will sometimes write this as:
\begin{equation}
    \bm{H}_{\bm{X}}=\bm{H}_{\bm{S};\bm{J}}
\end{equation}
for brevity, where $\bm{X}=\bm{S}\odot\bm{J}$. As $\bm{H}_{\bm{S};\bm{J}}$ is distributed identically to $-\bm{H}_{\bm{S};\bm{J}}$, the ground state problem is equivalent to the maximal-energy state problem for these systems; for this reason, and to avoid confusing negatives, we will from here on out consider the maximal-energy state problem. In what follows we will make reference to the \emph{interaction hypergraph} $G$ of operators of the form $\sum_{i=1}^D\bm{A}_i$, which is the $n$-vertex hypergraph with hyperedge $\overline{j}\in 2^{\left[n\right]}$ if there exists an $\bm{A}_i$ with $\overline{j}=\operatorname{supp}\left(\bm{A}_i\right)$. We will also refer to the \emph{interaction degree}, which is just the degree of the interaction hypergraph.

The normalization $Z\left(p,n\right)$ in Eq.~\eqref{eq:random_ham} is chosen such that the \emph{limiting maximal energy} $E^\ast$ is $n$-independent and finite:
\begin{equation}
    E^\ast:=\limsup_{n\to\infty}E_n^\ast:=\limsup_{n\to\infty}\frac{1}{\sqrt{n}}\mathbb{E}_{\left(\bm{S},\bm{J}\right)}\left[\left\lVert\bm{H}_{\bm{S};\bm{J}}\right\rVert_{\mathrm{op}}\right].
\end{equation}
Our choice of normalization by $\frac{1}{\sqrt{n}}$ is to match the computer science convention and, as previously mentioned, corresponds to the (normalized) ground state energy up to a sign. Even though this is only a statement in expectation, $E^\ast\sqrt{n}$ can also be interpreted as the typical maximal energy as disordered systems generally exhibit \emph{self-averaging}, i.e., the operator norm of $\frac{1}{\sqrt{n}}\bm{H}_{\bm{S};\bm{J}}$ exponentially concentrates around its mean.
\begin{proposition}[Self-averaging]\label{prop:self_averaging}
    Consider $\bm{H}_{\bm{S};\bm{J}}$ as in Eq.~\eqref{eq:random_ham} and let:
    \begin{equation}
        \Delta=\frac{1}{Z\left(p,n\right)}\sup_{\bm{\sigma}\in\mathcal{S}_n}\max_{B\in\binom{\left[D\right]}{2pD}}\sum_{i\in B}\Tr\left(\bm{H}_i\bm{\sigma}\right)^2.
    \end{equation}
    Then:
    \begin{equation}
        \mathbb{P}_{\left(\bm{S},\bm{J}\right)}\left[\left\lvert\frac{1}{\sqrt{n}}\left\lVert\bm{H}_{\bm{S};\bm{J}}\right\rVert_{\mathrm{op}}-E_n^\ast\right\rvert\geq t\right]\leq 2\exp\left(-\frac{t^2}{2\Delta}n\right)+\exp\left(-\frac{3}{8}pD\right).
    \end{equation}
\end{proposition}
\begin{proof}
    Conditioned on $\left\lVert\bm{S}\right\rVert_1\leq 2pD$, the first term follows immediately from the well-known self-averaging property of sums of matrices with i.i.d.\ Gaussian coefficients~\cite[Corollary~4.14]{bandeira2023matrix}. The final result follows from Bernstein's inequality:
    \begin{equation}
        \mathbb{P}_{\bm{S}}\left[\sum_{i=1}^D S_i\geq 2pD\right]\leq\exp\left(-\frac{\left(pD\right)^2}{2\left(pD+\frac{pD}{3}\right)}\right)=\exp\left(-\frac{3}{8}pD\right)
    \end{equation}
    and the union bound.
\end{proof}
In what follows we will assume that $D\geq\operatorname{\Omega}\left(n\right)$ and $p\geq\operatorname{\Omega}\left(\frac{n}{D}\right)$ such that self-averaging occurs.

\subsection{Stable Quantum Algorithms}\label{sec:stab_qas}

Our main result is demonstrating that, for certain problem classes $\mathcal{H}=\left\{\bm{H}_{\bm{X}}\right\}_{\bm{X}}$, the ground state problem is hard for a class of quantum algorithms we call \emph{stable quantum algorithms}. Intuitively, these are quantum algorithms whose outputs $\bm{\rho}\left(\bm{X}\right)\in\mathcal{S}_n^{\mathrm{M}}$ vary in a ``Lipschitz'' way with respect to the inputs $\bm{X}\in\mathbb{R}^D$ to the algorithm. Classically, this is formalized via approximate Lipschitz continuity in that there exist $f,L>0$ such that, with high probability over problem instances $\bm{X},\bm{Y}$ and the randomness of the classical algorithm $\bm{\mathcal{A}}:\mathbb{R}^D\to\left\{0,1\right\}^{\times n}$,
\begin{equation}\label{eq:class_stab_cond}
    d_{\mathrm{H}}\left(\bm{\mathcal{A}}\left(\bm{X}\right),\bm{\mathcal{A}}\left(\bm{Y}\right)\right)\leq f+L\left\lVert\bm{X}-\bm{Y}\right\rVert,
\end{equation}
with $\left\lVert\cdot\right\rVert$ some norm on $\mathbb{R}^D$. For instance, topological obstructions in classical optimization problems for classical stable algorithms using this definition were studied in \revref\cite{9996948} and \revref\cite{10.1214/23-AAP1953}.

Our first task is to define a natural notion of stability for quantum algorithms. A first guess might be Lipschitz continuity with respect to the trace distance; that is, one would might call a quantum algorithm $\bm{\rho}\left(\bm{X}\right)$ ``stable'' if there existed some $f,L>0$ such that, with high probability,
\begin{equation}
    \frac{1}{2}\left\lVert\bm{\rho}\left(\bm{X}\right)-\bm{\rho}\left(\bm{Y}\right)\right\rVert_\ast\leq f+L\left\lVert\bm{X}-\bm{Y}\right\rVert.
\end{equation}
However, it is easy to see that such a notion of stability would be an extremely strong imposition on any quantum algorithm. For instance, consider the simple case where $\bm{X}$ and $\bm{Y}$ are bit strings and $\bm{\rho}$ just encodes its input in the computational basis. Then, for any $\bm{X}\neq\bm{Y}$,
\begin{equation}
    \frac{1}{2}\left\lVert\bm{\rho}\left(\bm{X}\right)-\bm{\rho}\left(\bm{Y}\right)\right\rVert_\ast=1
\end{equation}
and in particular is independent of $\left\lVert\bm{X}-\bm{Y}\right\rVert$.

Instead, we consider stability under what we call the \emph{quantum Wasserstein distance of order $2$}, and take for concreteness the $L^1$-norm on the space of inputs.
\change{
\begin{definition}[\change{Stability, informal statement of Definition~\ref{def:stable_qas}}]\label{def:stable_qas_inf_mt}
    An algorithm $\bm{\rho}\left(\bm{X}\right)$ is said to be \emph{stable} if it is approximately Lipschitz with respect to the quantum Wasserstein distance of order $2$:
    \begin{equation}\label{eq:stab_informal}
        \left\lVert\bm{\rho}\left(\bm{X}\right)-\bm{\rho}\left(\bm{Y}\right)\right\rVert_{W_2}\leq f+L\left\lVert\bm{X}-\bm{Y}\right\rVert_1.
    \end{equation}
\end{definition}}
The quantum Wasserstein distance of order $2$ is an immediate generalization of the well-known quantum Wasserstein distance of order $1$~\cite{9420734}, and we review both in Appendix~\ref{sec:wass}. In the previous example, the quantum Wasserstein distance reduces to just the simple Hamming distance (see Corollary~\ref{cor:quantum_w1_product_states} of Appendix~\ref{sec:wass}):
\begin{equation}
    \left\lVert\bm{\rho}\left(\bm{X}\right)-\bm{\rho}\left(\bm{Y}\right)\right\rVert_{W_2}=d_{\mathrm{H}}\left(\bm{X},\bm{Y}\right),
\end{equation}
thus strictly generalizing the classical notion of stability discussed in Eq.~\eqref{eq:class_stab_cond}. Informally, the quantum Wasserstein distance is a quantum ``earth mover's'' metric in that states which differ only by a channel acting on $\ell$ qubits differ in Wasserstein distance by $\operatorname{O}\left(\ell\right)$. We show in Appendix~\ref{sec:stable_q_algs} that many well-studied quantum algorithms fall under this definition of stability, including $\operatorname{O}\left(\log\left(n\right)\right)$-depth variational algorithms and quantum Lindbladian dynamics (i.e., quantized Langevin dynamics) for $\operatorname{O}\left(\log\left(n\right)\right)$ steps on constant-degree interaction hypergraphs. These examples generalize well-known stable classical algorithms known to be obstructed by the classical OGP to the quantum setting~\cite{doi:10.1073/pnas.2108492118}.

We now give formal definitions for these concepts. To begin, we formally define our notion of a quantum algorithm, which in general may depend on some classical source of randomness.
\begin{definition}[Quantum algorithm]\label{def:qa}
    Let $\left(\varOmega,\mathbb{P}\right)$ be a probability space. A \emph{quantum algorithm} is a map $\bm{\mathcal{A}}:\mathbb{R}^D\times\varOmega\to\mathcal{S}_n^{\mathrm{m}}$.
\end{definition}
We say that the quantum algorithm is \emph{deterministic} if the associated probability space $\left(\varOmega,\mathbb{P}\right)$ is trivial; that is, if $\bm{\mathcal{A}}$ associates with each input $\bm{X}\in\mathbb{R}^D$ a quantum state $\bm{\rho}\in\mathcal{S}_n^{\mathrm{m}}$. Similarly, we call a quantum algorithm \emph{pure} if $\bm{\mathcal{A}}$ has codomain $\mathcal{S}_n\subset\mathcal{S}_n^{\mathrm{m}}$. This distinction may seem strange, as mixed states can be interpreted as probability distributions over pure states. Indeed, every deterministic quantum algorithm can be interpreted as the expected output of a nondeterministic, pure quantum algorithm.
\begin{definition}[Associated pure quantum algorithm]
    Let $\bm{\mathcal{A}}\left(\bm{\rho}\right)$ be a deterministic quantum algorithm. Let $\mathcal{U}$ be the uniform distribution over $\left[0,1\right]$. We call a pure quantum algorithm $\bm{\tilde{\mathcal{A}}}\left(\bm{\rho},\omega\right)$ satisfying:
    \begin{equation}
        \bm{\mathcal{A}}\left(\bm{\rho}\right)=\mathbb{E}_{\omega\sim\mathcal{U}}\left[\bm{\tilde{\mathcal{A}}}\left(\bm{\rho},\omega\right)\right]
    \end{equation}
    a \emph{pure quantum algorithm associated with $\bm{\mathcal{A}}$}.
\end{definition}
While this distinction between pure, nondeterministic algorithms and mixed, deterministic algorithms may seem pedantic, it will matter as the expected quantum Wasserstein distance between the outputs of deterministic algorithms may differ from that of their associated pure quantum algorithms:
\begin{equation}
    \left\lVert\bm{\mathcal{A}}\left(\bm{\rho}\right)-\bm{\mathcal{A}}\left(\bm{\sigma}\right)\right\rVert_{W_2}=\left\lVert\mathbb{E}_{\omega\sim\mathcal{U}}\left[\bm{\tilde{\mathcal{A}}}\left(\bm{\rho},\omega\right)-\bm{\tilde{\mathcal{A}}}\left(\bm{\sigma},\omega\right)\right]\right\rVert_{W_2}\neq\mathbb{E}_{\omega\sim\mathcal{U}}\left[\left\lVert\bm{\tilde{\mathcal{A}}}\left(\bm{\rho},\omega\right)-\bm{\tilde{\mathcal{A}}}\left(\bm{\sigma},\omega\right)\right\rVert_{W_2}\right].
\end{equation}

We now formally define \emph{stable quantum algorithms}, with a definition specialized toward the disorder we consider in Eq.~\eqref{eq:random_ham}. As previously described, informally, an algorithm is said to be ``stable'' if the output of the algorithm only varies slightly in quantum Wasserstein distance under small changes of the inputs. We loosen this requirement by only demanding this be the case conditioned on the problem instance living on a sparse interaction hypergraph. The factor of $f$ in the definition that follows gives the algorithm some ``wiggle room'' when the inputs are extremely close, and the $\kappa$ parameter further loosens the requirement by only requiring Lipschitzness for inputs that are correlated in a certain way. Finally, we allow the algorithm to fail to be stable with some probability $p_{\mathrm{st}}$. This definition strictly generalizes definitions of stable algorithms introduced in the classical algorithm literature~\cite{9996948,10.1214/23-AAP1953} to the quantum setting, and also generalizes the notion of stability we informally introduced in Definition~\ref{def:stable_qas_inf_mt}.
\begin{definition}[Stable quantum algorithm]\label{def:stable_qas}
    Let $\bm{\mathcal{A}}$ be a quantum algorithm with associated probability space $\left(\varOmega,\mathbb{P}_\varOmega\right)$ as in Definition~\ref{def:qa}. Furthermore, let $\mathbb{P}_{\bm{X},\bm{Y}}^{\mathfrak{d},\kappa}$ be any distribution over $\mathbb{R}^{2D}$ constructed in the following way:
    \begin{enumerate}
        \item $\bm{S}$ is sampled i.i.d.\ Bernoulli with sparsity parameter $p$, conditioned on the maximum interaction degree of $\bm{H}_{\bm{S};\bm{J}}$ being at most $\mathfrak{d}$.
        \item The marginal distribution of $\bm{X}$ or $\bm{Y}$ for $\left(\bm{X},\bm{Y}\right)\sim\mathbb{P}_{\bm{X},\bm{Y}}^{\mathfrak{d},\kappa}$ is distributed as $\bm{S}\odot\bm{J}$, where $\bm{J}$ is multivariate standard normal.
        \item There exists a subset $\mathcal{Q}\in\left[n\right]$ of cardinality $\kappa n$ with the property:
        \begin{equation}
            X_i=Y_i\iff\operatorname{supp}\left(\bm{H}_i\right)\subseteq\mathcal{Q}.
        \end{equation}
    \end{enumerate}

    $\bm{\mathcal{A}}$ is said to be \emph{$\left(f,L,\mathfrak{d},\mathcal{K},p_{\mathrm{st}}\right)$-stable} if, for all $\kappa'\in\mathcal{K}\subseteq\left[0,1\right]$ and $\mathbb{P}_{\bm{X},\bm{Y}}^{\mathfrak{d},\kappa'}$,
    \begin{equation}
        \mathbb{P}_{\left(\bm{X},\bm{Y},\omega\right)\sim\mathbb{P}_{\bm{X},\bm{Y}}^{\mathfrak{d},\kappa'}\otimes\mathbb{P}_\varOmega}\left[\left\lVert\bm{\mathcal{A}}\left(\bm{X},\omega\right)-\bm{\mathcal{A}}\left(\bm{Y},\omega\right)\right\rVert_{W_2}\leq f+L\left\lVert\bm{X}-\bm{Y}\right\rVert_1\right]\geq 1-p_{\mathrm{st}}.
    \end{equation}
    We use the notation \emph{$\left(f,L,\mathfrak{d},\kappa,p_{\mathrm{st}}\right)$-stable} in the case $\mathcal{K}=\left[\kappa,1\right]$.
\end{definition}

Finally, we define a \emph{near-optimal} algorithm for a problem class $\mathcal{H}=\left\{\bm{H}_{\bm{X}}\right\}_{\bm{X}}$, which intuitively are algorithms achieving near-optimal energy with high probability over problem instances drawn from $\mathcal{H}$. \change{Just as with stability, we first given an informal definition for clarity, followed by the formal definition we use in our work.}
\change{
\begin{definition}[Near-optimal quantum algorithm, informal statement of Definition~\ref{def:no_qas}]
    An algorithm $\bm{\rho}\left(\bm{X}\right)$ is said to be \emph{near-optimal} for a set of Hamiltonians $\left\{\bm{H}_{\bm{X}}\right\}_{\bm{X}}$ if it approximately achieves the maximal energy w.h.p.:
    \begin{equation}
        \frac{\Tr\left(\bm{H}_{\bm{X}}\bm{\rho}\left(\bm{X}\right)\right)}{\left\lVert\bm{H}_{\bm{X}}\right\rVert_{\mathrm{op}}}\geq\gamma.
    \end{equation}
\end{definition}
}
\begin{definition}[Near-optimal quantum algorithm]\label{def:no_qas}
    Let $\bm{\mathcal{A}}$ be a quantum algorithm with associated probability space $\left(\varOmega,\mathbb{P}_\varOmega\right)$ as in Definition~\ref{def:qa}. Let $\bm{H}_{\bm{X}}$ be a distribution of problem instances as in Eq.~\eqref{eq:random_ham}. Then, $\bm{\mathcal{A}}$ is said to be \emph{$\left(\gamma,p_{\mathrm{f}}\right)$-optimal} for $\mathcal{H}=\left\{\bm{H}_{\bm{X}}\right\}_{\bm{X}}$ if:
    \begin{equation}
        \mathbb{P}_{\left(\bm{X},\omega\right)}\left[\Tr\left(\bm{H}_{\bm{X}}\bm{\mathcal{A}}\left(\bm{X},\omega\right)\right)\geq\gamma E^\ast\sqrt{n}\right]\geq 1-p_{\mathrm{f}}.
    \end{equation}
\end{definition}

\subsection{The Quantum Overlap Gap Property}

Our first main result is that, for any system satisfying what we call the \emph{quantum overlap gap property} (QOGP), no stable algorithm is also near-optimal. The QOGP is comprised of two constituent properties:
\begin{enumerate}
    \item There exists an efficient, local \emph{classical shadows estimator} for the model.
    \item This classical shadows estimator has a ``disallowed'' region of configurations that achieve high energy.
\end{enumerate}
We here define both of these concepts.

\subsubsection{Classical Shadows Estimators}

We first define what we mean by an efficient local classical shadows estimator. Informally, this is a protocol which can efficiently estimate expectation values of a quantum state $\bm{\rho}\in\mathcal{S}_n$ via randomized local measurements of $\bm{\rho}$~\cite{huang2020predicting}. In what follows we focus on the case where these measurements are Pauli measurements. The measurement results can be encoded as bit string representations of Pauli basis states, which (on $n$ qubits) can uniquely be identified with elements of $\mathbb{Z}_6^{\times n}$. We use $\mathcal{B}_6$ to denote the set of pure, product, computational basis states on $n$ $6$-dimensional qudits which naturally encodes this representation. \change{Before proceeding with the formal definition we use in our work, we give an informal definition that is more clear.
\begin{definition}[Efficient local shadows estimator, informal statement of Definition~\ref{def:eff_loc_shad_est}]
    A class of random Hamiltonians is said to have an \emph{efficient local shadows estimator} if there exists a protocol using a constant number of copies of any given state and $1$-local measurements to achieve an estimate of the energy up to a multiplicative error $\delta$ w.h.p.
\end{definition}
}
\begin{restatable}[Efficient local shadows estimator]{definition}{locshad}\label{def:eff_loc_shad_est}
    Consider the class $\mathcal{H}$ of random Hamiltonians of the form:
    \begin{equation}
        \bm{H}_{\bm{X}}=\frac{1}{\sqrt{Z\left(p,n\right)}}\sum_{i=1}^D X_i\bm{H}_i
    \end{equation}
    with limiting maximal energy $E^\ast$. Assume there exists a quantum channel $\bm{\mathcal{M}}$ and a linear function $\bm{\mathcal{R}}$ satisfying the following properties:
    \begin{enumerate}
        \item \textbf{Locality:} There exists a subset $\mathcal{B}\subseteq\mathcal{B}_6$ such that $\bm{\mathcal{M}}:\mathcal{S}_n^{\mathrm{m}}\to\operatorname{Conv}\left(\mathcal{B}\right)$ and is a convex combination of tensor product channels, i.e., it is of the form:
        \begin{equation}
            \bm{\mathcal{M}}\left(\bm{\rho}\right):=\left(\frac{1}{B}\sum_{b=1}^B\bigotimes_{i=1}^n\bm{\mathcal{L}}_i^{\left(b\right)}\right)\left(\bm{\rho}\right)\in\operatorname{Conv}\left(\mathcal{B}\right)
        \end{equation}
        for some $B\in\mathbb{N}$ and local channels $\left\{\bm{\mathcal{L}}_i^{\left(b\right)}\right\}_{b\in\left[B\right],i\in\left[n\right]}$. We let $\bm{\tilde{\mathcal{M}}}:\mathcal{S}_n^{\mathrm{m}}\otimes\mathcal{U}\to\mathcal{B}$ denote an associated pure quantum channel of $\bm{\mathcal{M}}$.
        \item \textbf{Linearity:} There exists a linear function of the form:
        \begin{equation}\label{eq:lin_map_shadows}
            \bm{\mathcal{R}}\left(\bm{H}_{\bm{X}}\right)=\frac{1}{\sqrt{Z\left(p,n\right)}}\sum_{i=1}^D X_i\bm{R}_i
        \end{equation}
        such that
        \begin{equation}
            \Tr\left(\bm{\mathcal{R}}\left(\bm{H}_{\bm{X}}\right)\bm{\mathcal{M}}\left(\bm{\rho}\right)\right)=\Tr\left(\bm{H}_{\bm{X}}\bm{\rho}\right).
        \end{equation}
        \item \textbf{Precision:} With probability at least $1-p_{\mathrm{b}}$ over the disorder,
        \begin{equation}
            \mathbb{P}_{\omega\sim\mathcal{U}}\left[\Tr\left(\bm{\mathcal{R}}\left(\bm{H}_{\bm{S};\bm{J}}\right)\bm{\tilde{\mathcal{M}}}\left(\bm{\rho},\omega\right)\right)-\Tr\left(\bm{H}_{\bm{S};\bm{J}}\bm{\rho}\right)\geq-\delta E^\ast\sqrt{n}\right]\geq 1-p_{\mathrm{est}}
        \end{equation}
        for all $\bm{\rho}\in\mathcal{S}_n^{\mathrm{m}}$, where $p_{\mathrm{est}}$ is bounded away from $1$ by an $n$-independent constant. We call the probability $1-p_{\mathrm{b}}$ event $\mathcal{V}$.
    \end{enumerate}
    We say that
    \begin{equation}\label{eq:estimator}
        \mathcal{E}\left(\bm{H}_{\bm{X}},\bm{\rho},\omega\right):=\Tr\left(\bm{\mathcal{R}}\left(\bm{H}_{\bm{X}}\right)\bm{\tilde{\mathcal{M}}}\left(\bm{\rho},\omega\right)\right)
    \end{equation}
    is a \emph{$\left(\delta,p_{\mathrm{est}},p_{\mathrm{b}}\right)$-efficient local shadows estimator} for $\mathcal{H}$.
\end{restatable}
In Appendix~\ref{sec:classical_shadows} we give some examples of such estimators, including the traditional Pauli shadows estimator~\cite{huang2020predicting} and its derandomized variant~\cite{PhysRevLett.127.030503}. Interestingly, it is known there exists no constant-sample complexity shadows estimator for the SYK model~\cite{chen2024optimal,king2025triply}. This immediately tells us that the SYK model does not satisfy the QOGP.
\begin{proposition}[The SYK model is non-glassy]
    The SYK model does not satisfy the QOGP.
\end{proposition}
\begin{proof}
    There exists a known $\operatorname{poly}\left(n\right)$ sample complexity lower-bound for single-copy classical shadows estimators for the SYK model~\cite{chen2024optimal,king2025triply}; in our language, any local shadows estimator must have:
    \begin{equation}
        p_{\mathrm{est}}=1-\operatorname{o}\left(1\right).
    \end{equation}
    As we require $p_{\mathrm{est}}$ be bounded away from $1$ by a constant in Definition~\ref{def:eff_loc_shad_est}, there exists no efficient local shadows estimator for the SYK model.
\end{proof}

\subsubsection{Topological Structure of the Quantum Overlap Gap Property}

We now consider the topological aspect of the QOGP. We begin by defining a set $\mathcal{S}\left(\gamma,m,\xi,\eta,\mathcal{I},R\right)$ associated with an efficient local shadows estimator and a problem class $\mathcal{H}$. Intuitively, this set is composed of $m\times R$-tuples of $\gamma$-near-optimal Pauli basis states $\left\{\ket{\psi^{\left(t\right),\left(r\right)}}\right\}_{t\in\left[m\right],r\in\left[R\right]}$ for $m$ copies of $\bm{H}$ drawn from $\mathcal{H}$, where the basis states are constrained to be a quantum Wasserstein distance between $\frac{1-\xi}{2}n$ and $\frac{1-\xi+\eta}{2}n$ from one another on average over $r\in\left[R\right]$. The copies of $\bm{H}$ are also allowed to be correlated in a certain way captured by the set $\mathcal{I}$.
\begin{restatable}[$\mathcal{S}\left(\gamma,m,\xi,\eta,\mathcal{I},R\right)$]{definition}{ssetdef}\label{def:s_def}
    Recall that $\bm{\mathcal{R}}\left(\bm{H}_{\bm{S};\bm{J}}\right)$ is of the form:
    \begin{equation}
        \bm{\mathcal{R}}_{\bm{S};\bm{J}}:=\bm{\mathcal{R}}\left(\bm{H}_{\bm{S};\bm{J}}\right)=\frac{1}{\sqrt{pZ\left(n\right)}}\sum_{i=1}^D S_i J_i\bm{R}_i.
    \end{equation}
    Let $m\in\mathbb{N}$, $0<\gamma<1$, $0<\eta<\xi\leq 1$, and $\mathcal{I}\subseteq\left\{0,1\right\}^{\times D}$. Let $\bm{S}$ be a draw of $D$ i.i.d. Bernoulli random variables with sparsity parameter $p$, and let $\left\{\bm{J}^{\left(t\right)}\right\}_{t=0}^m$ be $D\left(m+1\right)$ i.i.d. draws of standard normal random variables. Define the interpolating randomness for all $t\in\left[m\right]$ and $\bm{\tau}\in\mathcal{I}$:
    \begin{equation}\label{eq:interp_randomness}
        \bm{X}^{\left(t\right)}\left(\bm{\tau}\right):=\left(\bm{1}-\bm{\tau}\right)\odot\bm{X}^{\left(0\right)}+\bm{\tau}\odot\bm{X}^{\left(t\right)}:=\left(\bm{1}-\bm{\tau}\right)\odot\frac{\bm{S}}{\sqrt{p}}\odot\bm{J}^{\left(0\right)}+\bm{\tau}\odot\frac{\bm{S}}{\sqrt{p}}\odot\bm{J}^{\left(t\right)}.
    \end{equation}
    Finally, define:
    \begin{equation}\label{eq:interp_path}
        \bm{\mathcal{R}}_{\bm{X}}^{\left(t\right)}\left(\bm{\tau}\right):=\frac{1}{\sqrt{Z\left(n\right)}}\sum_{i=1}^D X_i^{\left(t\right)}\left(\bm{\tau}\right)\bm{R}_i
    \end{equation}
    for all $t\in\left[m\right]$ and $\bm{\tau}\in\mathcal{I}$.

    Recall that the estimator $\bm{\tilde{\mathcal{M}}}$ has codomain $\mathcal{B}$. We denote by $\mathcal{S}\left(\gamma,m,\xi,\eta,\mathcal{I},R\right)$ the set of all $m\times R$-tuples of states $\ket{\psi^{\left(t\right),\left(r\right)}}\in\mathcal{B}$ which satisfy the following properties:
    \begin{itemize}
        \item \textbf{$\gamma$-optimality:}\footnote{Note that the largest eigenvalue of $\bm{\mathcal{R}}^{\left(t\right)}\left(\bm{\tau}_t\right)$ may differ from $E^\ast\sqrt{n}$.} There exists $\left\{\bm{\tau}_t\right\}_{t=0}^m\in\mathcal{I}^{\times m}$ such that, for all $t\in\left[m\right]$,
        \begin{equation}\label{eq:high_energy_req_s}
            \max_{r\in\left[R\right]}\bra{\psi^{\left(t\right),\left(r\right)}}\bm{\mathcal{R}}^{\left(t\right)}\left(\bm{\tau}_t\right)\ket{\psi^{\left(t\right),\left(r\right)}}\geq\gamma E^\ast\sqrt{n}.
        \end{equation}
        \item \textbf{Hamming distance bound:} For any $t\neq t'\in\left[m\right]$,
        \begin{equation}\label{eq:q_w_1_bound}
            \frac{1}{R}\sum_{r=1}^R\left\lVert\ket{\psi^{\left(t\right),\left(r\right)}}\bra{\psi^{\left(t\right),\left(r\right)}}-\ket{\psi^{\left(t'\right),\left(r\right)}}\bra{\psi^{\left(t'\right),\left(r\right)}}\right\rVert_{W_1}\in\left[\frac{1-\xi}{2}n,\frac{1-\xi+\eta}{2}n\right].
        \end{equation}
    \end{itemize}
\end{restatable}

We finally are able to define the quantum $m$-OGP. We begin by defining a quantum version of a weaker topological obstruction known as the \emph{chaos property}~\cite{doi:10.1073/pnas.2108492118} that we call the \emph{quantum chaos property}.
\begin{definition}[Quantum chaos property]\label{def:qcp}
    A class of random Hamiltonians with an $\left(\delta,p_{\mathrm{est}},p_{\mathrm{b}}\right)$-efficient local shadows estimator satisfies the quantum chaos property with parameters $\left(\gamma^\ast,m,\eta,R\right)$ if, for any $\gamma>\gamma^\ast$,
    \begin{equation}\label{eq:ind_cond}
        \mathbb{P}\left[\mathcal{S}\left(\gamma,m,1,\eta,\left\{\bm{1}\right\},R\right)\neq\varnothing\right]\leq\exp\left(-\operatorname{\Omega}\left(n\right)\right)
    \end{equation}
    as $n\to\infty$.
\end{definition}
In short, we say a problem satisfies the quantum chaos property if, with high probability over $m$ \emph{independent} draws of $\mathcal{H}$, high-energy states product states are far in quantum Wasserstein distance. The $m$-QOGP is a statement that there is an obstructed range of quantum Wasserstein distances even over \emph{correlated} draws of $\mathcal{H}$. We assume a certain correlation structure defined by what we call a \emph{$\left(c,F,R\right)$-correlation set} associated with $\bm{\mathcal{R}}_{\bm{S};\bm{J}}$.
\begin{definition}[$\left(c,F,R\right)$-correlation set]\label{def:cfr_corr_set}
    A set $\mathcal{I}\subseteq\left\{0,1\right\}^{\times D}$ is said to be a $\left(c,F,R\right)$-correlation set if it is of cardinality at most $2^{cn}$, $0<F\leq 1$, and the set satisfies the following two properties:
    \begin{enumerate}
        \item For any $\bm{\tau}\in\mathcal{I}$ there exists a subset $\mathcal{Q}_{\bm{\tau}}\subseteq\left[n\right]$ such that:
        \begin{equation}
            \tau_i=0\iff\operatorname{supp}\left(\bm{R}_i\right)\subseteq\mathcal{Q}_{\bm{\tau}}.
        \end{equation}
        \item If $R\neq 1$, for any $\bm{\tau}\neq\bm{0}\in\mathcal{I}$,
        \begin{equation}\label{eq:q_set_def_less_than_n}
            \left\lvert\mathcal{Q}_{\bm{\tau}}\right\rvert\leq\left(1-F\right)n.
        \end{equation}
    \end{enumerate}
\end{definition}
\begin{definition}[$m$-QOGP]\label{def:mqogp}
    A class of random Hamiltonians with an $\left(\delta,p_{\mathrm{est}},p_{\mathrm{b}}\right)$-efficient local shadows estimator satisfies the $m$-quantum overlap gap property ($m$-QOGP) with parameters $\left(\gamma^\ast,m,\xi,\eta,c,\eta',F,R\right)$ if $0<\eta<\xi<1$, $c>0$, and $\eta'>1-\xi+\eta$ such that, for any $\gamma>\gamma^\ast$, the following properties hold for any $\left(c,F,R\right)$-correlation set $\mathcal{I}$:
    \begin{enumerate}
        \item\begin{equation}\label{eq:nonind_cond}
        \mathbb{P}\left[\mathcal{S}\left(\gamma,m,\xi,\eta,\mathcal{I},R\right)\neq\varnothing\right]\leq\exp\left(-\operatorname{\Omega}\left(n\right)\right)
        \end{equation}
        as $n\to\infty$.
        \item The quantum chaos property is satisfied:
        \begin{equation}
            \mathbb{P}\left[\mathcal{S}\left(\gamma,m,1,\eta',\left\{\bm{1}\right\},R\right)\neq\varnothing\right]\leq\exp\left(-\operatorname{\Omega}\left(n\right)\right)
        \end{equation}
        as $n\to\infty$.
    \end{enumerate}
\end{definition}
We will later show that both the quantum chaos property and the $m$-QOGP obstruct certain classes of quantum algorithms from finding near-optimal states of $\mathcal{H}$, though the presence $m$-QOGP will allow us to rule out a larger class. As the quantum Wasserstein distance is just the Hamming distance over bit strings~\cite{9420734}, our two definitions strictly generalize their classical counterparts~\cite{doi:10.1073/pnas.2108492118}. Both the chaos property and the $m$-OGP are illustrated in Fig.~\ref{fig:ogp_generalizations}.
\begin{figure}
    \begin{center}
        \includegraphics[width=0.9\linewidth]{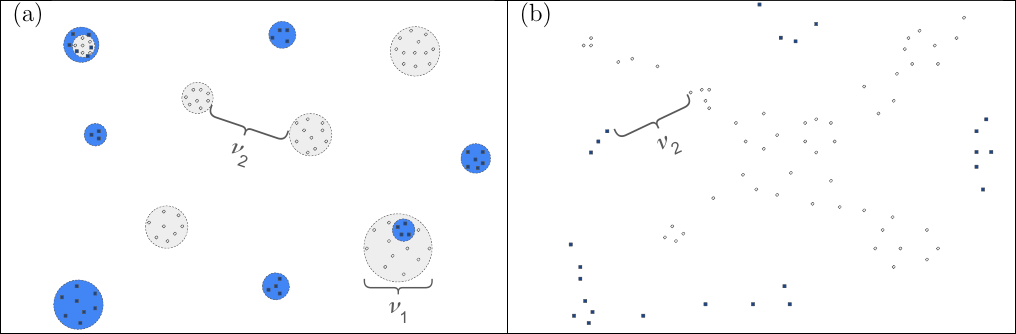}
        \caption{An illustration of the near-optimum space of a classical spin system satisfying (a) the $m$-overlap gap property ($m$-OGP) and (b) the chaos property with $m=2$. Small, solid-bordered circles (white) represent near-optimal solutions from one problem instance, and solid-bordered squares (blue) from another problem instance. The $m$-OGP is a statement that correlated problem instances have no near-optimal configurations with normalized Hamming distance in the closed set $\left[\nu_1,\nu_2\right]$. The chaos property only requires that near-optimal states of one problem instance are at least a normalized Hamming distance $\nu_2$ away from near-optimal states of an independent problem instance.\label{fig:ogp_generalizations}}
    \end{center}
\end{figure}

\subsection{Statement of Main Results}

Our first main result is that the $m$-QOGP implies algorithmic hardness for a large class of quantum algorithms. Roughly, the $m$-QOGP obstructs algorithms whose outputs change as $\operatorname{O}\left(\varDelta\right)$ in quantum Wasserstein distance when $\bm{X}$ is changed by $\varDelta$ in $L^1$ distance. As a supplementary result we also show that the quantum chaos property obstructs certain algorithms, though only those whose outputs change as $\operatorname{o}\left(\varDelta\right)$. The proof of this theorem is given as Sec.~\ref{sec:proof_of_alg_hardness}.
\begin{restatable}[$m$-QOGP implies algorithmic hardness]{theorem}{mqogpimpliesalghardness}\label{thm:m_qogp_implies_alg_hardness}
    Let $m$ be constant with respect to $n$. Assume the class $\mathcal{H}$ of $n$-qubit random Hamiltonians with limiting maximal energy $E^\ast$ has a $\left(\delta,p_{\mathrm{est}},p_{\mathrm{b}}\right)$-efficient local shadows estimator with $p_{\mathrm{est}}<1$. Let $\bm{\mathcal{R}}$ be the associated linear map (see Eq.~\eqref{eq:lin_map_shadows}). Assume that $\mathcal{H}$ satisfies either the $m$-QOGP with parameters $\left(\gamma^\ast,m,\xi,\eta,c,\eta',F,R\right)$ or the quantum chaos property with parameters $\left(\gamma^\ast,m,\eta,R\right)$ with $m$ independent of $n$. Fix
    \begin{equation}
        d_{\mathrm{max}}=r_{\mathrm{dense}}pd_{\mathrm{dense}}+b\sqrt{r_{\mathrm{dense}}pd_{\mathrm{dense}}\left(1-p\right)},
    \end{equation}
    where $b$ is a universal constant depending only on the $m$-QOGP parameters, $r_{\mathrm{dense}}$ is the maximum number of terms of $\bm{\mathcal{R}}\left(\bm{H}_{\bm{X}}\right)$ with support a given hyperedge in its interaction hypergraph:
    \begin{equation}
        r_{\mathrm{dense}}:=\max_{\overline{i}\in 2^{\left[n\right]}}\sum_{j=1}^D\bm{1}\left\{\overline{i}=\operatorname{supp}\left(\bm{R}_j\right)\right\},
    \end{equation}
    and $d_{\mathrm{dense}}$ the maximum degree of the interaction hypergraph of $\bm{\mathcal{R}}\left(\bm{H}_{\bm{X}}\right)$ when the sparsity $p=1$.

    If the $m$-QOGP holds, fix $Q\in\mathbb{N}$ independent of $n$; if only the quantum chaos property holds, fix $Q=F=1$. Furthermore, fix $\beta\in\mathbb{R}^+$ and parameters $f,L,\mathfrak{d},\kappa,p_{\mathrm{st}},\gamma,p_{\mathrm{f}}$ such that the following inequalities are satisfied:
    \begin{align}
        \mathfrak{d}&\geq d_{\mathrm{max}};\\
        \kappa&\leq\max\left(0,1-\frac{1.001}{Q}\right);\\
        F&\leq\frac{1}{Q};\\
        \frac{Q}{\beta^2}+Qp_{\mathrm{est}}^R&<1;\\
        \frac{\beta f}{n}+\frac{6d_{\mathrm{max}}\beta L}{Q}&\leq\frac{\eta}{8};\label{eq:stab_req}\\
        Q\exp_2\left(Q^{4mQ}\right)\left(3Qp_{\mathrm{st}}+3p_{\mathrm{f}}+p_{\mathrm{b}}\right)&\leq 1-\exp\left(-\operatorname{o}\left(n\right)\right);\\
        \gamma&>\gamma^\ast+\delta.
    \end{align}
    There exists no $\left(f,L,\mathfrak{d},\kappa,p_{\mathrm{st}}\right)$-stable and $\left(\gamma,p_{\mathrm{f}}\right)$-near optimal algorithm for $\mathcal{H}$.
\end{restatable}

As concrete applications of our first main result, we prove in Sec.~\ref{sec:qogp_is_exhibited} that the quantum chaos property and the $m$-QOGP are satisfied for two quantum spin glass models. First, we begin with the prototypical \emph{quantum $k$-spin} model with sparsification parameter $\operatorname{\Omega}\left(n^{-\left(k-1\right)}\right)\leq p\leq 1$:
\begin{equation}\label{eq:kloc_mod}
    \bm{H}_{k\mathrm{-spin}}=\frac{1}{\sqrt{p\binom{n}{k}}}\sum_{\overline{i}\in\binom{\left[n\right]}{k}}\sum_{\bm{b}\in\left\{1,2,3\right\}^{\times k}}S_{\overline{i},\bm{b}}J_{\overline{i},\bm{b}}\prod_{j=1}^k\bm{\sigma}_{i_j}^{\left(b_j\right)}.
\end{equation}
Here, recall that $\binom{\left[n\right]}{k}$ is the set of cardinality-$k$ subsets of $\left[n\right]=\left\{1,\ldots,n\right\}$, and that $\bm{\sigma}_i^{\left(1\right)},\bm{\sigma}_i^{\left(2\right)},\bm{\sigma}_i^{\left(3\right)}$ are the Pauli $X,Y,Z$ operators on qubit $i$, respectively. For this model:
\begin{align}
    r_{\mathrm{dense}}&=3^k,\\
    d_{\mathrm{dense}}&=\binom{n-1}{k-1}.
\end{align}
We show that $\bm{H}_{k\mathrm{-spin}}$ (for any choice of $p\geq\operatorname{\Omega}\left(n^{-\left(k-1\right)}\right)$) satisfies the quantum chaos property.
\begin{restatable}[The quantum $k$-spin model satisfies the quantum chaos property]{theorem}{klocqcp}\label{thm:k_qsg_qcp}
    The quantum $k$-spin model (Eq.~\eqref{eq:kloc_mod}) has a $\left(\delta,p_{\mathrm{est}},p_{\mathrm{b}}\right)$-efficient local shadows estimator given by the Pauli shadows algorithm~\cite{huang2020predicting}.
    
    Fix any $0<\gamma^\ast\leq 1$. The efficient local shadows estimator and model satisfy the quantum chaos property with parameters $\left(\gamma^\ast,m,\eta,R\right)$, for any choice of $m$, $\eta$, and $R$ satisfying:
    \begin{align}
        m&\geq 1+\frac{6\ln\left(6\right)}{\gamma^{\ast 2}E^{\ast 2}}9^k R,\\
        \eta&\leq\min\left(1,\left(\frac{\gamma^{\ast 2}E^{\ast 2}}{6\ln\left(2\right)9^k R}\right)^2,\frac{\gamma^{\ast 2}E^{\ast 2}}{3\ln\left(5\right)9^k R}\right).
    \end{align}
\end{restatable}

Unfortunately, we are unable to show the $m$-QOGP for the quantum $k$-spin model (though we conjecture that it does hold). As mentioned in Sec.~\ref{sec:discussion}, we believe the reason for this is that our current method overcounts states, and for this reason we are unable to get sufficiently tight bounds on the $m$-QOGP parameters in a way that they are self-consistent. However, we are able to show that a sparsified version\footnote{That is, beyond what the random sparsification of $\bm{S}$ gives us, we also sparsify by deterministically removing terms in the quantum $k$-spin model according to a set $\mathcal{P}$.} of the quantum $k$-spin model exhibits the $m$-QOGP. We consider what we call the \emph{$\left(\mathcal{P},k\right)$-quantum spin glass} model:
\begin{equation}\label{eq:pk_mod}
    \bm{H}_{\left(\mathcal{P},k\right)\mathrm{-s.g.}}:=\frac{1}{\sqrt{\left\lvert\mathcal{P}\right\rvert p\binom{n}{k}}}\sum_{\bm{b}\in\mathcal{P}}\sum_{\overline{i}\in\binom{\left[n\right]}{k}}S_{\bm{b},\overline{i}}J_{\bm{b},\overline{i}}\prod_{j=1}^k\bm{\sigma}_{i_j}^{\left(b_{i_j}\right)}.
\end{equation}
Here, $\mathcal{P}\subseteq\left\{1,2,3\right\}^{\times n}$ is a set of labels for qubit-wise-commuting Pauli operators. For this model:
\begin{align}
    r_{\mathrm{dense}}&=\left\lvert\mathcal{P}\right\rvert,\\
    d_{\mathrm{dense}}&=\binom{n-1}{k-1}.
\end{align}
We show that the $m$-QOGP holds (for any choice of $p\geq\operatorname{\Omega}\left(n^{-\left(k-1\right)}\right)$) if $\left\lvert\mathcal{P}\right\rvert$ grows subexponentially in $k$.
\begin{restatable}[The $\left(\mathcal{P},k\right)$-quantum spin glass satisfies the $m$-QOGP]{theorem}{pkmqogp}\label{thm:pk_qsg_ogp}
    The $\left(\mathcal{P},k\right)$-quantum spin glass model (Eq.~\eqref{eq:pk_mod}) has a $\left(\delta,p_{\mathrm{est}},p_{\mathrm{b}}\right)$-efficient local shadows estimator given by the derandomized shadows algorithm~\cite{PhysRevLett.127.030503}.

    Fix any $0<\gamma^\ast\leq 1$, and assume that $\mathcal{P}$ has the property that any pair $\bm{b}\neq\bm{b'}\in\mathcal{P}$ agree in at most a fraction $0\leq\phi<1$ of their entries. Fix any $0<\tilde{F}\leq 1$ and define:
    \begin{equation}
        \upsilon:=\left(\frac{1+\xi}{2}\right)\delta_{R,1}+\left(1-\tilde{F}\right)\left(1-\delta_{R,1}\right).
    \end{equation}
    This efficient local shadows estimator and model satisfy the $m$-QOGP with parameters $\left(\gamma^\ast,m,\xi,\eta,c,\eta',F,R\right)$ for any choice of $m,\xi,\eta,c,\eta',F,R$ satisfying:
    \begin{align}
        \xi&>\eta;\\
        \xi-\eta&\geq\max\left(1-\left(\frac{\gamma^{\ast 2}E^{\ast 2}}{24\ln\left(2\right)R}\right)^2,1-\frac{\gamma^{\ast 2}E^{\ast 2}}{12\ln\left(5\right)R}\right);\\
        1+\frac{8\ln\left(6\right)}{\gamma^{\ast 2}E^{\ast 2}}R\leq m&\leq 1+\frac{1}{\upsilon^k+\left\lvert\mathcal{P}\right\rvert\phi^k};\\
        c&\leq\frac{1}{48};\\
        F&\geq\tilde{F};\\
        1-\xi+\eta<&\eta'<3\max\left(\left(\frac{\gamma^{\ast 2}E^{\ast 2}}{24\ln\left(2\right)R}\right)^2,\frac{\gamma^{\ast 2}E^{\ast 2}}{12\ln\left(5\right)R}\right).
    \end{align}
\end{restatable}
The assumption that different bases $\bm{b}\neq\bm{b'}\in\mathcal{P}$ sufficiently differ in their entries is minimal given that uniformly randomly chosen $\bm{b}$ and $\bm{b'}$ are expected to differ in a $\frac{2}{3}$ fraction of their entries, with a concentration about this expectation exponential in $n$.

These results can be combined with the algorithmic hardness result (Theorem~\ref{thm:m_qogp_implies_alg_hardness}) to yield the following concrete hardness results. We begin by considering the $\left(\mathcal{P},k\right)$-quantum spin glass. In short, our result roughly states that for any sparsity parameter $p$ such that $d_{\mathrm{max}}=\operatorname{O}_k\left(1\right)$ and any choice of approximation ratio $0<\gamma\leq 1$, there exists a sufficiently large locality $k$ and system size $n$ such that algorithms with any $L=\operatorname{O}\left(1\right)$ Lipschitz constant fail to achieve it.
\begin{corollary}[Stable algorithms fail for the $\left(\mathcal{P},k\right)$-quantum spin glass]\label{cor:stab_fail_pk_qsg}
    Consider the class $\mathcal{H}$ of $\left(\mathcal{P},k\right)$-quantum spin glass Hamiltonians, where $\mathcal{P}$ has the property that any pair $\bm{b}\neq\bm{b'}\in\mathcal{P}$ agree in at most a $k$-independent fraction $0\leq\phi<1$ of their entries. Assume that $\left\lvert\mathcal{P}\right\rvert$ is subexponential in $k$. Finally, let $d_{\mathrm{max}}$ be as described in Theorem~\ref{thm:m_qogp_implies_alg_hardness}, with $p$ taken such that $d_{\mathrm{max}}=\operatorname{O}_k\left(1\right)$.

    For any choice of $0<\gamma\leq 1$ and $0<\epsilon<1$, for sufficiently large $n$ and $k$, there is no $\left(f,L,\mathfrak{d},\kappa,p_{\mathrm{st}}\right)$-stable and $\left(\gamma,p_{\mathrm{f}}\right)$-near optimal algorithm for $\mathcal{H}$ where:
    \begin{align}
        \mathfrak{d}&\geq d_{\mathrm{max}};\\
        \kappa&\leq\max\left(0,1-1.001\left\lceil\frac{2d_{\mathrm{max}}^2}{\epsilon^2}\right\rceil^{-1}\right);\\
        p_{\mathrm{st}}+p_{\mathrm{f}}&\leq\operatorname{o}\left(1\right);\\
        \frac{f}{n}+\epsilon L&\leq\operatorname{O}_k\left(\frac{E^{\ast 4}}{\left\lvert\mathcal{P}\right\rvert^2}\right).
    \end{align}
\end{corollary}
\begin{proof}
    By Proposition~\ref{prop:derand_pauli_shadows_params}, for any $\delta>0$ the derandomized shadows algorithm~\cite{PhysRevLett.127.030503} is a $\left(\delta,p_{\mathrm{est}},p_{\mathrm{b}}\right)$-efficient local shadows estimator with $p_{\mathrm{b}}=\exp\left(-\operatorname{\Omega}\left(n\right)\right)$ and
    \begin{equation}
        p_{\mathrm{est}}=\frac{1}{1+0.99\left\lvert\mathcal{P}\right\rvert^{-1}\delta^2}.
    \end{equation}
    Fix any $\delta<\gamma$ and take $\gamma^\ast=\gamma-\delta>0$. We also take:
    \begin{align}
        Q&=\left\lceil\frac{2d_{\mathrm{max}}^2}{\epsilon^2}\right\rceil,\\
        \beta&=\sqrt{2Q},\\
        F&=\frac{1}{Q},
    \end{align}
    and
    \begin{equation}
        R=\left\lceil\frac{\ln\left(4Q\right)}{\ln\left(p_{\mathrm{est}}^{-1}\right)}\right\rceil=\operatorname{\Omega}_k\left(\ln\left(Q\right)\left\lvert\mathcal{P}\right\rvert\right).
    \end{equation}
    This choice of $\beta$ and $R$ is such that:
    \begin{equation}
        \frac{Q}{\beta^2}+Qp_{\mathrm{est}}^R\leq\frac{1}{2}+\frac{1}{4}<1.
    \end{equation}
    As $\upsilon$ is bounded away from $1$ by $F=\frac{1}{Q}=\operatorname{\Theta}_k\left(1\right)$, and as $\left\lvert\mathcal{P}\right\rvert$---and therefore also $R$---is assumed to be subexponential in $k$, it is therefore the case that there exists $m\in\mathbb{N}$ such that:
    \begin{equation}
        1+\frac{8\ln\left(6\right)}{\gamma^{\ast 2}E^{\ast 2}}R\leq m\leq 1+\frac{1}{\upsilon^k+\left\lvert\mathcal{P}\right\rvert\phi^k}
    \end{equation}
    for sufficiently large $k$. Finally, take:
    \begin{equation}
        \begin{aligned}
            \xi&=\max\left(1-\frac{1}{2}\left(\frac{\gamma^{\ast 2}E^{\ast 2}}{24\ln\left(2\right)R}\right)^2,1-\frac{1}{2}\frac{\gamma^{\ast 2}E^{\ast 2}}{12\ln\left(5\right)R}\right),\\
            \eta&=\min\left(\frac{1}{2}\left(\frac{\gamma^{\ast 2}E^{\ast 2}}{24\ln\left(2\right)R}\right)^2,\frac{1}{2}\left(\frac{\gamma^{\ast 2}E^{\ast 2}}{12\ln\left(5\right)R}\right)\right).
        \end{aligned}
    \end{equation}
    The final parameters follow by the given choices of $R$ and $Q$. Note that the required bound on $p_{\mathrm{st}}+p_{\mathrm{f}}$ is technically a constant depending only on $k$ and $\epsilon$, though it is triply exponentially small in $k$ as it is doubly exponentially small in $m$.
\end{proof}
If $\left\lvert\mathcal{P}\right\rvert$ is further assumed to only grow as $k^{0.124}$, $L$ can even grow with $k$.
\begin{corollary}[Stable algorithms fail for the $\left(\mathcal{P},k\right)$-quantum spin glass, $\left\lvert\mathcal{P}\right\rvert\leq\operatorname{O}_k\left(k^{0.124}\right)$]\label{cor:stab_algs_fail_sparse_q_spin_glass}
    Consider the class $\mathcal{H}$ of $\left(\mathcal{P},k\right)$-quantum spin glass Hamiltonians, where $\mathcal{P}$ has the property that any pair $\bm{b}\neq\bm{b'}\in\mathcal{P}$ agree in at most a $k$-independent fraction $0\leq\phi<1$ of their entries. Assume that $\left\lvert\mathcal{P}\right\rvert\leq\operatorname{O}_k\left(k^{0.124}\right)$. Finally, let $d_{\mathrm{max}}$ be as described in Theorem~\ref{thm:m_qogp_implies_alg_hardness}, with $p$ taken such that $d_{\mathrm{max}}=\operatorname{O}_k\left(1\right)$.

    For any choice of $0<\gamma\leq 1$ and $0<\epsilon<1$, for sufficiently large $n$ and $k$, there is no $\left(f,L,\mathfrak{d},\kappa,p_{\mathrm{st}}\right)$-stable and $\left(\gamma,p_{\mathrm{f}}\right)$-near optimal algorithm for $\mathcal{H}$ where:
    \begin{align}
        \mathfrak{d}&\geq d_{\mathrm{max}};\\
        \kappa&\leq 1-1.001\left\lceil\frac{2d_{\mathrm{max}}^2 k^{0.999}\ln\left(k\right)^2}{\epsilon^2}\right\rceil^{-1};\\
        p_{\mathrm{st}}+p_{\mathrm{f}}&\leq\operatorname{o}\left(1\right);\\
        f&=\operatorname{o}\left(n\right);\\
        L&\leq\epsilon^{-1}\ln\left(k\right).
    \end{align}
\end{corollary}
\begin{proof}
    In a similar fashion to the proof of Corollary~\ref{cor:stab_fail_pk_qsg}, we take:
    \begin{align}
        Q&=\left\lceil\frac{2d_{\mathrm{max}}^2 k^{0.999}\ln\left(k\right)^2}{\epsilon^2}\right\rceil,\\
        \beta&=\sqrt{2Q},\\
        F&=\frac{1}{Q},
    \end{align}
    and
    \begin{equation}
        R=\left\lceil\frac{\ln\left(4Q\right)}{\ln\left(p_{\mathrm{est}}^{-1}\right)}\right\rceil=\operatorname{\Omega}_k\left(\ln\left(Q\right)\left\lvert\mathcal{P}\right\rvert\right).
    \end{equation}
    This choice of $\beta$ and $R$ is such that:
    \begin{equation}
        \frac{Q}{\beta^2}+Qp_{\mathrm{est}}^R\leq\frac{1}{2}+\frac{1}{4}<1.
    \end{equation}
    As $\upsilon$ is bounded away from $1$ by $F=\frac{1}{Q}=\operatorname{\Theta}_k\left(\frac{1}{k^{0.999}\ln\left(k\right)^2}\right)$, and as $\left\lvert\mathcal{P}\right\rvert$---and therefore also $R$---is assumed to be subexponential in $k$, it is therefore the case that there exists $m\in\mathbb{N}$ such that:
    \begin{equation}
        1+\frac{8\ln\left(6\right)}{\gamma^{\ast 2}E^{\ast 2}}R\leq m\leq 1+\frac{1}{\upsilon^k+\left\lvert\mathcal{P}\right\rvert\phi^k}
    \end{equation}
    for sufficiently large $k$. Finally, take:
    \begin{equation}
        \begin{aligned}
            \xi&=\max\left(1-\frac{1}{2}\left(\frac{\gamma^{\ast 2}E^{\ast 2}}{24\ln\left(2\right)R}\right)^2,1-\frac{1}{2}\frac{\gamma^{\ast 2}E^{\ast 2}}{12\ln\left(5\right)R}\right),\\
            \eta&=\min\left(\frac{1}{2}\left(\frac{\gamma^{\ast 2}E^{\ast 2}}{24\ln\left(2\right)R}\right)^2,\frac{1}{2}\left(\frac{\gamma^{\ast 2}E^{\ast 2}}{12\ln\left(5\right)R}\right)\right).
        \end{aligned}
    \end{equation}
    Note that $E^\ast=\operatorname{\Omega}_k\left(\frac{1}{\sqrt{\left\lvert\mathcal{P}\right\rvert}}\right)$~\cite{anschuetz2024boundsgroundstateenergy}. In particular, for sufficiently large $k$,
    \begin{equation}
        \sqrt{Q}\eta\geq\operatorname{\Omega}_k\left(\frac{\sqrt{Q}E^{\ast 4}}{\log\left(Q\right)^2\left\lvert\mathcal{P}\right\rvert^2}\right)\geq\operatorname{\tilde{\Omega}}\left(\frac{k^{\frac{0.99}{2}}}{k^{4\times 0.124}}\right)\geq 1.
    \end{equation}
    The assumed choice of $f$ and the required stability (Eq.~\eqref{eq:stab_req}) give the final result. Note that the required bound on $p_{\mathrm{st}}+p_{\mathrm{f}}$ is technically a constant depending only on $k$ and $\epsilon$, though it is triply exponentially small in $k$ as it is doubly exponentially small in $m$.
\end{proof}
Finally, a stricter class of algorithms fail for the traditional quantum $k$-spin model. As mentioned previously, here we require $L=\operatorname{o}\left(1\right)$ as we are only able to show that the quantum $k$-spin model satisfies the quantum chaos property, not the full $m$-QOGP.
\begin{corollary}[Very stable algorithms fail for the quantum $k$-spin model]
    Consider the class $\mathcal{H}$ of quantum $k$-spin model Hamiltonians. Let $d_{\mathrm{max}}$ be as described in Theorem~\ref{thm:m_qogp_implies_alg_hardness}, with $p$ taken such that $d_{\mathrm{max}}=\operatorname{O}_k\left(1\right)$.

    For any choice of $0<\gamma\leq 1$ and $0<\epsilon<1$, for sufficiently large $n$ and $k$, there is no $\left(f,L,\mathfrak{d},0,p_{\mathrm{st}}\right)$-stable and $\left(\gamma,p_{\mathrm{f}}\right)$-near optimal algorithm for $\mathcal{H}$ where:
    \begin{align}
        \mathfrak{d}&\geq d_{\mathrm{max}};\\
        p_{\mathrm{st}}+p_{\mathrm{f}}&<\frac{1}{6};\\
        \frac{f}{n}+\epsilon L&\leq\operatorname{o}\left(1\right).
    \end{align}
\end{corollary}
\begin{proof}
    The proof proceeds identically to that of Corollary~\ref{cor:stab_fail_pk_qsg} with the choice $Q=1$.
\end{proof}

In Appendix~\ref{sec:stable_q_algs}, we demonstrate that a wide variety of algorithms are stable. For example, we show that any quantum algorithm based on Trotterized Hamiltonian evolution for time\footnote{In units where the associated Hamiltonian operator norm is $n$-independent.} $t$ are $\left(\sqrt{n},L,\mathfrak{d},0,0\right)$ stable for any $\mathfrak{d}$, where (up to universal constants):
\begin{equation}
    L=\frac{1}{\sqrt{n}}\left(\frac{3}{2}k\mathfrak{d}\right)^t.
\end{equation}
When $t\leq\frac{\ln\left(\frac{E^{\ast 4}\sqrt{n}}{\left\lvert\mathcal{P}\right\rvert^2}\right)}{\ln\left(\frac{3}{2}k\mathfrak{d}\right)}$, this is:
\begin{equation}\label{eq:lip_const_bound}
    L\leq\frac{E^{\ast 4}}{\left\lvert\mathcal{P}\right\rvert^2}.
\end{equation}
This class of algorithms includes the quantum approximate optimization algorithm (QAOA) of depth $t$~\cite{farhi2014quantumapproximate}, the variational quantum eigensolver (VQE)~\cite{peruzzo2014variational} with Hamiltonian variational ansatz of depth $t$~\cite{PhysRevA.92.042303}, and quantum annealing~\cite{farhi2000quantum} or quantum Lindbladian dynamics for time $t$. Eq.~\eqref{eq:lip_const_bound} culminates in the following result.
\begin{corollary}[$\log\left(n\right)$-depth Trotterized algorithms fail to optimize quantum spin glass models, informal]\label{cor:trott_algs_fail}
    For any $p=\operatorname{\Theta}\left(n^{-\left(k-1\right)}\right)$ and $\left\lvert\mathcal{P}\right\rvert$ subexponential in $k$, there exists a universal constant $C>0$ and sufficiently large $n$ then $k$ such that for the $\left(\mathcal{P},k\right)$ quantum spin glass model, quantum algorithms based on Trotterized Hamiltonian evolution for time $t\leq C\log\left(n\right)$ fail to achieve a constant approximation ratio. The same is true with $t\leq 0.999C\log\left(n\right)$ for the quantum $k$-spin model.
\end{corollary}
These hardness results echo what is rigorously known for classical algorithms and the classical $p$-spin models: for any $0<\gamma\leq 1$, there exists a constant $C$ such that $C\log\left(n\right)$-time Langevin dynamics~\cite{doi:10.1137/22M150263X} and $C\frac{\log\left(n\right)}{\log\log\left(n\right)}$-depth Boolean circuits~\cite[Theorem~2.3]{gamarnik2022circuitlowerboundspspin} fail to reach the near-ground state. Further work on both the classical and quantum sides is needed to rigorously rule out higher-depth algorithms, though---just as is conjectured classically---we suspect that Linbladian dynamics must be run for time exponential in the system size to approximately optimize the spin glass models we consider here.

\section{The \texorpdfstring{$m$}{m}-Quantum Overlap Gap Property Implies Algorithmic Hardness}\label{sec:proof_of_alg_hardness}

In this section we prove Theorem~\ref{thm:m_qogp_implies_alg_hardness}.

We will achieve this via a proof by contradiction. We assume that there exists some stable, near-optimal quantum algorithm $\bm{\mathcal{A}}$ for the given class of random Hamiltonians $\mathcal{H}$, and also fix the $\left(\delta,p_{\mathrm{est}},p_{\mathrm{b}}\right)$-efficient local shadows estimator for $\mathcal{H}$ such that they satisfy the $m$-QOGP (or in the weaker case, the quantum chaos property). We then proceed according to the following outline:
\begin{enumerate}
    \item We cite a known result lower bounding the probability that the sparsity $\bm{S}$ induces an interaction hypergraph with bounded maximal degree.
    \item We show that the existence of $\bm{\mathcal{A}}$ implies the existence of a stable, near-optimal quantum algorithm which is also deterministic.
    \item We show that the existence of a stable, near-optimal, deterministic quantum algorithm and an efficient local shadows estimator for $\mathcal{H}$ implies the existence of a ``near-stable,'' near-optimal quantum algorithm $\bm{\mathcal{I}}$ with codomain the product state space of classical shadow representations.
    \item We consider the interpolation path of Eq.~\eqref{eq:interp_path} over many replicas, and show that $\bm{\mathcal{I}}$ is near-stable and near-optimal over the replicas with high probability.
    \item We strengthen the ``near-stability'' of $\bm{\mathcal{I}}$ to a statement that, with high probability, pairwise distances between the outputs of $\bm{\mathcal{I}}$ are stable along the interpolation paths.
    \item Due to $\mathcal{H}$ satisfying the quantum chaos property by assumption, we show that with high probability all $m$-tuples of $T$ independently-sampled instances have near-optimal states which are distant in Wasserstein distance with high probability.
    \item We show that with high probability there exists some point along the interpolation path where this algorithm outputs a configuration disallowed by the $m$-QOGP due to the pairwise-stability and near-optimality of $\mathcal{I}$.
    \item Finally, we show that there exist choices of parameters such that all ``with high probability'' events have a nontrivial intersection. This contradicts the assumption of the existence of $\bm{\mathcal{A}}$.
\end{enumerate}
Each step of the proof strategy is given its own subsection in what follows for organization.

\subsection{Probability of a Constant-Degree Interaction Hypergraph}

We here cite the probability that $\bm{S}$ induces an interaction hypergraph of degree at most a constant, specialized to the interaction hypergraph of the efficient local shadows estimator $\bm{\mathcal{R}}\left(\bm{H}_{\bm{X}}\right)$ we are considering (as defined in Eq.~\eqref{eq:lin_map_shadows}). The general theorem is due to \revref\cite{RIORDAN_SELBY_2000}.
\begin{proposition}[Maximum degree of interaction hypergraph, adaptation of \revref{\cite[Theorem~2.1]{RIORDAN_SELBY_2000}}]\label{prop:max_deg_rand_hyp}
    Given the linear map $\bm{\mathcal{R}}$ associated with the efficient local shadows estimator for $\mathcal{H}$ (see Eq.~\eqref{eq:lin_map_shadows}), let
    \begin{equation}
        r_{\mathrm{dense}}:=\max_{\overline{i}\in 2^{\left[n\right]}}\sum_{j=1}^D\bm{1}\left\{\overline{i}=\operatorname{supp}\left(\bm{R}_j\right)\right\}
    \end{equation}
    be the maximum number of terms in $\bm{\mathcal{R}}\left(\bm{H}_{\bm{X}}\right)$ with support a given hyperedge in its interaction hypergraph, and $d_{\mathrm{dense}}$ the maximum degree of the interaction hypergraph of $\bm{\mathcal{R}}\left(\bm{H}_{\bm{X}}\right)$ when the sparsity $p=1$.

    For every $\epsilon>0$, there exists a constant $b_\epsilon>0$ depending only on $\epsilon$ such that:
    \begin{equation}
        \mathbb{P}_G\left[\mathcal{X}\right]\geq\exp\left(-\epsilon n\right),
    \end{equation}
    where $\mathcal{X}$ is the event that the maximum degree $d_{\mathrm{max}}$ of the interaction hypergraph $G$ is bounded:
    \begin{equation}\label{eq:x_event_def}
        \mathcal{X}=\left\{d_{\mathrm{max}}\left(G\right)\leq r_{\mathrm{dense}}pd_{\mathrm{dense}}+b_\epsilon\sqrt{r_{\mathrm{dense}}pd_{\mathrm{dense}}\left(1-p\right)}\right\}.
    \end{equation}
\end{proposition}
\begin{proof}
    Note that the probability a hyperedge is included in the interaction hypergraph $p_{\mathrm{h.e.}}$ is given by:
    \begin{equation}
        p_{\mathrm{h.e.}}\leq r_{\mathrm{dense}}p
    \end{equation}
    by the union bound. Therefore, as
    \begin{equation}
        \begin{aligned}
            \mathcal{X}&=\left\{d_{\mathrm{max}}\left(G\right)\leq r_{\mathrm{dense}}pd_{\mathrm{dense}}+b_\epsilon\sqrt{r_{\mathrm{dense}}pd_{\mathrm{dense}}\left(1-p\right)}\right\}\\
            &\supseteq\left\{d_{\mathrm{max}}\left(G\right)\leq p_{\mathrm{h.e.}}d_{\mathrm{dense}}+b_\epsilon\sqrt{p_{\mathrm{h.e.}}d_{\mathrm{dense}}\left(1-p\right)}\right\},
        \end{aligned}
    \end{equation}
    the probability bound immediately follows from \revref\cite[Theorem~2.1]{RIORDAN_SELBY_2000} and noting that the probability is bounded by the $k=2$ case (see, e.g., \revref\cite[Eq.~(10)]{zhang2023higher}).
\end{proof}
In what follows, we will let $\epsilon$ be arbitrary, and only fix it at the very end.

\subsection{Reduction to Deterministic Quantum Algorithms}

We now prove that, WLOG, one can consider deterministic quantum algorithms. The proof follows a similar strategy as \revref\cite[Lemma~6.11]{gamarnik2022algorithmsbarrierssymmetricbinary}, though is slightly more involved due to requiring the reduction work for many $\kappa'$ as we will need to use the same algorithm over our entire interpolation path.
\begin{lemma}[Reduction to deterministic quantum algorithms]\label{lem:det_rand_alg_red}
    Let $\bm{\mathcal{A}}\left(\bm{X},\omega\right)$ be a quantum algorithm that is both $\left(f,L,\mathfrak{d},\kappa,p_{\mathrm{st}}\right)$-stable and $\left(\gamma,p_{\mathrm{f}}\right)$-optimal for the class of random Hamiltonians $\mathcal{H}=\left\{\bm{H}_{\bm{X}}\right\}_{\bm{X}\in\mathbb{R}^D}$. Let $\mathcal{K}\subseteq\left[\kappa,1\right]$ be a multiset of cardinality $Q$. Then, there exists a deterministic quantum algorithm $\bm{\tilde{\mathcal{A}}}\left(\bm{X}\right)$ that is both $\left(f,L,\mathfrak{d},\mathcal{K},3Qp_{\mathrm{st}}\right)$-stable and $\left(\gamma,3p_{\mathrm{f}}\right)$-optimal for $\mathcal{H}$.
\end{lemma}
\begin{proof}
    Let $\left(\varOmega,\mathbb{P}_\varOmega\right)$ be the probability space associated with $\bm{\mathcal{A}}$, let $E^\ast$ be the limiting maximal energy of $\mathcal{H}$, and let $\mathbb{P}_{\bm{X},\bm{Y}}^{\mathfrak{d},\kappa'}$ be as defined in Definition~\ref{def:stable_qas}. For notational convenience, we define the events for all $\omega\in\varOmega$ and $\bm{X},\bm{Y}\in\mathbb{R}^D$:
    \begin{align}
        \mathcal{E}_{\mathrm{st}}^{\left(\omega\right)}\left(\bm{X},\bm{Y}\right)&:=\left\{\left\lVert\bm{\mathcal{A}}\left(\bm{X},\omega\right)-\bm{\mathcal{A}}\left(\bm{Y},\omega\right)\right\rVert_{W_2}\leq f+L\left\lVert\bm{X}-\bm{Y}\right\rVert_1\right\},\\
        \mathcal{E}_{\mathrm{no}}^{\left(\omega\right)}\left(\bm{X}\right)&:=\left\{\Tr\left(\bm{H}_{\bm{X}}\bm{\mathcal{A}}\left(\bm{X},\omega\right)\right)\geq\gamma E^\ast\sqrt{n}\right\},
    \end{align}
    and the events for all $\omega\in\varOmega$ and $\kappa'\in\mathcal{K}$:
    \begin{align}
        \mathcal{F}_{\mathrm{st}}^{\left(\kappa'\right)}\left(\omega\right)&:=\left\{\mathbb{P}_{\left(\bm{X},\bm{Y}\right)\sim\mathbb{P}_{\bm{X},\bm{Y}}^{\mathfrak{d},\kappa'}}\left[\mathcal{E}_{\mathrm{st}}^{\left(\omega\right)}\left(\bm{X},\bm{Y}\right)^\complement\right]>3Qp_{\mathrm{st}}\right\}=\left\{\mathbb{P}_{\left(\bm{X},\bm{Y}\right)\sim\mathbb{P}_{\bm{X},\bm{Y}}^{\mathfrak{d},\kappa'}}\left[\mathcal{E}_{\mathrm{st}}^{\left(\omega\right)}\left(\bm{X},\bm{Y}\right)\right]\leq 1-3Qp_{\mathrm{st}}\right\},\\
        \mathcal{F}_{\mathrm{no}}\left(\omega\right)&:=\left\{\mathbb{P}_{\bm{X}}\left[\mathcal{E}_{\mathrm{no}}^{\left(\omega\right)}\left(\bm{X}\right)^\complement\right]>3p_{\mathrm{f}}\right\}=\left\{\mathbb{P}_{\bm{X}}\left[\mathcal{E}_{\mathrm{no}}^{\left(\omega\right)}\left(\bm{X}\right)\right]\leq 1-3p_{\mathrm{f}}\right\}.
    \end{align}
    By the law of total probability, for all $\kappa'\in\left[\kappa,1\right]$,
    \begin{align}
        \mathbb{E}_{\omega\sim\mathbb{P}_\varOmega}\left[\mathbb{P}_{\left(\bm{X},\bm{Y}\right)\sim\mathbb{P}_{\bm{X},\bm{Y}}^{\mathfrak{d},\kappa'}}\left[\mathcal{E}_{\mathrm{st}}^{\left(\omega\right)}\left(\bm{X},\bm{Y}\right)^\complement\right]\right]&=\mathbb{P}_{\left(\bm{X},\bm{Y},\omega\right)\sim\mathbb{P}_{\bm{X},\bm{Y}}^{\mathfrak{d},\kappa'}\otimes\mathbb{P}_\varOmega}\left[\mathcal{E}_{\mathrm{st}}^{\left(\omega\right)}\left(\bm{X},\bm{Y}\right)^\complement\right]\leq p_{\mathrm{st}},\\
        \mathbb{E}_{\omega\sim\mathbb{P}_\varOmega}\left[\mathbb{P}_{\bm{X}}\left[\mathcal{E}_{\mathrm{no}}^{\left(\omega\right)}\left(\bm{X}\right)^\complement\right]\right]&=\mathbb{P}_{\left(\bm{X},\omega\right)}\left[\mathcal{E}_{\mathrm{no}}^{\left(\omega\right)}\left(\bm{X}\right)^\complement\right]\leq p_{\mathrm{f}},
    \end{align}
    where the inequalities follow from the stability and near-optimality of $\bm{\mathcal{A}}$.
    By Markov's inequality, for all $\kappa'\in\mathcal{K}$,
    \begin{align}
        \mathbb{P}_{\omega\sim\mathbb{P}_\varOmega}\left[\mathcal{F}_{\mathrm{st}}^{\left(\kappa'\right)}\left(\omega\right)\right]&\leq\frac{\mathbb{E}_{\omega\sim\mathbb{P}_\varOmega}\left[\mathbb{P}_{\left(\bm{X},\bm{Y}\right)\sim\mathbb{P}_{\bm{X},\bm{Y}}^{\mathfrak{d},\kappa'}}\left[\mathcal{E}_{\mathrm{st}}^{\left(\omega\right)}\left(\bm{X},\bm{Y}\right)^\complement\right]\right]}{3Qp_{\mathrm{st}}}\leq\frac{1}{3Q},\\
        \mathbb{P}_{\omega\sim\mathbb{P}_\varOmega}\left[\mathcal{F}_{\mathrm{no}}\left(\omega\right)\right]&\leq\frac{\mathbb{E}_{\omega\sim\mathbb{P}_\varOmega}\left[\mathbb{P}_{\bm{X}}\left[\mathcal{E}_{\mathrm{no}}^{\left(\omega\right)}\left(\bm{X}\right)^\complement\right]\right]}{3p_{\mathrm{f}}}\leq\frac{1}{3}.
    \end{align}
    Furthermore, by the union bound,
    \begin{equation}
        \mathbb{P}_{\omega\sim\mathbb{P}_\varOmega}\left[\mathcal{F}_{\mathrm{no}}\left(\omega\right)\cup\bigcup_{\kappa'\in\mathcal{K}}\mathcal{F}_{\mathrm{st}}^{\left(\kappa'\right)}\left(\omega\right)\right]\leq\frac{2}{3}.
    \end{equation}
    In particular, there must exist $\omega^\ast\in\varOmega$ such that the event
    \begin{equation}
        \left(\mathcal{F}_{\mathrm{no}}\left(\omega^\ast\right)\cup\bigcup_{\kappa'\in\mathcal{K}}\mathcal{F}_{\mathrm{st}}^{\left(\kappa'\right)}\left(\omega^\ast\right)\right)^\complement=\mathcal{F}_{\mathrm{no}}\left(\omega^\ast\right)^\complement\cap\bigcap_{\kappa'\in\mathcal{K}}\mathcal{F}_{\mathrm{st}}^{\left(\kappa'\right)}\left(\omega^\ast\right)^\complement
    \end{equation}
    occurs. By definition, then,
    \begin{equation}
        \bm{\tilde{\mathcal{A}}}\left(\bm{X}\right):=\bm{\mathcal{A}}\left(\bm{X},\omega^\ast\right)
    \end{equation}
    is a deterministic quantum algorithm that is both $\left(f,L,\mathfrak{d},\mathcal{K},3Qp_{\mathrm{st}}\right)$-stable and $\left(\gamma,3p_{\mathrm{f}}\right)$-optimal for $\mathcal{H}$.
\end{proof}

\subsection{Efficient Local Shadows Reduction}

We now prove that, if there exists a stable, near-optimal quantum algorithm for a problem class with an efficient local shadows estimator, there exists a near-optimal quantum algorithm with codomain the space of classical shadow representations $\mathcal{B}$ which satisfies a weaker notion of stability than that given in Definition~\ref{def:stable_qas}. In what follows, recall that we use $\mathcal{U}$ to denote the uniform distribution over $\left[0,1\right]$.
\begin{lemma}[Nondeterministic classical shadows reduction]\label{lem:class_shads_red}
    Consider a class of random Hamiltonians $\mathcal{H}=\left\{\bm{H}_{\bm{X}}\right\}_{\bm{X}}$ with $\left(\delta,p_{\mathrm{est}},p_{\mathrm{b}}\right)$-efficient local shadows estimator with $\bm{\mathcal{M}}$, $\bm{\tilde{\mathcal{M}}}$, $\mathcal{B}$, $\bm{\mathcal{R}}_{\bm{X}}$, and $\mathcal{V}$ as in Definition~\ref{def:eff_loc_shad_est}. Assume there exists an $\left(f,L,\mathfrak{d},\mathcal{K},3Qp_{\mathrm{st}}\right)$-stable and $\left(\gamma,3p_{\mathrm{f}}\right)$-near optimal deterministic quantum algorithm $\bm{\mathcal{A}}\left(\bm{X}\right)$ for $\mathcal{H}$. Then, there exists a pure quantum algorithm $\bm{\mathcal{G}}:\mathbb{R}^D\times\left[0,1\right]\to\mathcal{B}$ such that, for all $\kappa'\in\mathcal{K}$,
    \begin{equation}\label{eq:g_close_to_stab}
        \mathbb{P}_{\left(\bm{X},\bm{Y}\right)\sim\mathbb{P}_{\bm{X},\bm{Y}}^{\mathfrak{d},\kappa'}}\left[\left\lVert\mathbb{E}_{\omega\sim\mathcal{U}}\left[\bm{\mathcal{G}}\left(\bm{X},\omega\right)-\bm{\mathcal{G}}\left(\bm{Y},\omega\right)\right]\right\rVert_{W_2}\leq f+L\left\lVert\bm{X}-\bm{Y}\right\rVert_1\right]\geq 1-3Qp_{\mathrm{st}}.
    \end{equation}
    Furthermore, for any $\bm{X}$ satisfying the probability $1-3p_{\mathrm{f}}-p_{\mathrm{b}}$ event $\left\{\Tr\left(\bm{H}_{\bm{X}}\bm{\mathcal{A}}\left(\bm{X}\right)\right)\geq\gamma E^\ast\sqrt{n}\right\}\cap\mathcal{V}$, $\bm{\mathcal{G}}\left(\bm{X},\omega\right)$ satisfies:
    \begin{equation}\label{eq:g_near_opt}
        \mathbb{P}_{\omega\sim\mathcal{U}}\left[\Tr\left(\bm{\mathcal{R}}_{\bm{X}}\bm{\mathcal{G}}\left(\bm{X},\omega\right)\right)\geq\left(\gamma-\delta\right)E^\ast\sqrt{n}\right]\geq 1-p_{\mathrm{est}}.
    \end{equation}
\end{lemma}
\begin{proof}
    We separately consider Eq.~\eqref{eq:g_close_to_stab} and Eq.~\eqref{eq:g_near_opt}.

    \textbf{Eq.~\eqref{eq:g_close_to_stab}.} Recall the channel $\bm{\mathcal{M}}$ associated with the efficient local shadows estimator. As the quantum Wasserstein distance is nonincreasing under convex combinations of tensor product channels (see Proposition~\ref{prop:contractive_tensor_prod}), we have for all $\bm{\rho},\bm{\sigma}\in\mathcal{S}_n^{\mathrm{m}}$:
    \begin{equation}
        \left\lVert\mathbb{E}_{\omega\sim\mathcal{U}}\left[\bm{\tilde{\mathcal{M}}}\left(\bm{\rho},\omega\right)-\bm{\tilde{\mathcal{M}}}\left(\bm{\sigma},\omega\right)\right]\right\rVert_{W_2}=\left\lVert\bm{\mathcal{M}}\left(\bm{\rho}\right)-\bm{\mathcal{M}}\left(\bm{\sigma}\right)\right\rVert_{W_2}\leq\left\lVert\bm{\rho}-\bm{\sigma}\right\rVert_{W_2}.
    \end{equation}
    In particular, for any $\mathbb{P}_{\bm{X},\bm{Y}}^{\mathfrak{d},\kappa'}$ as in Definition~\ref{def:stable_qas}, and defining:
    \begin{equation}\label{eq:g_def}
        \bm{\mathcal{G}}\left(\bm{X},\omega\right):=\bm{\tilde{\mathcal{M}}}\left(\bm{\mathcal{A}}\left(\bm{X}\right),\omega\right),
    \end{equation}
    we have:
    \begin{equation}\label{eq:quant_stab_red}
        \mathbb{P}_{\left(\bm{X},\bm{Y}\right)\sim\mathbb{P}_{\bm{X},\bm{Y}}^{\mathfrak{d},\kappa'}}\left[\left\lVert\mathbb{E}_{\omega\sim\mathcal{U}}\left[\bm{\mathcal{G}}\left(\bm{X},\omega\right)-\bm{\mathcal{G}}\left(\bm{Y},\omega\right)\right]\right\rVert_{W_2}\leq f+L\left\lVert\bm{X}-\bm{Y}\right\rVert_1\right]\geq 1-3Qp_{\mathrm{st}}
    \end{equation}
    for all $\kappa'\in\mathcal{K}$ since $\bm{\mathcal{A}}\left(\bm{X}\right)$ is $\left(f,L,\mathfrak{d},\mathcal{K},3Qp_{\mathrm{st}}\right)$-stable.

    \textbf{Eq.~\eqref{eq:g_near_opt}.} Recall that, by the definition of $\left(\delta,p_{\mathrm{est}},p_{\mathrm{b}}\right)$-efficient local shadows estimators (Definition~\ref{def:eff_loc_shad_est}), with probability at least $1-p_{\mathrm{b}}$ over $\bm{X}\sim\mathcal{N}\left(0,1\right)^{\otimes D}$,
    \begin{equation}
        \mathbb{P}_{\omega\sim\mathcal{U}}\left[\Tr\left(\bm{\mathcal{R}}_{\bm{X}}\bm{\tilde{\mathcal{M}}}\left(\bm{\rho},\omega\right)\right)-\Tr\left(\bm{H}_{\bm{X}}\bm{\rho}\right)\geq-\delta E^\ast\sqrt{n}\right]\geq 1-p_{\mathrm{est}}
    \end{equation}
    for all $\bm{\rho}\in\mathcal{S}_n^{\mathrm{m}}$. In particular,
    \begin{equation}\label{eq:est_no}
        \mathbb{P}_{\omega\sim\mathcal{U}}\left[\Tr\left(\bm{\mathcal{R}}_{\bm{X}}\bm{\tilde{\mathcal{M}}}\left(\bm{\mathcal{A}}\left(\bm{X}\right),\omega\right)\right)-\Tr\left(\bm{H}_{\bm{X}}\bm{\mathcal{A}}\left(\bm{X}\right)\right)<-\delta E^\ast\sqrt{n}\right]\leq p_{\mathrm{est}}.
    \end{equation}
    Furthermore, by the definition of near-optimality (Definition~\ref{def:no_qas}),
    \begin{equation}
        \mathbb{P}_{\bm{X}}\left[\Tr\left(\bm{H}_{\bm{X}}\bm{\mathcal{A}}\left(\bm{X}\right)\right)\geq\gamma E^\ast\sqrt{n}\right]\geq 1-3p_{\mathrm{f}}.
    \end{equation}
    For any $\bm{X}$ such that this event and the probability $1-p_{\mathrm{b}}$ event $\mathcal{V}$ hold, we have by Eq.~\eqref{eq:est_no} that:
    \begin{equation}
        \mathbb{P}_{\omega\sim\mathcal{U}}\left[\Tr\left(\bm{\mathcal{R}}_{\bm{X}}\bm{\tilde{\mathcal{M}}}\left(\bm{\mathcal{A}}\left(\bm{X}\right),\omega\right)\right)\geq\left(\gamma-\delta\right)E^\ast\sqrt{n}\right]\geq 1-p_{\mathrm{est}}.
    \end{equation}
\end{proof}

Ideally, we would strengthen this notion of stability to that of Definition~\ref{def:stable_qas}, which would allow us to directly leverage the machinery of classical algorithmic hardness results based on OGPs. Unfortunately, as mentioned in Sec.~\ref{sec:stab_qas}, it is generally the case that:
\begin{equation}
    \left\lVert\mathbb{E}_{\omega\sim\mathcal{U}}\left[\bm{\mathcal{G}}\left(\bm{X},\omega\right)-\bm{\mathcal{G}}\left(\bm{Y},\omega\right)\right]\right\rVert_{W_2}\neq\mathbb{E}_{\omega\sim\mathcal{U}}\left[\left\lVert\bm{\mathcal{G}}\left(\bm{X},\omega\right)-\bm{\mathcal{G}}\left(\bm{Y},\omega\right)\right\rVert_{W_2}\right],
\end{equation}
making such a strengthening generally impossible. However, we can get close; as the quantum Wasserstein distance upper bounds the classical Wasserstein distance for quantum states mutually diagonal in a product state basis (see Proposition~\ref{prop:quant_wass_prod_states}), we have the following consequence of stability according to $\left\lVert\mathbb{E}_{\omega\sim\mathcal{U}}\left[\bm{\mathcal{G}}\left(\bm{X},\omega\right)-\bm{\mathcal{G}}\left(\bm{Y},\omega\right)\right]\right\rVert_{W_2}$. Note that the quantum Wasserstein distance of order $1$ on orthonormal product states is exactly the Hamming distance (see Proposition~\ref{prop:quant_wass_prod_states}), so here $\left\lVert\ket{\bm{s}}\bra{\bm{s}}-\ket{\bm{t}}\bra{\bm{t}}\right\rVert_{W_1}$ can simply be thought of as $d_{\mathrm{H}}\left(\bm{s},\bm{t}\right)$.
\begin{proposition}[Expected Wasserstein distance]\label{prop:exp_wass_dist}
    For each $\bm{X}\in\mathbb{R}^D$, let $p_{\bm{X}}\left(\ket{\bm{s}}\bra{\bm{s}}\right)$ be the distribution of $\bm{\mathcal{G}}\left(\bm{X},\omega\right)$ over $\omega\sim\mathcal{U}$, i.e.,
    \begin{equation}
        \mathbb{E}_{\omega\sim\mathcal{U}}\left[\bm{\mathcal{G}}\left(\bm{X},\omega\right)\right]=:\sum_{\bm{s}\in\left[6\right]^{\times n}}p_{\bm{X}}\left(\ket{\bm{s}}\bra{\bm{s}}\right)\ket{\bm{s}}\bra{\bm{s}}.
    \end{equation}
    Then, for every pair $\left(\bm{X},\bm{Y}\right)\in\mathbb{R}^D\times\mathbb{R}^D$, there exists some probability distribution $\pi_{\left(\bm{X},\bm{Y}\right)}\left(\ket{\bm{s}}\bra{\bm{s}},\ket{\bm{t}}\bra{\bm{t}}\right)$ over $\left(\ket{\bm{s}}\bra{\bm{s}},\ket{\bm{t}}\bra{\bm{t}}\right)\in\mathcal{B}^{\times 2}$ satisfying:
    \begin{align}
        \mathbb{E}_{\left(\ket{\bm{s}}\bra{\bm{s}},\ket{\bm{t}}\bra{\bm{t}}\right)\sim\pi_{\left(\bm{X},\bm{Y}\right)}}\left[\left\lVert\ket{\bm{s}}\bra{\bm{s}}-\ket{\bm{t}}\bra{\bm{t}}\right\rVert_{W_1}^2\right]&\leq\left\lVert\mathbb{E}_{\omega\sim\mathcal{U}}\left[\bm{\mathcal{G}}\left(\bm{X},\omega\right)-\bm{\mathcal{G}}\left(\bm{Y},\omega\right)\right]\right\rVert_{W_2}^2,\\
        \sum_{\ket{\bm{t}}\bra{\bm{t}}\in\mathcal{B}}\pi_{\left(\bm{X},\bm{Y}\right)}\left(\ket{\bm{s}}\bra{\bm{s}},\ket{\bm{t}}\bra{\bm{t}}\right)&=p_{\bm{X}}\left(\ket{\bm{s}}\bra{\bm{s}}\right),\label{eq:first_coupling_marg}\\
        \sum_{\ket{\bm{s}}\bra{\bm{s}}\in\mathcal{B}}\pi_{\left(\bm{X},\bm{Y}\right)}\left(\ket{\bm{s}}\bra{\bm{s}},\ket{\bm{t}}\bra{\bm{t}}\right)&=p_{\bm{Y}}\left(\ket{\bm{t}}\bra{\bm{t}}\right).\label{eq:sec_coupling_marg}
    \end{align}
\end{proposition}
\begin{proof}
    As $\mathcal{B}\subseteq\mathcal{B}_6$---the set of $n$-$d$it strings with $d=6$---this is an immediate application of the alternative formulation of the quantum Wasserstein distance over product states given as Proposition~\ref{prop:quant_wass_prod_states}, stated and proved in Appendix~\ref{sec:wass}.
\end{proof}
Unfortunately, $\pi_{\left(\bm{X},\bm{Y}\right)}$ depends on $\left(\bm{X},\bm{Y}\right)$, so this fact cannot be immediately applied to demonstrate the stability of $\bm{\mathcal{G}}$. However, this fact will be enough to allow us to prove the existence of---for some $R\in\mathbb{N}$---stable, pure quantum algorithms $\left\{\bm{\mathcal{I}}_r\right\}_{r=1}^R$ which are stable and near-optimal \emph{collectively} over a distribution of $\left(Q+1\right)$-tuples of inputs. We formalize this as the following lemma.
\begin{lemma}[Reduction to pure, deterministic algorithms]\label{lem:stab_loc_shad_red_along_path}
    Let $\bm{\mathcal{A}}$, $\bm{\mathcal{G}}$, $\mathbb{P}_{\bm{X},\bm{Y}}^{\mathfrak{d},\kappa'}$, and $\bm{\mathcal{R}}_{\bm{X}}$ be as in Lemma~\ref{lem:class_shads_red}. Let $Q\in\mathbb{N}$, and let $\mathbb{P}_Q$ be a distribution over $\bm{X}=\left(\bm{X}_q\right)_{q=0}^Q\in\left(\mathbb{R}^D\right)^{\times\left(Q+1\right)}$ such that the marginal distribution over any pair $\left(\bm{X}_q,\bm{X}_{q+1}\right)$ is $\mathbb{P}_{\bm{X},\bm{Y}}^{\mathfrak{d},\kappa_q}$ with $\kappa_q\in\mathcal{K}$. Fix any $\beta>0$ and $R\in\mathbb{N}$ such that
    \begin{equation}
        \frac{Q}{\beta^2}+Qp_{\mathrm{est}}^R<1.
    \end{equation}
    Consider as well the event:
    \begin{equation}\label{eq:exp_ir_stab_small}
        \begin{aligned}
            \mathcal{C}_{\bm{X}}:=&\bigcap_{q=0}^{Q-1}\left\{\left\lVert\bm{\mathcal{A}}\left(\bm{X}_q\right)-\bm{\mathcal{A}}\left(\bm{X}_{q+1}\right)\right\rVert_{W_2}\leq f+L\left\lVert\bm{X}_q-\bm{X}_{q+1}\right\rVert_1\right\}\\
            &\cap\bigcap_{q=1}^Q\left\{\left\{\Tr\left(\bm{H}_{\bm{X}_q}\bm{\mathcal{A}}\left(\bm{X}_q\right)\right)\geq\gamma E^\ast\sqrt{n}\right\}\cap\mathcal{V}\right\},
        \end{aligned}
    \end{equation}
    which occurs with probability:
    \begin{equation}\label{eq:prob_c_event}
        \mathbb{P}_{\bm{X}\sim\mathbb{P}_Q}\left[\mathcal{C}_{\bm{X}}\right]\geq 1-\left(3Q^2p_{\mathrm{st}}+3Qp_{\mathrm{f}}+Qp_{\mathrm{b}}\right).
    \end{equation}

    For any $\bm{X}\sim\mathbb{P}_Q$ where the event $\mathcal{C}_{\bm{X}}$ occurs, there exists a set of $R$ pure, deterministic quantum algorithms $\left\{\bm{\mathcal{I}}_r:\mathbb{R}^D\to\mathcal{B}\right\}_{r=1}^R$ satisfying for all integer $0\leq q\leq Q-1$:
    \begin{equation}
        \frac{1}{R}\sum_{r=1}^R\left\lVert\bm{\mathcal{I}}_r\left(\bm{X}_q\right)-\bm{\mathcal{I}}_r\left(\bm{X}_{q+1}\right)\right\rVert_{W_1}\leq\beta\left(f+L\left\lVert\bm{X}_q-\bm{X}_{q+1}\right\rVert_1\right)
    \end{equation}
    and, for all $q\in\left[Q\right]$,
    \begin{equation}
        \max_{r\in\left[R\right]}\Tr\left(\bm{\mathcal{R}}_{\bm{X}_q}\bm{\mathcal{I}}_r\left(q\right)\right)\geq\left(\gamma-\delta\right)E^\ast\sqrt{n}.
    \end{equation}
\end{lemma}
\begin{proof}
    Consider $\bm{X}\sim\mathbb{P}_Q$. Let $\pi_{\left(\bm{X}_q,\bm{X}_{q+1}\right)}\left(\ket{\bm{s}_q}\bra{\bm{s}_q},\ket{\bm{s}_{q+1}}\bra{\bm{s}_{q+1}}\right)$ be as in Proposition~\ref{prop:exp_wass_dist}. We have:
    \begin{equation}
        \mathbb{E}_{\left(\ket{\bm{s}_q}\bra{\bm{s}_q},\ket{\bm{s}_{q+1}}\bra{\bm{s}_{q+1}}\right)\sim\pi_{\left(\bm{X}_q,\bm{X}_{q+1}\right)}}\left[\left\lVert\ket{\bm{s}}\bra{\bm{s}}-\ket{\bm{t}}\bra{\bm{t}}\right\rVert_{W_1}^2\right]\leq\left\lVert\mathbb{E}_{\omega\sim\mathcal{U}}\left[\bm{\mathcal{G}}\left(\bm{X}_q,\omega\right)-\bm{\mathcal{G}}\left(\bm{X}_{q+1},\omega\right)\right]\right\rVert_{W_2}^2.
    \end{equation}
    We also recall from Proposition~\ref{prop:exp_wass_dist}:
    \begin{equation}\label{eq:cons_of_margs}
        \begin{aligned}
            p_{\bm{X}_q}\left(\ket{\bm{s}_q}\bra{\bm{s}_q}\right)&=\sum_{\ket{\bm{s}_{q+1}}\bra{\bm{s}_{q+1}}\in\mathcal{B}}\pi_{\left(\bm{X}_q,\bm{X}_{q+1}\right)}\left(\ket{\bm{s}_q}\bra{\bm{s}_q},\ket{\bm{s}_{q+1}}\bra{\bm{s}_{q+1}}\right)\\
            &=\sum_{\ket{\bm{s}_{q-1}}\bra{\bm{s}_{q-1}}\in\mathcal{B}}\pi_{\left(\bm{X}_{q-1},\bm{X}_q\right)}\left(\ket{\bm{s}_{q-1}}\bra{\bm{s}_{q-1}},\ket{\bm{s}_q}\bra{\bm{s}_q}\right),
        \end{aligned}
    \end{equation}
    where the final equality holds due to the compatibility of the marginals of the $\pi_{\left(\bm{X}_q,\bm{X}_{q+1}\right)}$ (Eqs.~\eqref{eq:first_coupling_marg} and~\eqref{eq:sec_coupling_marg}). Finally, we define the probability distribution $\varPi_{\bm{X}}$ over $\left(\ket{\bm{s}_q}\bra{\bm{s}_q}\right)_{q=0}^Q\in\mathcal{B}^{\times\left(Q+1\right)}$ given by:
    \begin{equation}
        \varPi_{\bm{X}}\left(\left(\ket{\bm{s}_i}\bra{\bm{s}_i}\right)_{q=0}^Q\right):=\pi_{\left(\bm{X}_0,\bm{X}_1\right)}\left(\ket{\bm{s}_0}\bra{\bm{s}_0},\ket{\bm{s}_1}\bra{\bm{s}_1}\right)\prod_{q=1}^{Q-1}\frac{\pi_{\left(\bm{X}_q,\bm{X}_{q+1}\right)}\left(\ket{\bm{s}_q}\bra{\bm{s}_q},\ket{\bm{s}_{q+1}}\bra{\bm{s}_{q+1}}\right)}{p_{\bm{X}_q}\left(\ket{\bm{s}_q}\bra{\bm{s}_q}\right)}.
    \end{equation}
    Using the consistency of the single-variable marginals (Eq.~\eqref{eq:cons_of_margs}), it is easy to see by direct calculation that the two-variable marginals of $\varPi_{\bm{X}}$ agree with the $\pi_{\left(\bm{X}_q,\bm{X}_{q+1}\right)}$:
    \begin{equation}\label{eq:bigpi_two_var_margs}
        \varPi_{\bm{X}}\left(\ket{\bm{s}_q}\bra{\bm{s}_q},\ket{\bm{s}_{q+1}}\bra{\bm{s}_{q+1}}\right)=\pi_{\left(\bm{X}_q,\bm{X}_{q+1}\right)}\left(\ket{\bm{s}_q}\bra{\bm{s}_q},\ket{\bm{s}_{q+1}}\bra{\bm{s}_{q+1}}\right).
    \end{equation}

    Now, define a sample space $\varOmega:=\mathcal{B}^{\times\left(Q+1\right)}$. We use the notation $\omega_q$ (zero-indexed) to denote the projection of $\omega\in\varOmega$ to the $q$th of the factors $\mathcal{B}$. With this notation, we define the pure quantum algorithm $\bm{\tilde{\mathcal{I}}}:\left\{q\right\}_{q=0}^Q\times\varOmega\to\mathcal{B}$:
    \begin{equation}
        \bm{\tilde{\mathcal{I}}}\left(q,\omega\right)=\omega_q.
    \end{equation}
    By Markov's inequality, conditioned on $\bm{X}$ being such that the event $\mathcal{C}_{\bm{X}}$ occurs, we have from Markov's inequality that for any $\beta>0$:
    \begin{equation}
        \begin{aligned}
            \mathbb{P}_{\bm{\omega}\sim\varPi_{\bm{X}}^{\times R}\mid\mathcal{C}_{\bm{X}}}&\left[\frac{1}{R}\sum_{r=1}^R\left\lVert\bm{\tilde{\mathcal{I}}}\left(q,\omega_r\right)-\bm{\tilde{\mathcal{I}}}\left(q+1,\omega_r\right)\right\rVert_{W_1}\geq\beta\left(f+L\left\lVert\bm{X}_q-\bm{X}_{q+1}\right\rVert_1\right)\right]\\
            &=\mathbb{P}_{\bm{\omega}\sim\varPi_{\bm{X}}^{\times R}\mid\mathcal{C}_{\bm{X}}}\left[\left(\frac{1}{R}\sum_{r=1}^R\left\lVert\bm{\tilde{\mathcal{I}}}\left(q,\omega_r\right)-\bm{\tilde{\mathcal{I}}}\left(q+1,\omega_r\right)\right\rVert_{W_1}\right)^2\geq\beta^2\left(f+L\left\lVert\bm{X}_q-\bm{X}_{q+1}\right\rVert_1\right)^2\right]\\
            &\leq\frac{\mathbb{E}_{\bm{\omega}\sim\varPi_{\bm{X}}^{\times R}\mid\mathcal{C}_{\bm{X}}}\left[\left(\frac{1}{R}\sum_{r=1}^R\left\lVert\bm{\tilde{\mathcal{I}}}\left(q,\omega_r\right)-\bm{\tilde{\mathcal{I}}}\left(q+1,\omega_r\right)\right\rVert_{W_1}\right)^2\right]}{\beta^2\left(f+L\left\lVert\bm{X}_q-\bm{X}_{q+1}\right\rVert_1\right)^2}\\
            &\leq\frac{\mathbb{E}_{\omega\sim\varPi_{\bm{X}}\mid\mathcal{C}_{\bm{X}}}\left[\left\lVert\bm{\tilde{\mathcal{I}}}\left(q,\omega_r\right)-\bm{\tilde{\mathcal{I}}}\left(q+1,\omega_r\right)\right\rVert_{W_1}^2\right]}{\beta^2\left(f+L\left\lVert\bm{X}_q-\bm{X}_{q+1}\right\rVert_1\right)^2}\\
            &\leq\frac{1}{\beta^2}
        \end{aligned}
    \end{equation}
    for any integer $0\leq q\leq Q-1$. By the union bound,
    \begin{equation}\label{eq:stab_r_algs}
        \mathbb{P}_{\bm{\omega}\sim\varPi_{\bm{X}}^{\times R}\mid\mathcal{C}_{\bm{X}}}\left[\bigcap_{q=0}^{Q-1}\frac{1}{R}\sum_{r=1}^R\left\lVert\bm{\tilde{\mathcal{I}}}\left(q,\omega_r\right)-\bm{\tilde{\mathcal{I}}}\left(q+1,\omega_r\right)\right\rVert_{W_1}\leq\beta\left(f+L\left\lVert\bm{X}_q-\bm{X}_{q+1}\right\rVert_1\right)\right]\geq 1-\frac{Q}{\beta^2}.
    \end{equation}
    Furthermore, by independence over the $R$ replicas and the union bound,
    \begin{equation}\label{eq:no_r_algs}
        \mathbb{P}_{\bm{\omega}\sim\varPi_{\bm{X}}^{\times R}\mid\mathcal{C}_{\bm{X}}}\left[\bigcap_{q=1}^Q\bigcup_{r=1}^R\Tr\left(\bm{\mathcal{R}}_{\bm{X}_q}\bm{\tilde{\mathcal{I}}}\left(q,\omega_r\right)\right)\geq\left(\gamma-\delta\right)E^\ast\sqrt{n}\right]\geq 1-Qp_{\mathrm{est}}^R.
    \end{equation}
    In particular, assuming $\beta$ and $R$ are sufficiently large such that:
    \begin{equation}
        \frac{Q}{\beta^2}+Qp_{\mathrm{est}}^R<1,
    \end{equation}
    we have from the law of total probability that there exists some $\bm{\omega^\ast}\in\varOmega^{\times R}=\mathcal{B}^{\times\left(R\times\left(Q+1\right)\right)}$ such that the events in Eqs.~\eqref{eq:stab_r_algs} and~\eqref{eq:no_r_algs} occur. The final result follows by defining:
    \begin{equation}
        \bm{\mathcal{I}}_r\left(\bm{X}_q\right):=\bm{\tilde{\mathcal{I}}}\left(q,\omega_r^\ast\right).
    \end{equation}
\end{proof}

\subsection{Considering Many Replicas}

We now apply Lemmas~\ref{lem:det_rand_alg_red} through~\ref{lem:stab_loc_shad_red_along_path} in sequence to $T+1$ replicas. Reasoning about many replicas will later allow us to demonstrate that certain events occur w.h.p. on some cardinality-$m$ subset of them, which will then be used in a proof by contradiction combined with the $m$-QOGP (or the quantum chaos property in the weaker case).

For concreteness we consider a specific choice of interpolation path. For each integer $0\leq q\leq Q$, we define $\bm{\tau}_q\in\left\{0,1\right\}^{\times D}$ to be of the form:
\begin{equation}
    \bm{\tau}_q=\left(\bm{1}\left\{\operatorname{supp}\left(\bm{R}_i\right)\cap\mathcal{P}_q\neq\varnothing\right\}\right)_{i=1}^D,
\end{equation}
where $\mathcal{P}_q\subseteq\left[n\right]$ is defined as the subset:
\begin{equation}
    \mathcal{P}_q:=\left\{i\right\}_{i=1}^{\min\left(n,q\left\lceil\frac{n}{Q}\right\rceil\right)}
\end{equation}
with $\mathcal{P}_0:=\varnothing$. In words, $\bm{\tau}_q$ is the indicator vector that is $1$ whenever the local Hamiltonian term $\bm{R}_i$ has support intersecting the first $\frac{q}{Q}$-fraction of qubits. By construction these $\bm{\tau}$ have the property:
\begin{equation}
    \tau_i=0\iff\operatorname{supp}\left(\bm{R}_i\right)\subseteq\left[n\right]\setminus\mathcal{P}_q;
\end{equation}
furthermore, for any $q>0$,
\begin{equation}
    \left\lvert\left[n\right]\setminus\mathcal{P}_q\right\rvert\leq\left(1-\frac{q}{Q}\right)n.
\end{equation}
The set $\mathcal{I}=\left\{\bm{\tau}_q\right\}_{q=0}^Q$ by construction is a $\left(c,F,R\right)$-correlation set (Definition~\ref{def:cfr_corr_set}) with $c=n^{-1}\log_2\left(Q\right)$, $F=\frac{1}{Q}$, and any $R$. We also define a distribution $\mathbb{P}_{T,Q}$, given by the joint distribution of:
\begin{equation}
    \bm{X}_q^{\left(t\right)}:=\left(\bm{\tau}_Q-\bm{\tau}_q\right)\odot\frac{\bm{S}}{\sqrt{p}}\odot\bm{J}^{\left(0\right)}+\bm{\tau}_q\odot\frac{\bm{S}}{\sqrt{p}}\odot\bm{J}^{\left(t\right)}.
\end{equation}

Our main result here is combining all of the lemmas proven to this point, and bounding the probability that collective stability and near-optimality hold for the $\bm{\mathcal{I}}_r$ over the $T$ replicas.
\begin{proposition}[Considering many replicas]\label{prop:red_pure_prod_algs}
    Let $\mathcal{H}=\left\{\bm{H}_{\bm{X}}\right\}$ be a class of random Hamiltonians with $\left(\delta,p_{\mathrm{est}},p_{\mathrm{b}}\right)$-efficient local shadows estimator, and let $\bm{\mathcal{R}}$ be the associated linear map as defined in Definition~\ref{def:eff_loc_shad_est}. Fix $T\in\mathbb{N}$, $Q\in\mathbb{N}$, and $\gamma\in\left[0,1\right]$. Let $\bm{\mathcal{A}}$ be a quantum algorithm that is both $\left(f,L,\mathfrak{d},\kappa,p_{\mathrm{st}}\right)$-stable and $\left(\gamma,p_{\mathrm{f}}\right)$-optimal for $\mathcal{H}$, where we assume that $\kappa\leq\max\left(0,1-\frac{1.001}{Q}\right)$. Finally, fix any $\beta>0$ and $R\in\mathbb{N}$ such that
    \begin{equation}
        \frac{Q}{\beta^2}+Qp_{\mathrm{est}}^R<1.
    \end{equation}

    Condition on the event $\mathcal{X}$ (as defined in Proposition~\ref{prop:max_deg_rand_hyp}), and assume that $\mathfrak{d}\geq d_{\mathrm{max}}$. Then,
    \begin{equation}
        \mathbb{P}_{\bm{X}\sim\mathbb{P}_{T,Q}}\left[\tilde{\mathcal{Y}}\mid\mathcal{X}\right]\geq 1-T\left(3Q^2 p_{\mathrm{st}}+3Qp_{\mathrm{f}}+Qp_{\mathrm{b}}\right),
    \end{equation}
    where we have defined the event:
    \begin{equation}\label{eq:suc_int}
        \begin{aligned}
            \mathcal{\tilde{Y}}&:=\exists\left\{\bm{\mathcal{I}}_r:\mathbb{R}^D\to\mathcal{B}\right\}_{r=1}^R:\\
            &\left\{\left\{\bigcap_{\substack{t\in\left[T\right]\\0\leq q\leq Q-1}}\frac{1}{R}\sum_{r=1}^R\left\lVert\bm{\mathcal{I}}_r\left(\bm{X}_q^{\left(t\right)}\right)-\bm{\mathcal{I}}_r\left(\bm{X}_{q+1}^{\left(t\right)}\right)\right\rVert_{W_1}\leq\beta\left(f+L\left\lVert\bm{X}_q^{\left(t\right)}-\bm{X}_{q+1}^{\left(t\right)}\right\rVert_1\right)\right\}\right.\\
            &\left.\cap\left\{\bigcap_{\substack{t\in T\\q\in\left[Q\right]}}\bigcup_{r=1}^R\Tr\left(\bm{\mathcal{R}}_{\bm{X}_q^{\left(t\right)}}\bm{\mathcal{I}}_r\left(\bm{X}_q^{\left(t\right)}\right)\right)\geq\left(\gamma-\delta\right)E^\ast\sqrt{n}\right\}\right\}.
        \end{aligned}
    \end{equation}
\end{proposition}
\begin{proof}
    First, note that for all $t\in\left[T\right]$:
    \begin{equation}
        \begin{aligned}
            \left\lVert\bm{X}_q^{\left(t\right)}-\bm{X}_{q+1}^{\left(t\right)}\right\rVert_0=\left\lVert 2\left(\bm{\tau}_{q+1}-\bm{\tau}_q\right)\right\rVert_0\leq\left\lceil\frac{n}{Q}\right\rceil.
        \end{aligned}
    \end{equation}
    In particular, in the notation of Lemma~\ref{lem:stab_loc_shad_red_along_path}, conditioned on $\mathcal{X}$ the distribution $\mathbb{P}_{T,Q}$ marginalizes to a distribution $\mathbb{P}_Q$ over each replica $t$ of the form of $\mathbb{P}_{\bm{X},\bm{Y}}^{\mathfrak{d},\kappa_q}$ where,
    \begin{equation}
        \kappa_q\geq 1-\frac{1}{n}\left\lceil\frac{n}{Q}\right\rceil\geq\max\left(0,1-\frac{1.001}{Q}\right),
    \end{equation}
    with the final inequality following in the limit of sufficiently large $n$.

    We now consider $\left\{\bm{\mathcal{I}}_r:\mathbb{R}^D\to\mathcal{B}\right\}_{r=1}^R$ as in Lemma~\ref{lem:stab_loc_shad_red_along_path}. Though the full construction used in the proof depends on the entire path $\left(\bm{X}_q\right)_{q=0}^Q$, inspection of the proof reveals that $\bm{\mathcal{I}}_r\left(\bm{X}_0\right)$ depends only on $\bm{X}_0$, and in particular we have the consistency relation:
    \begin{equation}
        \bm{\mathcal{I}}_r\left(\bm{X}_0^{\left(t\right)}\right)=\bm{\mathcal{I}}_r\left(\bm{X}_0^{\left(t'\right)}\right)
    \end{equation}
    for all $t,t'\in\left[T\right]$ and $r\in\left[R\right]$ as $\bm{X}_0^{\left(t\right)}=\bm{X}_0^{\left(t'\right)}$. The final result holds by recalling the probability of the event $\mathcal{C}_{\bm{X}}$ occurring from Lemma~\ref{lem:stab_loc_shad_red_along_path}, as well as the union bound.
\end{proof}

\subsection{Stability Between Replicas}

We now show that the pairwise quantum Wasserstein distances of the $\bm{\mathcal{I}}_r$ between replicas are stable as one goes from $\bm{\tau}_q$ to $\bm{\tau}_{q+1}$ along each interpolation path. As the output of the $\bm{\mathcal{I}}_r$ between replicas are identical at $q=0$, the pairwise distances here are $0$; this fact along with the following stability result will then allow us to reason about the quantum Wasserstein distances at general $q\in\left[Q\right]$.
\begin{lemma}[Stability of quantum $W_1$ distance along interpolation paths]\label{lem:stab_int}
    Let $\mathcal{\tilde{Y}}$ and $\left\{\bm{\mathcal{I}}_r\right\}_{r=1}^R$ be as in Proposition~\ref{prop:red_pure_prod_algs}. Conditioned on the event $\mathcal{X}$ (Eq.~\eqref{eq:x_event_def}) occurring,
    \begin{equation}\label{eq:prob_bound_stability}
        \mathbb{P}_{\bm{X}\sim\mathbb{P}_{T,Q}}\left[\mathcal{Y}\mid\mathcal{X}\right]\geq 1-T\left(3Q^2p_{\mathrm{st}}+3Qp_{\mathrm{f}}+Qp_{\mathrm{b}}\right)-\left(T+1\right)Q\exp\left(-\operatorname{\Omega}\left(\frac{d_{\mathrm{max}}n}{Q}\right)\right),
    \end{equation}
    where
    \begin{equation}
        \begin{aligned}
            \mathcal{Y}:=\tilde{\mathcal{Y}}\cap&\\
            \bigcap_{\substack{t\neq t'\in\left[T\right]\\0\leq q\leq Q-1}}&\left\{\frac{1}{Rn}\sum_{r=1}^R\left\lvert\left\lVert\bm{\mathcal{I}}_r\left(\bm{X}_q^{\left(t\right)}\right)-\bm{\mathcal{I}}_r\left(\bm{X}_q^{\left(t'\right)}\right)\right\rVert_{W_1}-\left\lVert\bm{\mathcal{I}}_r\left(\bm{X}_{q+1}^{\left(t\right)}\right)-\bm{\mathcal{I}}_r\left(\bm{X}_{q+1}^{\left(t'\right)}\right)\right\rVert_{W_1}\right\rvert\leq\frac{2\beta f}{n}+\frac{12d_{\mathrm{max}}\beta L}{Q}\right\}.
        \end{aligned}
    \end{equation}
\end{lemma}
\begin{proof}
    Note that, for all integer $0\leq t\leq T$ and $0\leq q\leq Q-1$, each $\bm{X}_{q+1}^{\left(t\right)}-\bm{X}_q^{\left(t\right)}$ has (conditioned on the event $\mathcal{X}$ occurring) at most $\frac{d_{\mathrm{max}}n}{Q}$ nonzero entries. Thus, by standard tail bounds on the $L^1$-norm of Gaussian random vectors~\cite{10.1155/2022/1456713} and the union bound,
    \begin{equation}
        \mathbb{P}_{\bm{X}\sim\mathbb{P}_{T,Q}}\left[\bigcap_{t=0}^T\bigcap_{q=0}^{Q-1}\left\lVert\bm{X}_{q+1}^{\left(t\right)}-\bm{X}_q^{\left(t\right)}\right\rVert_1\leq\frac{6d_{\mathrm{max}}n}{Q}\right]\geq 1-\left(T+1\right)Q\exp\left(-\operatorname{\Omega}\left(\frac{d_{\mathrm{max}}n}{Q}\right)\right).
    \end{equation}
    If we are able to demonstrate that $\mathcal{Y}$ holds when conditioned on
    \begin{equation}\label{eq:whp_x_diff_bound}
        \mathcal{D}:=\left\{\bigcap_{t=0}^T\bigcap_{q=0}^{Q-1}\left\lVert\bm{X}_{q+1}^{\left(t\right)}-\bm{X}_q^{\left(t\right)}\right\rVert_1\leq\frac{6d_{\mathrm{max}}n}{Q}\right\}\cap\tilde{\mathcal{Y}}
    \end{equation}
    and $\mathcal{X}$ we will have proven the lemma as, by the union bound, the probability of $\mathcal{D}$ occurring (conditioned on $\mathcal{X}$) is at least:
    \begin{equation}
        \mathbb{P}_{\bm{X}\sim\mathbb{P}_{T,Q}}\left[\mathcal{D}\mid\mathcal{X}\right]\geq 1-T\left(3Q^2p_{\mathrm{st}}+3Qp_{\mathrm{f}}+Qp_{\mathrm{b}}\right)-\left(T+1\right)Q\exp\left(-\operatorname{\Omega}\left(\frac{d_{\mathrm{max}}n}{Q}\right)\right).
    \end{equation}

    By the triangle inequality, we have for all $t\neq t'\in\left[T\right]$ and $r\in\left[R\right]$:
    \begin{equation}
        \begin{aligned}
            &\left\lVert\bm{\mathcal{I}}_r\left(\bm{X}_q^{\left(t\right)}\right)-\bm{\mathcal{I}}_r\left(\bm{X}_q^{\left(t'\right)}\right)\right\rVert_{W_1}-\left\lVert\bm{\mathcal{I}}_r\left(\bm{X}_{q+1}^{\left(t\right)}\right)-\bm{\mathcal{I}}_r\left(\bm{X}_{q+1}^{\left(t'\right)}\right)\right\rVert_{W_1}\\
            &\leq\left\lVert\bm{\mathcal{I}}_r\left(\bm{X}_q^{\left(t\right)}\right)-\bm{\mathcal{I}}_r\left(\bm{X}_{q+1}^{\left(t\right)}\right)\right\rVert_{W_1}+\left\lVert\bm{\mathcal{I}}_r\left(\bm{X}_{q+1}^{\left(t'\right)}\right)-\bm{\mathcal{I}}_r\left(\bm{X}_q^{\left(t'\right)}\right)\right\rVert_{W_1},
        \end{aligned}
    \end{equation}
    and similarly for its negative. Thus, conditioned on the event $\mathcal{D}$, for all $t\neq t'\in\left[T\right]$:
    \begin{equation}
        \begin{aligned}
            \frac{1}{R}\sum_{r=1}^R&\left\lvert\left\lVert\bm{\mathcal{I}}_r\left(\bm{X}_q^{\left(t\right)}\right)-\bm{\mathcal{I}}_r\left(\bm{X}_q^{\left(t'\right)}\right)\right\rVert_{W_1}-\left\lVert\bm{\mathcal{I}}_r\left(\bm{X}_{q+1}^{\left(t\right)}\right)-\bm{\mathcal{I}}_r\left(\bm{X}_{q+1}^{\left(t'\right)}\right)\right\rVert_{W_1}\right\rvert\\
            &\leq\frac{1}{R}\sum_{r=1}^R\left\lVert\bm{\mathcal{I}}_r\left(\bm{X}_q^{\left(t\right)}\right)-\bm{\mathcal{I}}_r\left(\bm{X}_{q+1}^{\left(t\right)}\right)\right\rVert_{W_1}+\frac{1}{R}\sum_{r=1}^R\left\lVert\bm{\mathcal{I}}_r\left(\bm{X}_{q+1}^{\left(t'\right)}\right)-\bm{\mathcal{I}}_r\left(\bm{X}_q^{\left(t'\right)}\right)\right\rVert_{W_1}\\
            &\leq 2\beta f+\beta L\left\lVert\bm{X}_q^{\left(t\right)}-\bm{X}_{q+1}^{\left(t\right)}\right\rVert_1+\beta L\left\lVert\bm{X}_q^{\left(t'\right)}-\bm{X}_{q+1}^{\left(t'\right)}\right\rVert_1\\
            &\leq 2\beta f+2\beta L\left(\frac{6d_{\mathrm{max}}n}{Q}\right),
        \end{aligned}
    \end{equation}
    where the final line follows from the condition Eq.~\eqref{eq:whp_x_diff_bound}. Dividing both sides by $n$ gives $\mathcal{Y}$.
\end{proof}

\subsection{Distant Clustering for Independent Instances}

Recall the definition of the random set $\mathcal{S}\left(\gamma,m,\xi,\eta,\mathcal{I},R\right)$ given in Definition~\ref{def:s_def}; we use
\begin{equation}
    \mathcal{S}_{\bm{X}^{\left(0\right)},\left\{\bm{X}^{\left(t\right)}\right\}_{t\in\mathcal{M}}}\left(\gamma,m,\xi,\eta,\left\{\bm{1}\right\},R\right)
\end{equation}
to denote $\mathcal{S}\left(\gamma,m,\xi,\eta,\left\{\bm{1}\right\},R\right)$ conditioned on the randomness. We now bound the probability that, for all choices of $m$ replicas $\left\{\bm{X}^{\left(t\right)}\right\}_{t\in\mathcal{M}}$ from a set of $T$ instances, $\mathcal{S}_{\bm{X}^{\left(0\right)},\left\{\bm{X}^{\left(t\right)}\right\}_{t\in\mathcal{M}}}\left(\gamma,m,\xi,\eta,\mathcal{I},R\right)$ is empty.
\begin{lemma}[Distant clustering for independent instances]\label{lem:ind_clust_prob}
    Assume the problem class $\mathcal{H}$ with efficient local shadows estimator satisfies the quantum chaos property with parameters $\left(\gamma^\ast,m,\eta',R\right)$. Consider any $T\geq m$. Then:
    \begin{equation}
        \mathbb{P}_{\bm{X}\sim\mathbb{P}_{T,Q}}\left[\mathcal{Z}\right]\geq 1-\binom{T}{m}\exp\left(-\operatorname{\Omega}\left(n\right)\right),
    \end{equation}
    where
    \begin{equation}
       \mathcal{Z}:=\bigcap_{\substack{\mathcal{M}\in\binom{\left[T\right]}{m}:}}\left\{\mathcal{S}_{\bm{X}^{\left(0\right)},\left\{\bm{X}^{\left(t\right)}\right\}_{t\in\mathcal{M}}}\left(\gamma^\ast,m,1,\eta',\left\{\bm{1}\right\},R\right)=\varnothing\right\}.
    \end{equation}
\end{lemma}
\begin{proof}
    This follows immediately from the union bound and the definition of the quantum chaos property (Definition~\ref{def:qcp}).
\end{proof}

\subsection{Topologically Obstructed Configurations Conditioned on Events}

Our strategy is now to show, conditioned on all of the previously-introduced events occurring, that the algorithm must output configurations that are topologically obstructed by the $m$-QOGP (or, in the weaker case, the quantum chaos property). To do this we construct a graph $G_{T,Q}=\left(V,E\right)$ which depends on the randomness $\bm{X}\sim\mathbb{P}_{T,Q}$ in the following way. If $\mathcal{H}$ satisfies the $m$-QOGP with parameters $\left(\gamma^\ast,m,\xi,\eta,c,\eta',F,R\right)$, we define $G_{T,Q}$ as:
\begin{itemize}
    \item $G_{T,Q}$ has $T$ vertices, i.e., $V=\left[T\right]$;
    \item $\left(t,t'\right)\in E$ if and only if $t\neq t'$ and $\exists q\in\left[Q\right]$ such that:
    \begin{equation}\label{eq:overlap_of_alg_output}
        \frac{1}{R}\sum_{r=1}^R\left\lVert\bm{\mathcal{I}}_r\left(\bm{X}_q^{\left(t\right)}\right)-\bm{\mathcal{I}}_r\left(\bm{X}_q^{\left(t'\right)}\right)\right\rVert_{W_1}\in\left[\frac{1-\xi}{2}n,\frac{1-\xi+\eta}{2}n\right].
    \end{equation}
    We color the edge with the smallest $q\in\left[Q\right]$ for which Eq.~\eqref{eq:overlap_of_alg_output} is satisfied.
\end{itemize}
If $\mathcal{H}$ only satisfies the quantum chaos property with parameters $\left(\gamma^\ast,m,\eta,R\right)$, we take the same definition for $G_{T,Q}$ with $\xi=1$.

We claim that $G_{T,Q}$ is \emph{$m$-admissible} when conditioned on the events $\mathcal{Y}$ (from Lemma~\ref{lem:stab_int}) and $\mathcal{Z}$ (from Lemma~\ref{lem:ind_clust_prob}) occurring. First, we define $m$-admissibility.
\begin{definition}[$m$-admissibility]
    Let $m\in\mathbb{N}$. A graph $G=\left(V,E\right)$ is said to be \emph{$m$-admissible} if, for all $\mathcal{M}\subseteq V$ with $\left\lvert\mathcal{M}\right\rvert=m$, there exist distinct $i,j\in\mathcal{M}$ such that $\left(i,j\right)\in E$.
\end{definition}

\begin{lemma}[$m$-admissibility of $G_{T,Q}$, $m$-QOGP]\label{lem:m_admiss}
    Assume the problem class $\mathcal{H}$ with efficient local shadows estimator satisfies the $m$-QOGP with parameters $\left(\gamma^\ast,m,\xi,\eta,c,\eta',F,R\right)$ or the quantum chaos property with parameters $\left(\gamma^\ast,m,\eta,R\right)$. Further, assume that:
    \begin{equation}\label{eq:suff_stab_m_admiss}
        \frac{2\beta f}{n}+\frac{12d_{\mathrm{max}}\beta L}{Q}\leq\frac{\eta}{4}.
    \end{equation}
    Conditioned on the events $\mathcal{Y}$ and $\mathcal{Z}$, $G_{T,Q}$ is $m$-admissible.
\end{lemma}
\begin{proof}
    In what follows, we take $\xi=1$ and $\eta'=\eta$ if $\mathcal{H}$ satisfies only the quantum chaos property. Given the definition of $G_{T,Q}$, the lemma statement is implied if one shows that, for arbitrary $\mathcal{M}\subseteq V$ with $\left\lvert\mathcal{M}\right\rvert=m$,
    \begin{equation}\label{eq:target_for_m_admiss}
        \frac{1}{Rn}\sum_{r=1}^R\left\lVert\bm{\mathcal{I}}_r\left(\bm{X}_q^{\left(t\right)}\right)-\bm{\mathcal{I}}_r\left(\bm{X}_q^{\left(t'\right)}\right)\right\rVert_{W_1}\in\left[\frac{1-\xi}{2},\frac{1-\xi+\eta}{2}\right]
    \end{equation}
    for some $t,t'\in\mathcal{M}$ and $q\in\left[Q\right]$ when conditioned on the events $\mathcal{Y}$ (from Lemma~\ref{lem:stab_int}) and $\mathcal{Z}$ (from Lemma~\ref{lem:ind_clust_prob}) occurring. For notational convenience, for the remainder of this proof we define:
    \begin{equation}
        p_{t,t';q}:=\frac{1}{Rn}\sum_{r=1}^R\left\lVert\bm{\mathcal{I}}_r\left(\bm{X}_q^{\left(t\right)}\right)-\bm{\mathcal{I}}_r\left(\bm{X}_q^{\left(t'\right)}\right)\right\rVert_{W_1}.
    \end{equation}

    We first claim that $p_{t,t';q}$ is Lipschitz in $q$ for any choice of $t,t'\in\mathcal{M}$. This follows from the assumption given in Eq.~\eqref{eq:suff_stab_m_admiss} and the conditioning on the event $\mathcal{Y}$, such that:
    \begin{equation}\label{eq:p_is_lipschitz}
        \left\lvert p_{t,t';q}-p_{t,t';q+1}\right\rvert\leq\frac{2f}{n}+\frac{12d_{\mathrm{max}}\beta L}{Q}\leq\frac{\eta}{4}
    \end{equation}
    for all $t,t'\in\left[T\right]$ and $0\leq q\leq Q-1$. Furthermore, as $\bm{X}_0^{\left(t\right)}=\bm{X}_0^{\left(t'\right)}$ for all $t,t'\in\left[T\right]$, we have:
    \begin{equation}\label{eq:p_k_0}
        p_{t,t';0}=0
    \end{equation}
    for all $t,t'\in\left[T\right]$. Finally, conditioned on $\mathcal{Z}$,
    \begin{equation}
        p_{t,t';Q}>\frac{\eta'}{2}
    \end{equation}
    for some $t,t'\in\mathcal{M}$. We call the $\left(t,t'\right)$ pair for which this is true $\left(s^\ast,t^\ast\right)$.

    Let $\delta:=\frac{\eta}{4}$. Recall from the definition of the $m$-QOGP (Definition~\ref{def:mqogp}) that $\eta'\geq 1-\xi+\eta$, so it is additionally true that:
    \begin{equation}\label{eq:p_k_Q}
        p_{s^\ast,t^\ast;Q}>\frac{\eta'}{2}\geq\frac{1-\xi+\eta}{2}>\frac{1-\xi+\delta}{2}.
    \end{equation}
    Let $q^\ast$ be the largest $0\leq q\leq Q$ such that $p_{s^\ast,t^\ast;q}\leq\frac{1-\xi+\delta}{2}$; by Eq.~\eqref{eq:p_k_0} such a $q^\ast$ exists, and by Eq.~\eqref{eq:p_k_Q} $q^\ast<Q$ so $q^\ast+1\in\left[Q\right]$. By definition, for all $q>q^\ast$ it is the case that $p_{s^\ast,t^\ast;q}>\frac{1-\xi+\delta}{2}$.\footnote{Note that the converse is not necessarily true as $p_{s^\ast,t^\ast;q}$ may not be monotonic in $q$.} Similarly, by the Lipschitz property demonstrated in Eq.~\eqref{eq:p_is_lipschitz}, it must be that:
    \begin{equation}
        p_{s^\ast,t^\ast;q^\ast+1}<\frac{1-\xi+\delta}{2}+\frac{\eta}{4}=\frac{1-\xi+3\delta}{2}=\frac{1-\xi+\eta-\delta}{2}.
    \end{equation}
    Taken together, $q^\ast+1\in\left[Q\right]$ is such that:
    \begin{equation}
        p_{s^\ast,t^\ast;q^\ast+1}\in\left(\frac{1-\xi+\delta}{2},\frac{1-\xi+\eta-\delta}{2}\right)\subset\left[\frac{1-\xi}{2},\frac{1-\xi+\eta}{2}\right].
    \end{equation}
    Eq.~\eqref{eq:target_for_m_admiss} is thus satisfied for arbitrary $\mathcal{M}$ with $\left\lvert\mathcal{M}\right\rvert=m$ by choosing $\left(t,t'\right)=\left(s^\ast,t^\ast\right)$ and $q=q^\ast$.
\end{proof}

As $G_{T,Q}$ is $m$-admissible, due to a result from Ramsey theory it has a clique of cardinality $m$ for sufficiently large $T=\left\lvert V\right\rvert$ with respect to $m$ and $Q$~\cite{10.1214/23-AAP1953}.
\begin{proposition}[$G$ contains a monochromatic $m$-clique~{\cite[Proposition~6.12]{10.1214/23-AAP1953}}]
    Assume $G$ is $m$-admissible, has $C$ edge colors, and has $\exp_2\left(C^{4mC}\right)$ vertices. Then, $G$ has a monochromatic clique of cardinality $m$.
\end{proposition}
This immediately gives the following result when applied to $G_{T,Q}$.
\begin{proposition}[$G_{T,Q}$ contains an $m$-clique]\label{prop:m_clique}
    If $T=\exp_2\left(Q^{4mQ}\right)$, $G_{T,Q}$ conditioned on $\mathcal{Y}$ and $\mathcal{Z}$ has a monochromatic clique of cardinality $m$ if:
    \begin{equation}
        \frac{2f}{n}+\frac{12d_{\mathrm{max}}\beta L}{Q}\leq\frac{\eta}{4}.
    \end{equation}
\end{proposition}

\subsection{Completing the Proof}

We now have all of the ingredients to complete Theorem~\ref{thm:m_qogp_implies_alg_hardness}. First, we lower bound the probability that the event:
\begin{equation}
    \mathcal{W}:=\mathcal{X}\cap\mathcal{Y}\cap\mathcal{Z}
\end{equation}
occurs, with $\mathcal{X}$, $\mathcal{Y}$, and $\mathcal{Z}$ defined in Proposition~\ref{prop:max_deg_rand_hyp}, Lemma~\ref{lem:stab_int}, and Lemma~\ref{lem:ind_clust_prob}, respectively.
\begin{lemma}[Probability of good events]
    Assume $p_{\mathrm{st}}$, $p_{\mathrm{f}}$, $p_{\mathrm{est}}$, $p_{\mathrm{b}}$, $m$, $T\in\mathbb{N}$, $Q\in\mathbb{N}$, $\beta\in\mathbb{R}^+$, and $R\in\mathbb{N}$ are such that:
    \begin{align}
        \frac{Q}{\beta^2}+Qp_{\mathrm{est}}^R&<1;\\
        \binom{T}{m}&\leq\exp\left(\operatorname{o}\left(n\right)\right);\\
        TQ\left(3Qp_{\mathrm{st}}+3p_{\mathrm{f}}+p_{\mathrm{b}}\right)&\leq 1-\exp\left(-\operatorname{o}\left(n\right)\right).
    \end{align}
    Fix any $\epsilon>0$ sufficiently small, and consider the corresponding $d_{\mathrm{max}}$ in Eq.~\eqref{eq:x_event_def}. Then,
    \begin{equation}
        \mathbb{P}_{\bm{X}\sim\mathbb{P}_{T,Q}}\left[\mathcal{W}\right]\geq\left(1+\operatorname{o}\left(1\right)\right)\exp\left(-\epsilon n\right).
    \end{equation}
\end{lemma}
\begin{proof}
    From Proposition~\ref{prop:max_deg_rand_hyp}, Lemma~\ref{lem:stab_int}, Lemma~\ref{lem:ind_clust_prob}, and the union bound, $\mathcal{W}$ occurs with probability at least:
    \begin{equation}
        \begin{aligned}
            \mathbb{P}_{\bm{X}\sim\mathbb{P}_{T,Q}}\left[\mathcal{W}\right]&\geq 1-\mathbb{P}_{\bm{X}\sim\mathbb{P}_{T,Q}}\left[\mathcal{Z}^\complement\right]-\mathbb{P}_{\bm{X}\sim\mathbb{P}_{T,Q}}\left[\left(\mathcal{X}\cap\mathcal{Y}\right)^\complement\right]\\
            &=\mathbb{P}_{\bm{X}\sim\mathbb{P}_{T,Q}}\left[\mathcal{Y}\mid\mathcal{X}\right]\mathbb{P}_{\bm{X}\sim\mathbb{P}_{T,Q}}\left[\mathcal{X}\right]-\mathbb{P}_{\bm{X}\sim\mathbb{P}_{T,Q}}\left[\mathcal{Z}^\complement\right]\\
            &\geq\left(1+\operatorname{o}\left(1\right)\right)\exp\left(-\epsilon n\right)-\exp\left(-\operatorname{\Omega}\left(n\right)\right)\\
            &\geq\left(1+\operatorname{o}\left(1\right)\right)\exp\left(-\epsilon n\right),
        \end{aligned}
    \end{equation}
    where the final inequality follows if $\epsilon>0$ is chosen to be sufficiently small.
\end{proof}

We now use Proposition~\ref{prop:m_clique} to show a contradiction with the statement of Theorem~\ref{thm:m_qogp_implies_alg_hardness}. Conditioned on $\mathcal{W}$, Proposition~\ref{prop:m_clique} states that there exists some subset $\mathcal{M}\in\binom{\left[T\right]}{m}$ and $q\in\left[Q\right]$ such that, for all $t\neq t'\in\mathcal{M}$,
\begin{align}
    \frac{1}{R}\sum_{r=1}^R\left\lVert\bm{\mathcal{I}}_r\left(\bm{X}_q^{\left(t\right)}\right)-\bm{\mathcal{I}}_r\left(\bm{X}_q^{\left(t'\right)}\right)\right\rVert_{W_1}&\in\left[\frac{1-\xi}{2}n,\frac{1-\xi+\eta}{2}n\right],\\
    \max_{r\in\left[R\right]}\Tr\left(\bm{\mathcal{R}}_{\bm{X}_q^{\left(t\right)}}\bm{\mathcal{I}}_r\left(\bm{X}_q^{\left(t\right)}\right)\right)&\geq\left(\gamma-\delta\right)E^\ast\sqrt{n},\\
    \max_{r\in\left[R\right]}\Tr\left(\bm{\mathcal{R}}_{\bm{X}_q^{\left(t'\right)}}\bm{\mathcal{I}}_r\left(\bm{X}_q^{\left(t'\right)}\right)\right)&\geq\left(\gamma-\delta\right)E^\ast\sqrt{n},
\end{align}
where as before we take $\xi$ to be $1$ if $\mathcal{H}$ satisfies only the quantum chaos property. Namely, recalling the definition of the set $\mathcal{S}\left(\cdot\right)$ (Definition~\ref{def:s_def}), conditioned on $\mathcal{W}$ it is the case that this set is nonempty. That is,
\begin{equation}
    \begin{aligned}
        \mathbb{P}_{\bm{X}\sim\mathbb{P}_{T,Q}}\left[\mathcal{S}\left(\gamma-\delta,m,\xi,\eta,\left\{\bm{\tau}_q\right\}_{q\in\left[Q\right]},R\right)\neq\varnothing\right]&\geq\mathbb{P}_{\bm{X}\sim\mathbb{P}_{T,Q}}\left[\mathcal{W}\right]\\
        &\geq\left(1+\operatorname{o}\left(1\right)\right)\exp\left(-\epsilon n\right),
    \end{aligned}
\end{equation}
where recall that $\epsilon>0$ can be chosen to be arbitrarily small.

We first assume $\mathcal{H}$ satisfies the $m$-QOGP with parameters $\left(\gamma^\ast,m,\xi,\eta,c,\eta',F,R\right)$. By Definition~\ref{def:mqogp},
\begin{equation}
    \mathbb{P}_{\bm{X}\sim\mathbb{P}_{T,Q}}\left[\mathcal{S}\left(\gamma-\delta,m,\xi,\eta,\left\{\bm{\tau}_q\right\}_{q\in\left[Q\right]},R\right)\neq\varnothing\right]\leq\exp\left(-\operatorname{\Omega}\left(n\right)\right)
\end{equation}
when $\gamma-\delta>\gamma^\ast$ and $Q\leq\exp_2\left(cn\right)$. This yields a contradiction for sufficiently small $\epsilon>0$---that is, sufficiently large constant $b$ in Proposition~\ref{prop:max_deg_rand_hyp}---and sufficiently large $n$, completing the proof of Theorem~\ref{thm:m_qogp_implies_alg_hardness} in this case.

Now assume that $\mathcal{H}$ satisfies only the quantum chaos property with parameters $\left(\gamma^\ast,m,\eta,R\right)$. Assume further that $Q=1$. By Definition~\ref{def:qcp},
\begin{equation}
    \begin{aligned}
        \mathbb{P}_{\bm{X}\sim\mathbb{P}_{T,Q}}\left[\mathcal{S}\left(\gamma-\delta,m,1,\eta,\left\{\bm{\tau}_1\right\},R\right)\neq\varnothing\right]&=\mathbb{P}_{\bm{X}\sim\mathbb{P}_{T,Q}}\left[\mathcal{S}\left(\gamma-\delta,m,1,\eta,\left\{\bm{1}\right\},R\right)\neq\varnothing\right]\\
        &\leq\exp\left(-\operatorname{\Omega}\left(n\right)\right)
    \end{aligned}
\end{equation}
when $\gamma-\delta>\gamma^\ast$. This yields a contradiction for sufficiently small $\epsilon>0$---that is, sufficiently large constant $b$ in Proposition~\ref{prop:max_deg_rand_hyp}---and sufficiently large $n$, completing the proof of Theorem~\ref{thm:m_qogp_implies_alg_hardness}.

\section{Quantum Spin Glasses Exhibit the Quantum Overlap Gap Property}\label{sec:qogp_is_exhibited}

We here demonstrate that the quantum $k$-spin model (Eq.~\eqref{eq:kloc_mod}):
\begin{equation}
    \bm{H}_{k\mathrm{-spin}}=\frac{1}{\sqrt{p\binom{n}{k}}}\sum_{\overline{i}\in\binom{\left[n\right]}{k}}\sum_{\bm{b}\in\left\{1,2,3\right\}^{\times k}}S_{\overline{i},\bm{b}}J_{\overline{i},\bm{b}}\prod_{j=1}^k\bm{\sigma}_{i_j}^{\left(b_j\right)}
\end{equation}
satisfies the quantum chaos property, and that the $\left(\mathcal{P},k\right)$-quantum spin glass model (Eq.~\eqref{eq:pk_mod}):
\begin{equation}
    \bm{H}_{\left(\mathcal{P},k\right)\mathrm{-s.g.}}:=\frac{1}{\sqrt{\left\lvert\mathcal{P}\right\rvert p\binom{n}{k}}}\sum_{\bm{b}\in\mathcal{P}}\sum_{\overline{i}\in\binom{\left[n\right]}{k}}S_{\bm{b},\overline{i}}J_{\bm{b},\overline{i}}\prod_{j=1}^k\bm{\sigma}_{i_j}^{\left(b_{i_j}\right)}=:\frac{1}{\sqrt{\left\lvert\mathcal{P}\right\rvert}}\sum_{\bm{b}\in\mathcal{P}}\bm{H}_{\bm{b},\bm{J}}
\end{equation}
satisfies the $m$-QOGP whenever distinct $\bm{b}\neq\bm{b'}\in\mathcal{P}$ differ in at least some fraction of their entries. These two results were previously stated as Theorems~\ref{thm:k_qsg_qcp} and~\ref{thm:pk_qsg_ogp}. \change{Demonstrating either the quantum chaos property or the $m$-QOGP relies on demonstrating the w.h.p. emptiness of the set $\mathcal{S}\left(\gamma,m,\xi,\eta,\mathcal{I}\right)$ in some parameter regime, and refer the reader to Definition~\ref{def:s_def} to recall the definition of this set.}

\subsection{Efficient Local Shadows Estimators}

As the quantum chaos property and $m$-QOGP are defined with respect to an efficient local shadows estimator, we first show that such estimators exist for both the $k$- and $\left(\mathcal{P},k\right)$-quantum spin glass models. In what follows, we recall the space $\mathcal{B}_6$ of $n$-$d$it strings with $d=6$ that are classical representations of Pauli basis states. We use the notation $\ket{\bm{b};\bm{s}}$ to represent elements of $\mathcal{B}_6$, where $\bm{b}\in\left\{1,2,3\right\}^{\times n}$ labels an $n$-qubit Pauli operator and $\bm{s}\in\left\{0,1\right\}^{\times n}$ the eigenstates of the operator labeled by $\bm{b}$. As we will only be interested in expectation values of observables in states $\ket{\bm{b};\bm{s}}\in\mathcal{B}_6$, we will often abuse notation and write an expectation value as:
\begin{equation}
    \bra{\bm{b};\bm{s}}\bm{O}\ket{\bm{b};\bm{s}}
\end{equation}
for $\bm{O}\in\mathbb{C}^{2^n\times 2^n}$; this should be understood as an expectation value of $\bm{O}$ in the Pauli basis state $\ket{\psi}\in\mathbb{C}^{2^n}$ labeled by $\ket{\bm{b};\bm{s}}$. We will similarly ``equate'' operators in the $d$it representation with operators in the qubit representation $\mathbb{C}^{2^n\times 2^n}$, and this should be considered as equating expectation values of the two under this correspondence.

The classical shadows estimators we consider here are variants of the \emph{Pauli shadows} framework~\cite{huang2020predicting,PhysRevLett.127.030503}. We review these algorithms in detail in Appendix~\ref{sec:classical_shadows}, but will echo the results required for our purposes now.

We show in Proposition~\ref{prop:pauli_shadows_params_k_loc_spin} of Appendix~\ref{sec:classical_shadows} that $\bm{H}_{k\mathrm{-spin}}$ has, for any choice of $\delta>0$, an $\left(\delta,p_{\mathrm{est}},p_{\mathrm{b}}\right)$-efficient local shadows estimator with:
\begin{equation}
    p_{\mathrm{est}}=\frac{1}{1+0.99\times k^{-2}3^{-k}\delta^2}
\end{equation}
and
\begin{equation}
    p_{\mathrm{b}}=\exp\left(-\operatorname{\Omega}\left(n\right)\right)
\end{equation}
given by the Pauli shadows estimator~\cite{huang2020predicting}. The associated basis $\mathcal{B}$ is the full space of classical representations of Pauli basis states $\mathcal{B}_6$, and $\bm{\mathcal{R}}\left(\bm{H}_{\bm{J}}\right)$ a simple rescaling:
\begin{equation}\label{eq:k_qsg_r_def_qcp_proof}
    \bm{\mathcal{R}}\left(\bm{H}_{k\mathrm{-spin}}\right)=3^k\bm{H}_{k\mathrm{-spin}}.
\end{equation}

Similarly, we show in Proposition~\ref{prop:derand_pauli_shadows_params} of Appendix~\ref{sec:classical_shadows} that $\bm{H}_{\left(\mathcal{P},k\right)\mathrm{-s.g.}}$ has, for any choice of $\delta>0$, an $\left(\delta,p_{\mathrm{est}},p_{\mathrm{b}}\right)$-efficient local shadows estimator with:
\begin{equation}
    p_{\mathrm{est}}=\frac{1}{1+0.99\left\lvert\mathcal{P}\right\rvert^{-1}\delta^2}
\end{equation}
and
\begin{equation}
    p_{\mathrm{b}}=\exp\left(-\operatorname{\Omega}\left(n\right)\right)
\end{equation}
given by the derandomized Pauli shadows estimator~\cite{PhysRevLett.127.030503}. This estimator has associated linear operator in our setting:
\begin{equation}\label{eq:pk_qsg_r_def_mqogp_proof}
    \bm{\mathcal{R}}\left(\bm{H}_{\left(\mathcal{P},k\right)\mathrm{-s.g.}}\right)=\left\lvert\mathcal{P}\right\rvert\bm{H}_{\left(\mathcal{P},k\right)\mathrm{-s.g.}}
\end{equation}

\subsection{Preliminaries}

Our proof strategy for both theorems is to use the \emph{first moment method}; more specifically, we do two things:
\begin{enumerate}
    \item We bound the number of states in $\mathcal{B}$ satisfying the Wasserstein distance constraints (Eq.~\eqref{eq:q_w_1_bound}).
    \item We bound the probability that one of these states has high energy (Eq.~\eqref{eq:high_energy_req_s}) given the Wasserstein distance constraints (Eq.~\eqref{eq:q_w_1_bound}) and the structural constraints on the $\bm{\tau}\in\mathcal{I}$ (Eq.~\eqref{eq:q_set_def_less_than_n}).
\end{enumerate}
In particular, we may define a set of $m\times R$-tuples of quantum states $\mathcal{F}\left(m,\xi,\eta,R\right)\subset\mathcal{B}^{\times\left(m\times R\right)}$ satisfying the quantum $W_1$ distance constraints (Eq.~\eqref{eq:q_w_1_bound}). We are interested in bounding the probability that the random variable
\begin{equation}
    \begin{aligned}
        M:=&\left\lvert \mathcal{S}\left(\gamma,m,\xi,\eta,\mathcal{I},R\right)\right\rvert=\\
        &\sum_{\left(\ket{\psi^{\left(t\right),\left(r\right)}}\right)_{t\in\left[m\right],r\in\left[R\right]}\in\mathcal{F}\left(m,\xi,\eta,R\right)}\bm{1}\left\{\exists\bm{\tau}\in\mathcal{I}^{\otimes m}:\min_{1\leq t\leq m}\max_{r\in\left[R\right]}\bra{\psi^{\left(t\right),\left(r\right)}}\bm{\mathcal{R}}^{\left(t\right)}\left(\bm{\tau}_t\right)\ket{\psi^{\left(t\right),\left(r\right)}}\geq\gamma E^\ast\sqrt{n}\right\}
    \end{aligned}
\end{equation}
is greater than $0$. By Markov's inequality this is bounded by the first moment:
\begin{equation}
    \mathbb{P}\left[M\geq 1\right]\leq\mathbb{E}\left[M\right],
\end{equation}
which by the union bound is in turn bounded by:
\begin{equation}\label{eq:m_exp}
    \begin{aligned}
        \mathbb{E}\left[M\right]&\leq\left\lvert\mathcal{F}\left(m,\xi,\eta,R\right)\right\rvert\left\lvert\mathcal{I}\right\rvert^m\max_{\left(\ket{\psi^{\left(t\right),\left(r\right)}}\right)_{t\in\left[m\right],r\in\left[R\right]}\in\mathcal{F}\left(m,\xi,\eta,R\right)}\mathbb{P}\left[\min_{1\leq t\leq m}\max_{r\in\left[R\right]}\bra{\psi^{\left(t\right),\left(r\right)}}\bm{\mathcal{R}}^{\left(t\right)}\left(\bm{\tau}_t\right)\ket{\psi^{\left(t\right),\left(r\right)}}\geq\gamma E^\ast\sqrt{n}\right]\\
        &\leq\left\lvert\mathcal{F}\left(m,\xi,\eta,R\right)\right\rvert\left\lvert\mathcal{I}\right\rvert^m R^m\max_{\substack{\left(\ket{\psi^{\left(t\right),\left(r\right)}}\right)_{t\in\left[m\right],r\in\left[R\right]}\in\mathcal{F}\left(m,\xi,\eta,R\right)\\\bm{r}\in\left[R\right]^{\times m}}}\mathbb{P}\left[\min_{1\leq t\leq m}\bra{\psi^{\left(t\right),\left(r_t\right)}}\bm{\mathcal{R}}^{\left(t\right)}\left(\bm{\tau}_t\right)\ket{\psi^{\left(t\right),\left(r_t\right)}}\geq\gamma E^\ast\sqrt{n}\right].
    \end{aligned}
\end{equation}
Recall that $\left\lvert\mathcal{I}\right\rvert\leq 2^{cn}$ for some $c>0$ by assumption (or that $c=0$ when considering the quantum chaos property). Similarly, $R^m=\operatorname{O}\left(1\right)$ by assumption. That only leaves bounding from above both $\left\lvert\mathcal{F}\left(m,\xi,\eta,R\right)\right\rvert$ and the probability term.

We first bound the probability that the Bernoulli randomness $\bm{S}$ is far from its mean.
\begin{lemma}[Concentration of $\bm{S}$]
    Let $\left\{\bm{H}_i\right\}_{i=1}^D$ be a set of operators where $\left\lVert\bm{H}_i\right\rVert_{\mathrm{op}}\leq 1$ for each $i\in\left[D\right]$. For every $t\neq t'\in\left[m\right]$, consider:
    \begin{equation}
        V_{t,t'}^{\left(0\right)}:=\frac{1}{pD}\sum_{i=1}^D\left(1-\tau_{t,i}\right)\left(1-\tau_{t',i}\right)S_i\bra{\psi^{\left(t\right)}}\bm{H}_i\ket{\psi^{\left(t\right)}}\bra{\psi^{\left(t'\right)}}\bm{H}_i\ket{\psi^{\left(t'\right)}};
    \end{equation}
    for every $t\in\left[m\right]$,
    \begin{equation}
        V_t^{\left(0\right)}:=\frac{1}{pD}\sum_{i=1}^D\left(1-\tau_{t,i}\right)S_i\bra{\psi^{\left(t\right)}}\bm{H}_i\ket{\psi^{\left(t\right)}}^2;
    \end{equation}
    and for every $t\in\left[m\right]$,
    \begin{equation}
        V_t^{\left(t\right)}:=\frac{1}{pD}\sum_{i=1}^D\tau_{t,i}S_i\bra{\psi^{\left(t\right)}}\bm{H}_i\ket{\psi^{\left(t\right)}}^2.
    \end{equation}
    Define the event:
    \begin{equation}
        \begin{aligned}
            \mathcal{V}:=&\bigcap_{t\neq t'\in\left[m\right]}\left\{\left\lvert V_{t,t'}^{\left(0\right)}-\mathbb{E}\left[V_{t,t'}^{\left(0\right)}\right\rvert\right]\leq\frac{1}{n^{0.49}}\right\}\\
            &\cap\bigcap_{t=1}^m\left\{\left\lvert V_t^{\left(0\right)}-\mathbb{E}\left[V_t^{\left(0\right)}\right]\right\rvert\leq\frac{1}{n^{0.49}}\right\}\\
            &\cap\bigcap_{t=1}^m\left\{\left\lvert V_t^{\left(t\right)}-\mathbb{E}\left[V_t^{\left(t\right)}\right]\right\rvert\leq\frac{1}{n^{0.49}}\right\}.
        \end{aligned}
    \end{equation}
    We have:
    \begin{equation}
        \mathbb{P}\left[\mathcal{V}\right]\geq 1-2m\left(m+1\right)\exp\left(-\frac{p^2 D^2}{2n^{0.98}}\right)=1-\exp\left(-\operatorname{\Omega}\left(n^{1.02}\right)\right).
    \end{equation}
\end{lemma}
\begin{proof}
    This follows immediately from Hoeffding's inequality and the union bound, noting that each term in the sums defining $V_{t,t'}^{\left(0\right)}$, $V_t^{\left(0\right)}$, and $V_t^{\left(t\right)}$ is bounded between $-\frac{1}{pD}$ and $\frac{1}{pD}$ as $\left\lVert\bm{H}_i\right\rVert_{\mathrm{op}}\leq 1$. The asymptotic scaling follows from the assumption that $p\geq\operatorname{\Omega}\left(\frac{n}{D}\right)$.
\end{proof}
Using this fact, we need only bound:
\begin{equation}
    \begin{aligned}
        \mathbb{E}\left[M\right]\leq\left\lvert\mathcal{F}\left(m,\xi,\eta,R\right)\right\rvert&\left\lvert\mathcal{I}\right\rvert^m R^m\left(\exp\left(-\operatorname{\Omega}\left(n^{1.02}\right)\right)\right.\\
        &\left.+\max_{\substack{\left(\ket{\psi^{\left(t\right),\left(r\right)}}\right)_{t\in\left[m\right],r\in\left[R\right]}\in\mathcal{F}\left(m,\xi,\eta,R\right)\\\bm{r}\in\left[R\right]^{\times m}}}\mathbb{P}\left[\min_{1\leq t\leq m}\bra{\psi^{\left(t\right),\left(r_t\right)}}\bm{\mathcal{R}}^{\left(t\right)}\left(\bm{\tau}_t\right)\ket{\psi^{\left(t\right),\left(r_t\right)}}\geq\gamma E^\ast\sqrt{n}\mid\mathcal{V}\right]\right);
    \end{aligned}
\end{equation}
that is, we can get away with conditioning on $\bm{S}$ satisfying $\mathcal{V}$ up to a superexponentially small additive error in the probability term.

We now proceed to bound $\left\lvert\mathcal{F}\left(m,\xi,\eta,R\right)\right\rvert$ and $\mathbb{P}\left[\min_{1\leq t\leq m}\bra{\psi^{\left(t\right),\left(r_t\right)}}\bm{\mathcal{R}}^{\left(t\right)}\left(\bm{\tau}_t\right)\ket{\psi^{\left(t\right),\left(r_t\right)}}\geq\gamma E^\ast\sqrt{n}\mid\mathcal{V}\right]$ for the quantum $k$-spin model and the $\left(\mathcal{P},k\right)$-quantum spin glass, thereby proving Theorems~\ref{thm:k_qsg_qcp} and~\ref{thm:pk_qsg_ogp}, respectively.

\subsection{Proof for the \texorpdfstring{$k$}{k}-Local Quantum Spin Glass (Theorem~\ref{thm:k_qsg_qcp})}

\subsubsection{Bounding \texorpdfstring{$\left\lvert\mathcal{F}\left(m,1,\eta,R\right)\right\rvert$}{|F(m,1,eta,R)|}}

We begin by bounding the cardinality of $\mathcal{F}\left(m,1,\eta,R\right)$.
\begin{lemma}[$\left\lvert\mathcal{F}\left(m,1,\eta,R\right)\right\rvert$ bound, quantum $k$-spin model]\label{lem:f_card_bound_k_qsg}
    Let $\operatorname{H}$ be the binary entropy function. Then:
    \begin{equation}
        \left\lvert\mathcal{F}\left(m,1,\eta,R\right)\right\rvert\leq\exp_2\left(\log_2\left(6\right)Rn+\left(\operatorname{H}\left(\frac{\eta}{2}\right)+\log_2\left(5\right)\frac{\eta}{2}\right)\left(m-1\right)Rn+\operatorname{O}\left(\log\left(n\right)\right)\right).
    \end{equation}
\end{lemma}
\begin{proof}
    $\left\lvert\mathcal{B}_6^{\otimes R}\right\rvert=6^{Rn}$; consider any one of these states $\ket{\psi}=\bigotimes_{r=1}^R\ket{\bm{b}^{\left(r\right)};\bm{s}^{\left(r\right)}}$. We now upper bound the number of states $\ket{\psi'}\in\mathcal{B}^{\otimes R}$ within a Hamming distance of $\frac{\eta}{2}Rn$ from $\ket{\psi}$. This is upper-bounded by:
    \begin{equation}
        \sum_{\substack{\rho\leq\frac{\eta}{2}\\\rho Rn\in\mathbb{N}}}5^{\rho Rn}\binom{Rn}{\rho Rn}\leq n^{\operatorname{O}\left(1\right)}5^{\frac{\eta}{2}Rn}\binom{Rn}{\frac{\eta}{2}Rn}.
    \end{equation}
    Repeating this $m-1$ times and then applying Stirling's approximation then gives the desired result.
\end{proof}

\subsubsection{Bounding the Probability Term}

As $\left(\bra{\psi^{\left(t\right),\left(r_t\right)}}\bm{\mathcal{R}}^{\left(t\right),\left(r_t\right)}\left(\bm{\tau}_t\right)\ket{\psi^{\left(t\right),\left(r_t\right)}}\right)_{t=1}^m$ is an $m$-dimensional random Gaussian vector when conditioned on $\bm{S}$, we proceed via a standard tail bound. We begin by recalling the variance of Pauli basis states in the quantum $k$-spin model~\cite[Lemma~7]{anschuetz2024boundsgroundstateenergy} (with the factor of $9^k$ from Eq.~\eqref{eq:k_qsg_r_def_qcp_proof}):
\begin{equation}
    \begin{aligned}
        \mathbb{E}_{\bm{J}}&\left[\bra{\bm{b}^{\left(t\right),\left(r_t\right)};\bm{s}^{\left(t\right),\left(r_t\right)}}\bm{\mathcal{R}}^{\left(t\right)}\left(\bm{\tau}_t\right)\ket{\bm{b}^{\left(t\right),\left(r_t\right)};\bm{s}^{\left(t\right),\left(r_t\right)}}^2\mid\mathcal{V}\right]\\
        &=\frac{9^k}{p\binom{n}{k}}\sum_{\overline{i}\in\binom{\left[n\right]}{k}}\sum_{\bm{c}\in\left\{1,2,3\right\}^k} S_{\overline{i},\bm{c}}\bra{\bm{b}^{\left(t\right),\left(r_t\right)};\bm{s}^{\left(t\right),\left(r_t\right)}}\prod_{j=1}^k\bm{\sigma}_{i_j}^{\left(c_j\right)}\ket{\bm{b}^{\left(t\right),\left(r_t\right)};\bm{s}^{\left(t\right),\left(r_t\right)}}^2\\
        &=9^k+\operatorname{O}\left(n^{-0.49}\right),
    \end{aligned}
\end{equation}
where the final line follows by the conditioning on $\mathcal{V}$. Similarly, the covariance for $t\neq t'$ is (recalling that here we are only interested in the case when all $\bm{\mathcal{R}}^{\left(t\right)}\left(\bm{\tau}_t\right)$ are independent):
\begin{equation}
    \mathbb{E}\left[\bra{\bm{b}^{\left(t\right),\left(r_t\right)};\bm{s}^{\left(t\right),\left(r_t\right)}}\bm{\mathcal{R}}^{\left(t\right)}\left(\bm{\tau}_t\right)\ket{\bm{b}^{\left(t\right),\left(r_t\right)};\bm{s}^{\left(t\right),\left(r_t\right)}}\bra{\bm{b}^{\left(t'\right),\left(r_{t'}\right)};\bm{s}^{\left(t'\right),\left(r_{t'}\right)}}\bm{\mathcal{R}}^{\left(t'\right)}\left(\bm{\tau}_{t'}\right)\ket{\bm{b}^{\left(t'\right),\left(r_{t'}\right)};\bm{s}^{\left(t'\right),\left(r_{t'}\right)}}\mid\mathcal{V}\right]=0.
\end{equation}
Standard Gaussian tail bounds (e.g., \revref\cite[Eq.~(1.5)]{hashorva2003multivariate}) then give:
\begin{equation}\label{eq:prob_corr_en_bound_k_qsg}
    \mathbb{P}\left[\min_{1\leq t\leq m}\bra{\psi^{\left(t\right),\left(r_t\right)}}\bm{\mathcal{R}}^{\left(t\right)}\left(\bm{\tau}_t\right)\ket{\psi^{\left(t\right),\left(r_t\right)}}\geq\gamma E^\ast\sqrt{n}\mid\mathcal{V}\right]\leq\exp\left(-\frac{m\gamma^2 E^{\ast 2}n}{2\times 9^k}+\operatorname{O}\left(n^{0.51}\right)\right).
\end{equation}

\subsubsection{Concluding the Proof}

Considering Eq.~\eqref{eq:prob_corr_en_bound_k_qsg} with Lemma~\ref{lem:f_card_bound_k_qsg}, we have:
\begin{equation}\label{eq:m_upper_bound_k_qsg}
    \mathbb{E}\left[M\right]\leq\exp_2\left(\log_2\left(6\right)Rn+\left(\operatorname{H}\left(\frac{\eta}{2}\right)+\log_2\left(5\right)\frac{\eta}{2}\right)\left(m-1\right)Rn-\frac{m\gamma^2 E^{\ast 2}}{2\ln\left(2\right)9^k}n+\operatorname{O}\left(n^{0.51}\right)\right).
\end{equation}
Thus, if we show the existence of a parameter regime where
\begin{equation}\label{eq:psi_qcp_def}
    \varPsi\left(\gamma,m,\eta,R\right):=\log_2\left(6\right)+\left(\operatorname{H}\left(\frac{\eta}{2}\right)+\log_2\left(5\right)\frac{\eta}{2}\right)\left(m-1\right)-\frac{m\gamma^2E^{\ast 2}}{2\ln\left(2\right)9^k R}<0,
\end{equation}
we will have proven Theorem~\ref{thm:k_qsg_qcp}.

First, recall the general upper bound:
\begin{equation}
    \operatorname{H}\left(\frac{\eta}{2}\right)\leq\sqrt{2\left(\frac{\eta}{2}\right)\left(1-\frac{\eta}{2}\right)}\leq\sqrt{\eta}.
\end{equation}
Recall our assumptions on the parameters:
\begin{align}
    m&\geq 1+\frac{6\ln\left(6\right)}{\gamma^2 E^{\ast 2}}9^k R,\\
    \eta&\leq\min\left(1,\left(\frac{\gamma^2 E^{\ast 2}}{6\ln\left(2\right)9^k R}\right)^2,\frac{\gamma^2 E^{\ast 2}}{3\ln\left(5\right)9^k R}\right).
\end{align}
Under these assumptions, we have:
\begin{equation}
    \begin{aligned}
        \varPsi\left(\gamma,m,\eta,R\right)\leq&\log_2\left(6\right)+\left(m-1\right)\frac{\gamma^2 E^{\ast 2}}{6\ln\left(2\right)9^k R}+\left(m-1\right)\log_2\left(5\right)\frac{\gamma^2 E^{\ast 2}}{6\ln\left(5\right)9^k R}\\
        &-\frac{\gamma^2E^{\ast 2}}{2\ln\left(2\right)9^k R}-\left(m-1\right)\frac{\gamma^2E^{\ast 2}}{2\ln\left(2\right)9^k R}\\
        =&\log_2\left(6\right)-\frac{\gamma^2E^{\ast 2}}{2\ln\left(2\right)9^k R}-\left(m-1\right)\frac{\gamma^2E^{\ast 2}}{6\ln\left(2\right)9^k R}\\
        <&0,
    \end{aligned}
\end{equation}
proving the desired result.

\subsection{Proof for the \texorpdfstring{$\left(\mathcal{P},k\right)$}{(P,k)}-Quantum Spin Glass (Theorem~\ref{thm:pk_qsg_ogp})}

\subsubsection{Bounding \texorpdfstring{$\left\lvert\mathcal{F}\left(m,\xi,\eta,R\right)\right\rvert$}{|F(m,xi,eta,R)|}}

We begin by bounding the cardinality of $\mathcal{F}\left(m,\xi,\eta,R\right)$.
\begin{lemma}[$\left\lvert\mathcal{F}\left(m,\xi,\eta,R\right)\right\rvert$ bound, $\left(\mathcal{P},k\right)$-Quantum Spin Glass]\label{lem:f_card_bound}
    Let $\operatorname{H}$ be the binary entropy function. Then:
    \begin{equation}
        \left\lvert\mathcal{F}\left(m,\xi,\eta,R\right)\right\rvert\leq\exp_2\left(\log_2\left(6\right)Rn+\left(\operatorname{H}\left(\frac{1-\xi+\eta}{2}\right)+\log_2\left(5\right)\frac{1-\xi+\eta}{2}\right)\left(m-1\right)Rn+\operatorname{O}\left(\log\left(n\right)\right)\right).
    \end{equation}
\end{lemma}
\begin{proof}
    $\left\lvert\mathcal{B}_6^{\otimes R}\right\rvert=6^{Rn}$; consider any one of these states $\ket{\psi}=\bigotimes_{r=1}^R\ket{\bm{b}^{\left(r\right)};\bm{s}^{\left(r\right)}}$. We now upper bound the number of states $\ket{\psi'}\in\mathcal{B}^{\otimes R}$ within a Hamming distance of $\frac{1-\xi+\eta}{2}Rn$ from $\ket{\psi}$. This is upper-bounded by:
    \begin{equation}
        \sum_{\substack{\rho\leq\frac{1-\xi+\eta}{2}\\\rho Rn\in\mathbb{N}}}5^{\rho Rn}\binom{Rn}{\rho Rn}\leq n^{\operatorname{O}\left(1\right)}5^{\frac{1-\xi+\eta}{2}Rn}\binom{Rn}{\frac{1-\xi+\eta}{2}Rn}.
    \end{equation}
    Repeating this $m-1$ times and then applying Stirling's approximation then gives the desired result.
\end{proof}

\subsubsection{Bounding the Probability Term}

As $\left(\bra{\psi^{\left(t\right),\left(r_t\right)}}\bm{\mathcal{R}}^{\left(t\right)}\left(\bm{\tau}_t\right)\ket{\psi^{\left(t\right),\left(r_t\right)}}\right)_{t=1}^m$ is an $m$-dimensional random Gaussian vector when conditioned on $\bm{S}$, we proceed via a tail bound for correlated multivariate Gaussians. We begin by lower bounding the variance (with the factor of $\left\lvert\mathcal{P}\right\rvert$ from Eq.~\eqref{eq:pk_qsg_r_def_mqogp_proof}):
\begin{equation}
    \begin{aligned}
        \mathbb{E}_{\bm{J}}&\left[\bra{\bm{b}^{\left(t\right),\left(r_t\right)};\bm{s}^{\left(t\right),\left(r_t\right)}}\bm{\mathcal{R}}^{\left(t\right)}\left(\bm{\tau}_t\right)\ket{\bm{b}^{\left(t\right),\left(r_t\right)};\bm{s}^{\left(t\right),\left(r_t\right)}}^2\mid\mathcal{V}\right]\\
        &=\frac{\left\lvert\mathcal{P}\right\rvert}{\left\lvert\mathcal{P}\right\rvert p\binom{n}{k}}\sum_{\bm{c}\in\mathcal{P}}\sum_{\overline{i}\in\binom{\left[n\right]}{k}}\left(\left(1-\tau_{t,\bm{c},\overline{i}}\right)S_{\bm{c},\overline{i}}+\tau_{t,\bm{c},\overline{i}}S_{\bm{c},\overline{i}}\right)\bra{\bm{b}^{\left(t\right),\left(r_t\right)};\bm{s}^{\left(t\right),\left(r_t\right)}}\prod_{j=1}^k\bm{\sigma}_{i_j}^{\left(c_j\right)}\ket{\bm{b}^{\left(t\right),\left(r_t\right)};\bm{s}^{\left(t\right),\left(r_t\right)}}^2\\
        &\geq 1+\operatorname{O}\left(n^{-0.49}\right),
    \end{aligned}
\end{equation}
where the final line follows by the conditioning on $\mathcal{V}$ and as, by the definition of $\mathcal{B}$, $\bm{b}^{\left(t\right),\left(r_t\right)}\in\mathcal{P}$.

We can similarly upper bound the covariance when conditioned on $\mathcal{V}$. Before proceeding, we define two quantities. First, we define the counting function $J_{\mathcal{Q}}:\left\{1,2,3\right\}^{\times n}\times\left\{1,2,3\right\}^{\times n}\times\left\{1,2,3\right\}^{\times n}\to\left[0,1\right]$ for any subset $\mathcal{Q}\subseteq\left[n\right]$:
\begin{equation}
    J_{\mathcal{Q}}\left(\bm{c},\bm{b},\bm{b'}\right)=\frac{1}{n}\sum_{i\in\mathcal{Q}}\delta_{c_i,b_i}\delta_{c_i,b_i'},
\end{equation}
which counts the number of indices on which $\bm{c}$, $\bm{b}$, and $\bm{b'}$ agree on positional indices labeled by $\mathcal{Q}$, normalized by the total number of qubits $n$. Second, we define the generalization of the Hamming distance:
\begin{equation}
    d_{\mathcal{Q},\bm{c}}\left(\left(\bm{b};\bm{s}\right),\left(\bm{b'};\bm{s'}\right)\right):=d_{\operatorname{H}}\left(\bm{\varPi}_{\mathcal{Q},\bm{c},\bm{b},\bm{b'}}\bm{s},\bm{\varPi}_{\mathcal{Q},\bm{c},\bm{b},\bm{b'}}\bm{s'}\right),
\end{equation}
where $\bm{\varPi}_{\mathcal{Q},\bm{c},\bm{b},\bm{b'}}$ is a projector onto indices $i\in\mathcal{Q}$ where $c_i=b_i=b_i'$. More specifically, collecting all $i\in\mathcal{Q}$ for which $c_i=b_i=b_i'$ into a vector $\bm{J}=\left\{i\right\}$, $\bm{\varPi}_{\mathcal{Q},\bm{c},\bm{b},\bm{b'}}$ is of the form:
\begin{equation}
    \left(\bm{\varPi}_{\mathcal{Q},\bm{c},\bm{b},\bm{b'}}\right)_{i,j}=\bm{1}\left\{j=J_i\right\}.
\end{equation}
While these definitions may seem ad hoc, it will allow us to generalize the classical observation:
\begin{equation}
    \begin{aligned}
        \frac{1}{\binom{n}{k}}\sum_{\overline{i}\in\binom{\mathcal{Q}}{k}}\bra{\bm{z}}\prod_{j=1}^k\bm{\sigma}_{i_j}^{\left(3\right)}\ket{\bm{z}}\bra{\bm{z'}}\prod_{j=1}^k\bm{\sigma}_{i_j}^{\left(3\right)}\ket{\bm{z'}}&=\left(\frac{\sum_{i\in\mathcal{Q}}\bra{\bm{z}}\bm{\sigma}_i^{\left(3\right)}\ket{\bm{z}}\bra{\bm{z'}}\bm{\sigma}_i^{\left(3\right)}\ket{\bm{z'}}}{n}\right)^k+\operatorname{O}\left(n^{-1}\right)\\
        &=\left(\frac{\left\lvert\mathcal{Q}\right\rvert-2d_{\mathrm{H}}\left(\left(z_i\right)_{i\in\mathcal{Q}},\left(z_i'\right)_{i\in\mathcal{Q}}\right)}{n}\right)^k+\operatorname{O}\left(n^{-1}\right)
    \end{aligned}
\end{equation}
for computational basis states to Pauli basis states:
\begin{equation}
    \frac{1}{\binom{n}{k}}\sum_{\overline{i}\in\binom{\mathcal{Q}}{k}}\bra{\bm{b};\bm{s}}\prod_{j=1}^k\bm{\sigma}_{i_j}^{\left(c_j\right)}\ket{\bm{b};\bm{s}}\bra{\bm{b'};\bm{s'}}\prod_{j=1}^k\bm{\sigma}_{i_j}^{\left(c_j\right)}\ket{\bm{b'};\bm{s'}}=\left(\frac{J_{\mathcal{Q}}\left(\bm{c},\bm{b},\bm{b'}\right)n-2d_{\mathcal{Q},\bm{c}}\left(\left(\bm{b};\bm{s}\right),\left(\bm{b'};\bm{s'}\right)\right)}{n}\right)^k+\operatorname{O}\left(n^{-1}\right).
\end{equation}

With this definition in hand, we proceed in upper bounding the covariance:
\begin{equation}
    \begin{aligned}
        &\mathbb{E}_{\bm{J}}\left[\bra{\bm{b}^{\left(t\right),\left(r_t\right)};\bm{s}^{\left(t\right),\left(r_t\right)}}\bm{\mathcal{R}}^{\left(t\right)}\left(\bm{\tau}_t\right)\ket{\bm{b}^{\left(t\right),\left(r_t\right)};\bm{s}^{\left(t\right),\left(r_t\right)}}\bra{\bm{b}^{\left(t'\right),\left(r_{t'}\right)};\bm{s}^{\left(t'\right),\left(r_{t'}\right)}}\bm{\mathcal{R}}^{\left(t'\right)}\left(\bm{\tau}_{t'}\right)\ket{\bm{b}^{\left(t'\right),\left(r_{t'}\right)};\bm{s}^{\left(t'\right),\left(r_{t'}\right)}}\mid\mathcal{V}\right]\\
        &=\frac{\left\lvert\mathcal{P}\right\rvert}{\left\lvert\mathcal{P}\right\rvert p\binom{n}{k}}\sum_{\bm{c}\in\mathcal{P}}\sum_{\overline{i}\in\binom{\left[n\right]}{k}}\left(1-\tau_{t,\bm{c},\overline{i}}\right)\left(1-\tau_{t',\bm{c},\overline{i}}\right)S_{\bm{c},\overline{i}}\bra{\bm{b}^{\left(t\right),\left(r_t\right)};\bm{s}^{\left(t\right),\left(r_t\right)}}\prod_{j=1}^k\bm{\sigma}_{i_j}^{\left(c_{i_j}\right)}\ket{\bm{b}^{\left(t\right),\left(r_t\right)};\bm{s}^{\left(t\right),\left(r_t\right)}}\\
        &\times\bra{\bm{b}^{\left(t'\right),\left(r_{t'}\right)};\bm{s}^{\left(t'\right),\left(r_{t'}\right)}}\prod_{j=1}^k\bm{\sigma}_{i_j}^{\left(c_{i_j}\right)}\ket{\bm{b}^{\left(t'\right),\left(r_{t'}\right)};\bm{s}^{\left(t'\right),\left(r_{t'}\right)}}\\
        &=\frac{1}{\binom{n}{k}}\sum_{\bm{c}\in\mathcal{P}}\sum_{\overline{i}\in\binom{\left[n\right]}{k}}\left(1-\tau_{t,\bm{c},\overline{i}}\right)\left(1-\tau_{t',\bm{c},\overline{i}}\right)\\
        &\times\bra{\bm{b}^{\left(t\right),\left(r_t\right)};\bm{s}^{\left(t\right),\left(r_t\right)}}\prod_{j=1}^k\bm{\sigma}_{i_j}^{\left(c_{i_j}\right)}\ket{\bm{b}^{\left(t\right),\left(r_t\right)};\bm{s}^{\left(t\right),\left(r_t\right)}}\bra{\bm{b}^{\left(t'\right),\left(r_{t'}\right)};\bm{s}^{\left(t'\right),\left(r_{t'}\right)}}\prod_{j=1}^k\bm{\sigma}_{i_j}^{\left(c_{i_j}\right)}\ket{\bm{b}^{\left(t'\right),\left(r_{t'}\right)};\bm{s}^{\left(t'\right),\left(r_{t'}\right)}}\\
        &+\operatorname{O}\left(n^{-0.49}\right).
    \end{aligned}
\end{equation}
Now, recall the assumed condition on the $\bm{\tau}_t$ from the definition of a $\left(c,F,R\right)$-correlation set (Definition~\ref{def:cfr_corr_set}): each has associated with it a subset $\mathcal{Q}_{\bm{\tau}_t}\subseteq\left[n\right]$ such that:
\begin{equation}
    \tau_{t,\bm{c},\overline{i}}=0\iff\overline{i}\subseteq\mathcal{Q}_{\bm{\tau}_t}.
\end{equation}
In particular, defining
\begin{equation}
    \mathcal{Q}_{t,t'}:=\mathcal{Q}_{\bm{\tau}_t}\cap\mathcal{Q}_{\bm{\tau}_{t'}},
\end{equation}
we have:
\begin{equation}
    \begin{aligned}
        &\mathbb{E}_{\bm{J}}\left[\bra{\bm{b}^{\left(t\right),\left(r_t\right)};\bm{s}^{\left(t\right),\left(r_t\right)}}\bm{\mathcal{R}}^{\left(t\right)}\left(\bm{\tau}_t\right)\ket{\bm{b}^{\left(t\right),\left(r_t\right)};\bm{s}^{\left(t\right),\left(r_t\right)}}\bra{\bm{b}^{\left(t'\right),\left(r_{t'}\right)};\bm{s}^{\left(t'\right),\left(r_{t'}\right)}}\bm{\mathcal{R}}^{\left(t'\right)}\left(\bm{\tau}_{t'}\right)\ket{\bm{b}^{\left(t'\right),\left(r_{t'}\right)};\bm{s}^{\left(t'\right),\left(r_{t'}\right)}}\mid\mathcal{V}\right]\\
        &=\frac{1}{\binom{n}{k}}\sum_{\bm{c}\in\mathcal{P}}\sum_{\overline{i}\in\binom{\mathcal{Q}_{t,t'}}{k}}\bra{\bm{b}^{\left(t\right),\left(r_t\right)};\bm{s}^{\left(t\right),\left(r_t\right)}}\prod_{j=1}^k\bm{\sigma}_{i_j}^{\left(c_{i_j}\right)}\ket{\bm{b}^{\left(t\right),\left(r_t\right)};\bm{s}^{\left(t\right),\left(r_t\right)}}\\
        &\times\bra{\bm{b}^{\left(t'\right),\left(r_{t'}\right)};\bm{s}^{\left(t'\right),\left(r_{t'}\right)}}\prod_{j=1}^k\bm{\sigma}_{i_j}^{\left(c_{i_j}\right)}\ket{\bm{b}^{\left(t'\right),\left(r_{t'}\right)};\bm{s}^{\left(t'\right),\left(r_{t'}\right)}}+\operatorname{O}\left(n^{-0.49}\right)\\
        &=\sum_{\bm{c}\in\mathcal{P}}\left(J_{\mathcal{Q}_{t,t'}}\left(\bm{c},\bm{b}^{\left(t\right),\left(r_t\right)},\bm{b}^{\left(t'\right),\left(r_{t'}\right)}\right)-\frac{2}{n}d_{\mathcal{Q}_{t,t'},\bm{c}}\left(\left(\bm{b}^{\left(t\right),\left(r_t\right)},\bm{s}^{\left(t\right),\left(r_t\right)}\right),\left(\bm{b}^{\left(t'\right),\left(r_{t'}\right)},\bm{s}^{\left(t'\right),\left(r_{t'}\right)}\right)\right)\right)^k+\operatorname{O}\left(n^{-0.49}\right).
    \end{aligned}
\end{equation}

We now consider two cases: when $\bm{b}^{\left(t\right),\left(r_t\right)}\neq\bm{b}^{\left(t'\right),\left(r_{t'}\right)}$ and when $\bm{b}^{\left(t\right),\left(r_t\right)}=\bm{b}^{\left(t'\right),\left(r_{t'}\right)}$. We begin with the former. As by assumption any two $\bm{b},\bm{b}'\in\mathcal{P}$ are equal in at most $\phi n$ of their elements, we have that always:
\begin{equation}
    J_{\mathcal{Q}_{t,t'}}\left(\bm{c},\bm{b},\bm{b'}\right)\leq\phi.
\end{equation}
In particular,
\begin{equation}
    \begin{aligned}
        \mathbb{E}_{\bm{J}}&\left[\bra{\bm{b}^{\left(t\right),\left(r_t\right)};\bm{s}^{\left(t\right),\left(r_t\right)}}\bm{\mathcal{R}}^{\left(t\right)}\left(\bm{\tau}_t\right)\ket{\bm{b}^{\left(t\right),\left(r_t\right)};\bm{s}^{\left(t\right),\left(r_t\right)}}\bra{\bm{b}^{\left(t'\right),\left(r_{t'}\right)};\bm{s}^{\left(t'\right),\left(r_{t'}\right)}}\bm{\mathcal{R}}^{\left(t'\right)}\left(\bm{\tau}_{t'}\right)\ket{\bm{b}^{\left(t'\right),\left(r_{t'}\right)};\bm{s}^{\left(t'\right),\left(r_{t'}\right)}}\mid\mathcal{V}\right]\\
        &\leq\left\lvert\mathcal{P}\right\rvert\phi^k+\operatorname{O}\left(n^{-0.49}\right).
    \end{aligned}
\end{equation}
In the case $\bm{b}^{\left(t\right),\left(r_t\right)}=\bm{b}^{\left(t'\right),\left(r_{t'}\right)}$, we instead have the upper bound:
\begin{equation}
    \begin{aligned}
        \mathbb{E}_{\bm{J}}&\left[\bra{\bm{b}^{\left(t\right),\left(r_t\right)};\bm{s}^{\left(t\right),\left(r_t\right)}}\bm{\mathcal{R}}^{\left(t\right)}\left(\bm{\tau}_t\right)\ket{\bm{b}^{\left(t\right),\left(r_t\right)};\bm{s}^{\left(t\right),\left(r_t\right)}}\bra{\bm{b}^{\left(t\right),\left(r_t\right)};\bm{s}^{\left(t'\right),\left(r_{t'}\right)}}\bm{\mathcal{R}}^{\left(t'\right)}\left(\bm{\tau}_{t'}\right)\ket{\bm{b}^{\left(t\right),\left(r_t\right)};\bm{s}^{\left(t'\right),\left(r_{t'}\right)}}\mid\mathcal{V}\right]\\
        \leq&\left(J_{\mathcal{Q}_{t,t'}}\left(\bm{b}^{\left(t\right),\left(r_t\right)},\bm{b}^{\left(t\right),\left(r_t\right)},\bm{b}^{\left(t\right),\left(r_t\right)}\right)-\frac{2}{n}d_{\mathcal{Q}_{t,t'},\bm{c}}\left(\left(\bm{b}^{\left(t\right),\left(r_t\right)},\bm{s}^{\left(t\right),\left(r_t\right)}\right),\left(\bm{b}^{\left(t\right),\left(r_t\right)},\bm{s}^{\left(t'\right),\left(r_{t'}\right)}\right)\right)\right)^k\\
        &+\sum_{\bm{c}\neq\bm{b}^{\left(t\right),\left(r_t\right)}\in\mathcal{P}}\left(J_{\mathcal{Q}_{t,t'}}\left(\bm{c},\bm{b}^{\left(t\right),\left(r_t\right)},\bm{b}^{\left(t\right),\left(r_t\right)}\right)-\frac{2}{n}d_{\mathcal{Q}_{t,t'},\bm{c}}\left(\left(\bm{b}^{\left(t\right),\left(r_t\right)},\bm{s}^{\left(t\right),\left(r_t\right)}\right),\left(\bm{b}^{\left(t\right),\left(r_t\right)},\bm{s}^{\left(t'\right),\left(r_{t'}\right)}\right)\right)\right)^k\\
        &+\operatorname{O}\left(n^{-0.49}\right)\\
        \leq&\left(\frac{\left\lvert\mathcal{Q}_{t,t'}\right\rvert}{n}-\frac{2}{n}\left\lVert\Tr_{\mathcal{Q}_{t,t'}^\complement}\left(\ket{\bm{b}^{\left(t\right),\left(r_t\right)};\bm{s}^{\left(t\right),\left(r_t\right)}}\bra{\bm{b}^{\left(t\right),\left(r_t\right)};\bm{s}^{\left(t\right),\left(r_t\right)}}-\ket{\bm{b}^{\left(t'\right),\left(r_{t'}\right)};\bm{s}^{\left(t'\right),\left(r_{t'}\right)}}\bra{\bm{b}^{\left(t'\right),\left(r_{t'}\right)};\bm{s}^{\left(t'\right),\left(r_{t'}\right)}}\right)\right\rVert_{W_1}\right)^k\\
        &+\left\lvert\mathcal{P}\right\rvert\phi^k+\operatorname{O}\left(n^{-0.49}\right).
    \end{aligned}
\end{equation}
We now recall the assumed Wasserstein distance bound (Eq.~\eqref{eq:q_w_1_bound}):
\begin{equation}
    \begin{aligned}
        \frac{1}{Rn}&\sum_{r=1}^R\left\lVert\ket{\bm{b}^{\left(t\right),\left(r\right)};\bm{s}^{\left(t\right),\left(r\right)}}\bra{\bm{b}^{\left(t\right),\left(r\right)};\bm{s}^{\left(t\right),\left(r\right)}}-\ket{\bm{b}^{\left(t'\right),\left(r\right)};\bm{s}^{\left(t'\right),\left(r\right)}}\bra{\bm{b}^{\left(t'\right),\left(r\right)};\bm{s}^{\left(t'\right),\left(r\right)}}\right\rVert_{W_1}\geq\frac{1-\xi}{2}\\
        \implies,&\text{ if }R=1,\\
        &\frac{1}{n}\left\lVert\ket{\bm{b}^{\left(t\right),\left(r_t\right)};\bm{s}^{\left(t\right),\left(r_t\right)}}\bra{\bm{b}^{\left(t\right),\left(r_t\right)};\bm{s}^{\left(t\right),\left(r_t\right)}}-\ket{\bm{b}^{\left(t'\right),\left(r_{t'}\right)};\bm{s}^{\left(t'\right),\left(r_{t'}\right)}}\bra{\bm{b}^{\left(t'\right),\left(r_{t'}\right)};\bm{s}^{\left(t'\right),\left(r_{t'}\right)}}\right\rVert_{W_1}\geq\frac{1-\xi}{2};
    \end{aligned}
\end{equation}
and the assumed bound on the $\left\lvert\mathcal{Q}_{\bm{\tau}}\right\rvert$ when $R\neq 1$ (Eq.~\eqref{eq:q_set_def_less_than_n}):
\begin{equation}
    \left\lvert\mathcal{Q}_{t,t'}\right\rvert=\left\lvert\mathcal{Q}_{\bm{\tau}}\cap\mathcal{Q}_{\bm{\tau'}}\right\rvert\leq\min\left(\left\lvert\mathcal{Q}_{\bm{\tau}}\right\rvert,\left\lvert\mathcal{Q}_{\bm{\tau'}}\right\rvert\right)\leq\left(1-F\right)n.
\end{equation}
The former gives:
\begin{equation}
    \begin{aligned}
        \frac{1}{n}&\left\lVert\Tr_{\mathcal{Q}_{t,t'}^\complement}\left(\ket{\bm{b}^{\left(t\right),\left(r_t\right)};\bm{s}^{\left(t\right),\left(r_t\right)}}\bra{\bm{b}^{\left(t\right),\left(r_t\right)};\bm{s}^{\left(t\right),\left(r_t\right)}}-\ket{\bm{b}^{\left(t'\right),\left(r_{t'}\right)};\bm{s}^{\left(t'\right),\left(r_{t'}\right)}}\bra{\bm{b}^{\left(t'\right),\left(r_{t'}\right)};\bm{s}^{\left(t'\right),\left(r_{t'}\right)}}\right)\right\rVert_{W_1}\\
        &\geq\max\left(0,\frac{1-\xi}{2}\delta_{R,1}-\left(1-\frac{\left\lvert\mathcal{Q}_{t,t'}\right\rvert}{n}\right)\right),
    \end{aligned}
\end{equation}
generally yielding the bound (for sufficiently large $k$ and $n$):
\begin{equation}
    \begin{aligned}
        \mathbb{E}_{\bm{J}}&\left[\bra{\bm{b}^{\left(t\right),\left(r_t\right)};\bm{s}^{\left(t\right),\left(r_t\right)}}\bm{\mathcal{R}}^{\left(t\right)}\left(\bm{\tau}_t\right)\ket{\bm{b}^{\left(t\right),\left(r_t\right)};\bm{s}^{\left(t\right),\left(r_t\right)}}\bra{\bm{b}^{\left(t'\right),\left(r_{t'}\right)};\bm{s}^{\left(t'\right),\left(r_{t'}\right)}}\bm{\mathcal{R}}^{\left(t'\right),\left(r_{t'}\right)}\left(\bm{\tau}_{t'}\right)\ket{\bm{b}^{\left(t'\right),\left(r_{t'}\right)};\bm{s}^{\left(t'\right),\left(r_{t'}\right)}}\mid\mathcal{V}\right]\\
        &\leq\left(\frac{\left\lvert\mathcal{Q}_{t,t'}\right\rvert}{n}\right)^k\left(1-\delta_{R,1}\right)+\min\left(\frac{\left\lvert\mathcal{Q}_{t,t'}\right\rvert}{n},1+\xi-\frac{\left\lvert\mathcal{Q}_{t,t'}\right\rvert}{n}\right)^k\delta_{R,1}+\left\lvert\mathcal{P}\right\rvert\phi^k+\operatorname{O}\left(n^{-0.49}\right)\\
        &\leq\left(1-F\right)^k\left(1-\delta_{R,1}\right)+\left(\frac{1+\xi}{2}\right)^k\delta_{R,1}+\left\lvert\mathcal{P}\right\rvert\phi^k+\operatorname{O}\left(n^{-0.49}\right)\\
        &=:\left(1-\epsilon\right)^k
    \end{aligned}
\end{equation}
for some $\epsilon$ bounded away from $0$ by an $n$-independent constant.

We claim that these bounds on the variances and covariances themselves define a valid multivariate Gaussian distribution. To see this, consider the matrix $\bm{\varSigma}$ with diagonal entries equal to $1$ and off-diagonal entries:
\begin{equation}
    \varSigma_{i,j}=\left(1-\epsilon\right)^k.
\end{equation}
We may rewrite this as:
\begin{equation}
    \bm{\varSigma}=\left(1-\left(1-\epsilon\right)^k\right)\bm{I}_m+\left(1-\epsilon\right)^k\bm{1}\otimes\bm{1}^\intercal,
\end{equation}
where $\bm{1}$ is the $m\times 1$ column vector with each entry equal to $1$. These two terms are mutually diagonalizable, and $\bm{1}\otimes\bm{1}^\intercal$ has a single eigenvalue equal to $m$ and all $m-1$ others equal to $0$. Thus, $\bm{\varSigma}$ is positive definite, and $\bm{\varSigma}$ defines a valid covariance matrix. This observation also immediately gives the determinant:
\begin{equation}
    \det\left(\bm{\varSigma}\right)=\left(1-\left(1-\epsilon\right)^k\right)^{m-1}\left(1+\left(m-1\right)\left(1-\epsilon\right)^k\right)
\end{equation}
and the inverse (which can be seen by taking identical eigenvectors to $\bm{\varSigma}$ and inverting the eigenvalues):
\begin{equation}\label{eq:cov_bound_inv}
    \begin{aligned}
        \bm{\varSigma}^{-1}&=\frac{1}{1-\left(1-\epsilon\right)^k}\bm{I}_m+\left(\frac{1}{1+\left(m-1\right)\left(1-\epsilon\right)^k}-\frac{1}{\left(1-\left(1-\epsilon\right)^k\right)}\right)m^{-1}\bm{1}\otimes\bm{1}^\intercal\\
        &=\frac{1}{1-\left(1-\epsilon\right)^k}\bm{I}_m-\frac{\left(1-\epsilon\right)^k}{\left(1+\left(m-1\right)\left(1-\epsilon\right)^k\right)\left(1-\left(1-\epsilon\right)^k\right)}\bm{1}\otimes\bm{1}^\intercal,
    \end{aligned}
\end{equation}
both of which will be useful later.

We now bound the probability term in Eq.~\eqref{eq:m_exp} using a random Gaussian vector $\bm{Y}\in\mathbb{R}^m$ with covariance matrix $\bm{\varSigma}$. We achieve this through the use of Slepian's lemma~\cite{6768942}:
\begin{lemma}[Slepian's lemma~{\cite[Lemma~1]{6768942}}]
    Let $\bm{X},\bm{Y}\in\mathbb{R}^m$ be centered normal random vectors with:
    \begin{equation}
        \mathbb{E}\left[X_i^2\right]=\mathbb{E}\left[Y_i^2\right]
    \end{equation}
    and
    \begin{equation}
        \mathbb{E}\left[X_i X_j\right]\leq\mathbb{E}\left[Y_i Y_j\right]
    \end{equation}
    for all $i,j\in\left[m\right]$. For any fixed $\bm{x}\in\mathbb{R}^m$,
    \begin{equation}
        \mathbb{P}\left[\bm{X}\geq\bm{x}\right]\leq\mathbb{P}\left[\bm{Y}\geq\bm{x}\right].
    \end{equation}
\end{lemma}
In particular, we have that:
\begin{equation}
    \mathbb{P}\left[\min_{1\leq t\leq m}\bra{\psi^{\left(t\right),\left(r_t\right)}}\bm{\mathcal{R}}^{\left(t\right)}\left(\bm{\tau}_t\right)\ket{\psi^{\left(t\right),\left(r_t\right)}}\geq\gamma E^\ast\sqrt{n}\mid\mathcal{V}\right]\leq\mathbb{P}\left[\bm{Y}\geq\gamma E^\ast\sqrt{n}\bm{1}\right].
\end{equation}

We now use a probability bound for a normal random vector due to \revref\cite{hashorva2003multivariate} to bound the right-hand side.
\begin{proposition}[Normal random vector probability bound~{\cite[Eq.~(1.5)]{hashorva2003multivariate}}]\label{prop:mult_norm_bound}
    Let $\bm{Y}\in\mathbb{R}^m$ be a centered normal random vector with positive definite covariance matrix $\bm{\varSigma}$, and let $\bm{x}\in\mathbb{R}^m$ be fixed. Suppose that $\bm{\varSigma}^{-1}\bm{x}>\bm{0}$, and let $1/\left(\bm{\varSigma}^{-1}\bm{x}\right)$ denote the entry-wise reciprocal. Then, letting $\bm{\hat{e}_i}$ denote the unit vector in the $i$th coordinate,
    \begin{equation}
        \frac{1-\left(1/\left(\bm{\varSigma}^{-1}\bm{x}\right)\right)^\intercal\bm{\varSigma}\left(1/\left(\bm{\varSigma}^{-1}\bm{x}\right)\right)}{\prod_{i=1}^m\bm{\hat{e}_i}^\intercal\bm{\varSigma}^{-1}\bm{x}}\leq\det\left(2\cpi\bm{\varSigma}\right)^{\frac{1}{2}}\exp\left(\frac{1}{2}\bm{x}^\intercal\bm{\varSigma}^{-1}\bm{x}\right)\mathbb{P}\left[\bm{Y}\geq\bm{x}\right]\leq\frac{1}{\prod_{i=1}^m\bm{\hat{e}_i}^\intercal\bm{\varSigma}^{-1}\bm{x}}.
    \end{equation}
\end{proposition}
Using Eq.~\eqref{eq:cov_bound_inv}, we calculate:
\begin{equation}
    \begin{aligned}
        \bm{\varSigma}^{-1}\bm{1}&=\left(\frac{1}{1-\left(1-\epsilon\right)^k}-\frac{m\left(1-\epsilon\right)^k}{\left(1+\left(m-1\right)\left(1-\epsilon\right)^k\right)\left(1-\left(1-\epsilon\right)^k\right)}\right)\bm{1}\\
        &=\frac{1}{1+\left(m-1\right)\left(1-\epsilon\right)^k}\bm{1}\\
        &>\bm{0}.
    \end{aligned}
\end{equation}
We also calculate using Eq.~\eqref{eq:cov_bound_inv}:
\begin{equation}
    \bm{1}^\intercal\bm{\varSigma}^{-1}\bm{1}=\frac{m}{1+\left(m-1\right)\left(1-\epsilon\right)^k}
\end{equation}
and, for $\bm{x}=\gamma E^\ast\sqrt{n}\bm{1}$,
\begin{equation}
    \frac{1}{\prod_{i=1}^m\bm{\hat{e}_i}^\intercal\bm{\varSigma}^{-1}\bm{x}}\leq\frac{1}{\operatorname{\Theta}\left(n^{\frac{m}{2}}\right)}<\operatorname{O}\left(1\right).
\end{equation}
We thus have from Proposition~\ref{prop:mult_norm_bound} that:
\begin{equation}\label{eq:prob_corr_en_bound}
    \begin{aligned}
        \mathbb{P}\left[\min_{1\leq t\leq m}\bra{\psi^{\left(t\right),\left(r_t\right)}}\bm{\mathcal{R}}^{\left(t\right)}\left(\bm{\tau}_t\right)\ket{\psi^{\left(t\right),\left(r_t\right)}}\geq\gamma E^\ast\sqrt{n}\mid\mathcal{V}\right]&\leq\mathbb{P}\left[\bm{Y}\geq\gamma E^\ast\sqrt{n}\bm{1}\right]\\
        &\leq\operatorname{O}\left(1\right)\exp\left(-\frac{m\gamma^2 E^{\ast 2}n}{2\left(1+\left(m-1\right)\left(1-\epsilon\right)^k\right)}\right).
    \end{aligned}
\end{equation}

\subsubsection{Concluding the Proof}

Let
\begin{equation}
    \upsilon:=\frac{1+\xi}{2}\delta_{R,1}+\left(1-F\right)\left(1-\delta_{R,1}\right),
\end{equation}
which we recall is bounded away from $1$. Considering Eq.~\eqref{eq:prob_corr_en_bound} with Lemma~\ref{lem:f_card_bound} and substituting back in the definition of $\epsilon$, we have:
\begin{equation}\label{eq:m_upper_bound}
    \begin{aligned}
        \mathbb{E}\left[M\right]\leq&\exp_2\left(\log_2\left(6\right)Rn+\left(\operatorname{H}\left(\frac{1-\xi+\eta}{2}\right)+\log_2\left(5\right)\frac{1-\xi+\eta}{2}\right)\left(m-1\right)Rn\right.\\
        &\left.-\frac{m\gamma^2 E^{\ast 2}}{2\ln\left(2\right)\left(1+\left(m-1\right)\left(\upsilon^k+\left\lvert\mathcal{P}\right\rvert\phi^k\right)\right)}n+cmn+\operatorname{O}\left(n^{0.51}\right)\right).
    \end{aligned}
\end{equation}
That is, if we show the existence of parameter regimes where
\begin{equation}
    \begin{aligned}
        \varPsi\left(\gamma,m,\xi,\eta,c,\upsilon,R\right):=&\log_2\left(6\right)+\left(\operatorname{H}\left(\frac{1-\xi+\eta}{2}\right)+\log_2\left(5\right)\frac{1-\xi+\eta}{2}\right)\left(m-1\right)+\frac{cm}{R}\\
        &-\frac{m\gamma^2E^{\ast 2}}{2\ln\left(2\right)\left(1+\left(m-1\right)\left(\upsilon^k+\left\lvert\mathcal{P}\right\rvert\phi^k\right)\right)R}<0,
    \end{aligned}
\end{equation}
and similarly we demonstrate that the quantum chaos property is satisfied, we will have proven Theorem~\ref{thm:pk_qsg_ogp}.

Just as in the proof of Theorem~\ref{thm:k_qsg_qcp}, we begin by recalling the upper bound:
\begin{equation}
    \operatorname{H}\left(\frac{1-\xi+\eta}{2}\right)\leq\sqrt{2\left(\frac{1-\xi+\eta}{2}\right)\left(\frac{1+\xi-\eta}{2}\right)}\leq\sqrt{1-\xi+\eta}.
\end{equation}
Recall our assumptions on the parameters:
\begin{align}
    \xi-\eta&\geq\max\left(0,1-\left(\frac{\gamma^2 E^{\ast 2}}{24\ln\left(2\right)R}\right)^2,1-\frac{\gamma^2 E^{\ast 2}}{12\ln\left(5\right)R}\right),\label{eq:xi_eta_rel_mqogp}\\
    1+\frac{8\ln\left(6\right)}{\gamma^2 E^{\ast 2}}R\leq m&\leq 1+\frac{1}{\upsilon^k+\left\lvert\mathcal{P}\right\rvert\phi^k},\\
    c&\leq\frac{1}{48}.
\end{align}
Under these assumptions we have:
\begin{equation}
    \begin{aligned}
        \varPsi\left(\gamma,m,\xi,\eta,c,\upsilon,R\right)\leq&\log_2\left(6\right)+\left(m-1\right)\frac{\gamma^2 E^{\ast 2}}{24\ln\left(2\right)R}+\left(m-1\right)\log_2\left(5\right)\frac{\gamma^2 E^{\ast 2}}{24\ln\left(5\right)R}+\left(m-1\right)\frac{m}{48\left(m-1\right)R}\\
        &-\frac{\gamma^2E^{\ast 2}}{4\ln\left(2\right)R}-\left(m-1\right)\frac{\gamma^2E^{\ast 2}}{4\ln\left(2\right)R}\\
        =&\log_2\left(6\right)-\frac{\gamma^2E^{\ast 2}}{4\ln\left(2\right)R}-\left(m-1\right)\frac{\gamma^2E^{\ast 2}}{8\ln\left(2\right)R}\\
        <&0,
    \end{aligned}
\end{equation}
proving that:
\begin{equation}
    \mathbb{P}\left[\mathcal{S}\left(\gamma,m,\xi,\eta,\mathcal{I},R\right)\neq\varnothing\right]\leq\exp\left(-\operatorname{\Omega}\left(n\right)\right).
\end{equation}

All that remains to prove is the satisfaction of the quantum chaos property. This is the setting where $\mathcal{I}=\left\{\bm{1}\right\}$, so in particular we need only show the existence of some $\eta'>1-\xi+\eta$ such that:\footnote{This formula is derived by ignoring the covariance term in $\Psi\left(\gamma,m,\xi,\eta,c,\upsilon\right)$ and taking $\xi=1$, $\eta\to\eta'$, and $c=0$. See Eq.~\eqref{eq:psi_qcp_def} for the analogue for the quantum $k$-spin model.}
\begin{equation}
    \tilde{\varPsi}\left(\gamma,m,\eta',R\right):=\log_2\left(6\right)+\left(\operatorname{H}\left(\frac{\eta'}{2}\right)+\log_2\left(5\right)\frac{\eta'}{2}\right)\left(m-1\right)-\frac{m\gamma^2E^{\ast 2}}{2\ln\left(2\right)R}<0
\end{equation}
for the same choices of $\gamma$ and $m$ as previously.

Recall the assumption on the parameters:
\begin{align}
    \xi&>\eta,\\
    \eta'&<3\max\left(\left(\frac{\gamma^2 E^{\ast 2}}{24\ln\left(2\right)R}\right)^2,\frac{\gamma^2 E^{\ast 2}}{12\ln\left(5\right)R}\right).
\end{align}
Under this assumption, we have:
\begin{equation}
    \begin{aligned}
        \tilde{\varPsi}\left(\gamma,m,\eta,R\right)\leq&\log_2\left(6\right)+\left(m-1\right)\frac{\gamma^2 E^{\ast 2}}{8\ln\left(2\right)R}+\left(m-1\right)\log_2\left(5\right)\frac{\gamma^2 E^{\ast 2}}{8\ln\left(5\right)R}\\
        &-\frac{\gamma^2E^{\ast 2}}{2\ln\left(2\right)R}-\left(m-1\right)\frac{\gamma^2E^{\ast 2}}{2\ln\left(2\right)R}\\
        =&\log_2\left(6\right)-\frac{\gamma^2E^{\ast 2}}{2\ln\left(2\right)R}-\left(m-1\right)\frac{\gamma^2E^{\ast 2}}{4\ln\left(2\right)R}\\
        <&0,
    \end{aligned}
\end{equation}
proving the desired result.

\section*{Acknowledgments}

E.R.A.\ is funded in part by the Walter Burke Institute for Theoretical Physics at Caltech. E.R.A. thanks David Gamarnik, Bobak T.\ Kiani, and Alexander Zlokapa for enlightening discussions and comments on a draft of this manuscript. E.R.A.\ also thanks Robbie King for answering questions regarding \revref\cite{chen2024optimal} and \revref\cite{king2025triply}, and Eren C.\ Kızıldağ for answering questions regarding \revref\cite{gamarnik2023shatteringisingpurepspin}.

\section*{Data availability statement}

We do not analyze or generate any datasets, because our work proceeds within a theoretical and mathematical approach.

\section*{Declarations}

E.R.A.\ acknowledges support in part from the Walter Burke Institute for Theoretical Physics at Caltech. The author has no competing interests to declare that are relevant to the content of this article.

\appendices

\section{Background on the Quantum Wasserstein Distance}\label{sec:wass}

We here review the quantum Wasserstein distance of order $1$~\cite{9420734}, as well as introduce a natural generalization of it to all orders $p$. In what follows we state results only on qubits for simplicity, but analogous results hold for qudits.

Informally, the quantum Wasserstein distance is a quantum ``earth mover's'' metric in that states which differ only by a channel acting on $\ell$ qubits differ in Wasserstein distance by $\operatorname{O}\left(\ell\right)$; in this way, it can be thought of as a quantum generalization of the Hamming distance (and indeed, it reduces to the Hamming distance on bit strings). More formally, it is defined in the following way. Here, $\left\lVert\cdot\right\rVert_\ast$ denotes the trace norm, and $\mathcal{O}_n$ is the space of Hermitian observables on $n$ qubits. We also use the notation $\Tr_{\mathcal{I}}$ to denote the partial trace when tracing out the qubits labeled by the index set $\mathcal{I}$. We begin with the $p=1$ definition, due to \revref\cite{9420734}.
\begin{definition}[Quantum Wasserstein norm on qubits~\cite{9420734}]
    Let $\bm{X}$ be a Hermitian, traceless observable on $n$ qubits. The \emph{quantum Wasserstein norm of order $1$} is defined as:
    \begin{equation}
        \left\lVert\bm{X}\right\rVert_{W_1}:=\frac{1}{2}\min_{\left\{\bm{X}_i\right\}_{i=1}^n\in\mathcal{B}\left(\bm{X}\right)}\left(\sum_{i=1}^n\left\lVert\bm{X}_i\right\rVert_\ast\right),
    \end{equation}
    where
    \begin{equation}
        \mathcal{B}\left(\bm{X}\right)=\left\{\left\{\bm{X}_i\right\}_{i=1}^n:\bm{X}=\sum_{i=1}^n\bm{X}_i\wedge\bm{X}_i\in\mathcal{O}_n\wedge\Tr_{\left\{i\right\}}\left(\bm{X}_i\right)=0\right\}.
    \end{equation}
\end{definition}
\begin{definition}[Quantum Wasserstein distance on qubits~\cite{9420734}]
    For $\bm{\rho},\bm{\sigma}\in\mathcal{S}_n^{\mathrm{m}}$, their \emph{quantum Wasserstein distance of order $1$} is:
    \begin{equation}
        W_1\left(\bm{\rho},\bm{\sigma}\right):=\left\lVert\bm{\rho}-\bm{\sigma}\right\rVert_{W_1}.
    \end{equation}
\end{definition}
We also here introduce what we call the \emph{quantum Wasserstein distance of order $p$} ($p\geq 1$) as an immediate generalization where, instead of taking the $L^1$-norm of $\left(\left\lVert\bm{X}_i\right\rVert_\ast\right)_{i=1}^n$, we take the $L^{\frac{1}{p}}$ $F$-norm. While the convention of taking the $L^{\frac{1}{p}}$-norm---not the $L^p$ norm---may seem strange, we will later see that under our definition $\left\lVert\cdot\right\rVert_{W_2}$ is related to the classical Wasserstein distance of order $2$ on states diagonal in the computational basis.
\begin{definition}[Quantum Wasserstein $F$-norm on qubits]
    Let $\bm{X}$ be a Hermitian, traceless observable on $n$ qubits. The \emph{quantum Wasserstein $F$-norm of order $p$} is defined as:
    \begin{equation}
        \left\lVert\bm{X}\right\rVert_{W_p}:=\min_{\left\{\bm{X}_i\right\}_{i=1}^n\in\mathcal{B}\left(\bm{X}\right)}\left(\sum_{i=1}^n\left\lVert\frac{1}{2}\bm{X}_i\right\rVert_\ast^{\frac{1}{p}}\right),
    \end{equation}
    where
    \begin{equation}
        \mathcal{B}\left(\bm{X}\right)=\left\{\left\{\bm{X}_i\right\}_{i=1}^n:\bm{X}=\sum_{i=1}^n\bm{X}_i\wedge\bm{X}_i\in\mathcal{O}_n\wedge\Tr_{\left\{i\right\}}\left(\bm{X}_i\right)=0\right\}.
    \end{equation}
\end{definition}
\begin{definition}[Quantum Wasserstein distance of order $p$ on qubits]
    For $\bm{\rho},\bm{\sigma}\in\mathcal{S}_n^{\mathrm{m}}$, their \emph{quantum Wasserstein distance of order $p$} is:
    \begin{equation}
        W_p\left(\bm{\rho},\bm{\sigma}\right):=\left\lVert\bm{\rho}-\bm{\sigma}\right\rVert_{W_p}.
    \end{equation}
\end{definition}

Unlike more traditional metrics on the space of quantum states---such as the trace distance---the quantum Wasserstein distance is not unitarily invariant, i.e., $\left\lVert\bm{U}\bm{X}\bm{U}^\dagger\right\rVert_{W_p}$ does not necessarily equal $\left\lVert\bm{X}\right\rVert_{W_p}$. Furthermore, the norm is not necessarily contractive under quantum channels. That said, the metric still has some nice properties which we review (or prove) in what follows.

First, we prove an equivalence of the various quantum Wasserstein norms.
\begin{proposition}[Equivalence of quantum Wasserstein norms]\label{prop:equiv_wass_norms}
    For any $q\leq p$,
    \begin{equation}
        \left\lVert\bm{X}\right\rVert_{W_q}^q\leq\left\lVert\bm{X}\right\rVert_{W_p}^p\leq n^{p-q}\left\lVert\bm{X}\right\rVert_{W_q}^q.
    \end{equation}
\end{proposition}
\begin{proof}
    By H\"{o}lder's inequality, when $q\leq p$,
    \begin{equation}
        \sum_{i=1}^n\left\lVert\frac{1}{2}\bm{X}_i\right\rVert_\ast^{\frac{1}{p}}\leq\left(\sum_{i=1}^n\left\lVert\frac{1}{2}\bm{X}_i\right\rVert_\ast^{\frac{1}{q}}\right)^{\frac{q}{p}}\left(\sum_{i=1}^n1\right)^{1-\frac{q}{p}}=\left(\sum_{i=1}^n\left\lVert\frac{1}{2}\bm{X}_i\right\rVert_\ast^{\frac{1}{q}}\right)^{\frac{q}{p}}n^{1-\frac{q}{p}}.
    \end{equation}
    The other direction holds by the general ordering of $L^p$-norms, i.e., $\left\lVert\cdot\right\rVert_{\frac{1}{q}}^q\leq\left\lVert\cdot\right\rVert_{\frac{1}{p}}^p$ when $q\leq p$.
\end{proof}
Furthermore, there is an equivalence of norms between the quantum Wasserstein and trace norms.
\begin{proposition}[Equivalence of trace and quantum Wasserstein norms~{\cite[Proposition~2]{9420734}}]\label{prop:w_trace_norm_rel}
    For any traceless $\bm{X}\in\mathcal{O}_n$,
    \begin{equation}
        \frac{1}{2}\left\lVert\bm{X}\right\rVert_\ast\leq\left\lVert\bm{X}\right\rVert_{W_1}\leq\frac{n}{2}\left\lVert\bm{X}\right\rVert_\ast.
    \end{equation}
\end{proposition}

Second, though the quantum Wasserstein norm is generally not contractive under the action of quantum channels, it is contractive under the action of tensor-product channels.
\begin{proposition}[Contractivity under tensor product channels]\label{prop:contractive_tensor_prod}
    For any channel of the form
    \begin{equation}
        \bm{\varLambda}=\frac{1}{B}\sum_{b=1}^B\bigotimes_{i=1}^n\bm{\varLambda}_i^{\left(b\right)},
    \end{equation}
    we have:
    \begin{equation}
        \left\lVert\bm{\varLambda}\left(\bm{X}\right)\right\rVert_{W_p}\leq\left\lVert\bm{X}\right\rVert_{W_p}.
    \end{equation}
    This inequality is saturated when $\bm{\varLambda}$ is a tensor product of unitary channels.
\end{proposition}
\begin{proof}
    Such channels send $\mathcal{B}\left(\bm{X}\right)$ to itself, as for all such channels and any $\left\{\bm{X}_i\right\}_{i=1}^n\in\mathcal{B}\left(\bm{X}\right)$:
    \begin{equation}
        \bm{\varLambda}\left(\bm{X}\right)=\sum_{i=1}^n\bm{\varLambda}\left(\bm{X}_i\right)
    \end{equation}
    by linearity and:
    \begin{equation}
        \Tr_{\left\{i\right\}}\left(\bm{\varLambda}\left(\bm{X}_i\right)\right)=0
    \end{equation}
    for all $i\in\left[n\right]$ by the tensor product structure of the $\bigotimes_{i=1}^n\bm{\varLambda}_i^{\left(b\right)}$. The claim follows as the nuclear norm is nonincreasing under quantum channels, with equality when the channel is unitary. In particular,
    \begin{equation}
        \begin{aligned}
            \left\lVert\bm{X}\right\rVert_{W_p}&=\min_{\left\{\bm{X}_i\right\}_{i=1}^n\in\mathcal{B}\left(\bm{X}\right)}\left(\sum_{i=1}^n\left\lVert\frac{1}{2}\bm{X}_i\right\rVert_\ast^{\frac{1}{p}}\right)\\
            &\geq\min_{\left\{\bm{X}_i\right\}_{i=1}^n\in\mathcal{B}\left(\bm{X}\right)}\left(\sum_{i=1}^n\left\lVert\frac{1}{2}\bm{\varLambda}\left(\bm{X}_i\right)\right\rVert_\ast^{\frac{1}{p}}\right)\\
            &=\min_{\left\{\bm{\varLambda}\left(\bm{X}_i\right)\right\}_{i=1}^n\in\mathcal{B}\left(\bm{\varLambda}\left(\bm{X}\right)\right)}\left(\sum_{i=1}^n\left\lVert\frac{1}{2}\bm{\varLambda}\left(\bm{X}_i\right)\right\rVert_\ast^{\frac{1}{p}}\right)\\
            &=\left\lVert\bm{\varLambda}\left(\bm{X}\right)\right\rVert_{W_p},
        \end{aligned}
    \end{equation}
    where the inequality is an equality when the channel is unitary.
\end{proof}
More generally, one can quantify the contractivity of a channel under the quantum Wasserstein distance through the \emph{Wasserstein contraction $F$-norm}~\cite{9420734} of a channel $\bm{\varLambda}$; this is the superoperator $F$-norm induced by the Wasserstein distance of order $p$.
\begin{definition}[Wasserstein contraction norm~\cite{9420734}]\label{def:wass_cont_norm}
    Let $\bm{\varLambda}$ be a (potentially unnormalized) quantum channel. Its \emph{Wasserstein contraction $F$-norm of order $p$} is given by:
    \begin{equation}
        \left\lVert\bm{\varLambda}\right\rVert_{W_p\to W_p}:=\max_{\bm{X}\in\mathcal{O}_n:\Tr\left(\bm{X}\right)=0}\frac{\left\lVert\bm{\varLambda}\left(\bm{X}\right)\right\rVert_{W_p}}{\left\lVert\bm{X}\right\rVert_{W_p}}.
    \end{equation}
\end{definition}
For unitary channels, i.e., $\bm{\varLambda}\left(\bm{\rho}\right):=\bm{U}\bm{\rho}\bm{U}^\dagger$ for unitary $\bm{U}$, we will slightly abuse notation and speak of the Wasserstein contraction norm of the unitary operator $\left\lVert\bm{U}\right\rVert_{W_p\to W_p}$. When $p=1$ the Wasserstein contractive norm of a quantum channel can more generally be bounded by its light cone, generalizing Proposition~\ref{prop:contractive_tensor_prod} (up to a constant).
\begin{proposition}[{\cite[Proposition~13]{9420734}}]\label{prop:contractive_gen_channel}
    Let $\bm{\varLambda}$ be a quantum channel on $n$ qubits, and define the light cone $\mathcal{I}_i$ of qubit $i$ as the minimal-cardinality subset of qubits such that: 
    \begin{equation}
        \Tr_{\mathcal{I}_i}\left(\bm{\varLambda}\left(\bm{X}\right)\right)=\bm{0}
    \end{equation}
    for any Hermitian $\bm{X}$ satisfying $\Tr_{\left\{i\right\}}\left(\bm{X}\right)=0$. Then:
    \begin{equation}
        \left\lVert\bm{\varLambda}\right\rVert_{W_1\to W_1}\leq\frac{3}{2}\max_{i\in\left[n\right]}\left\lvert\mathcal{I}_i\right\rvert.
    \end{equation}
\end{proposition}

Finally, the quantum Wasserstein distance over mixtures of product states in a shared basis upper bounds the classical Wasserstein distance. We
begin with the simple case of product states, i.e., when the states are pure. We achieve this through the following proposition, generalizing Proposition~4 of \revref\cite{9420734}.
\begin{proposition}[Quantum Wasserstein distance over tensor products]\label{prop:quantum_wp_tps}
    For any $m,n\in\mathbb{N}$, traceless $\bm{X}\in\mathcal{O}_{m+n}$, and $p\geq 1$,
    \begin{equation}\label{eq:x_dec_tens}
        \left\lVert\bm{X}\right\rVert_{W_p}\geq\left\lVert\Tr_{\left[m+n\right]\setminus\left[m\right]}\left(\bm{X}\right)\right\rVert_{W_p}+\left\lVert\Tr_{\left[m\right]}\left(\bm{X}\right)\right\rVert_{W_p}.
    \end{equation}
    Furthermore, for any $\bm{\rho},\bm{\sigma}\in\mathcal{S}_m^{\mathrm{m}}$ and $\bm{\rho'},\bm{\sigma'}\in\mathcal{S}_n^{\mathrm{m}}$,
    \begin{equation}\label{eq:state_decomp}
        \left\lVert\bm{\rho}\otimes\bm{\rho'}-\bm{\sigma}\otimes\bm{\sigma'}\right\rVert_{W_p}=\left\lVert\bm{\rho}-\bm{\sigma}\right\rVert_{W_p}+\left\lVert\bm{\rho'}-\bm{\sigma'}\right\rVert_{W_p}.
    \end{equation}
\end{proposition}
\begin{proof}
    Let $\left\{\bm{X}_i\right\}_{i=1}^{m+n}$ be any element of $\mathcal{B}\left(\bm{X}\right)$. Note that it is also the case that:
    \begin{equation}
        \left\{\Tr_{\left[m+n\right]\setminus\left[m\right]}\left(\bm{X}_i\right)\right\}_{i=1}^m\in\mathcal{B}\left(\Tr_{\left[m+n\right]\setminus\left[m\right]}\left(\bm{X}\right)\right)
    \end{equation}
    and
    \begin{equation}
        \left\{\Tr_{\left[m\right]}\left(\bm{X}_i\right)\right\}_{i=m+1}^{m+n}\in\mathcal{B}\left(\Tr_{\left[m\right]}\left(\bm{X}\right)\right).
    \end{equation}
    We therefore have that:
    \begin{equation}
        \begin{aligned}
            \left\lVert\Tr_{\left[m+n\right]\setminus\left[m\right]}\left(\bm{X}\right)\right\rVert_{W_p}+\left\lVert\Tr_{\left[m\right]}\left(\bm{X}\right)\right\rVert_{W_p}&\leq\sum_{i=1}^m\left\lVert\frac{1}{2}\Tr_{\left[m+n\right]\setminus\left[m\right]}\left(\bm{X}_i\right)\right\rVert_\ast^{\frac{1}{p}}+\sum_{i=m+1}^{m+n}\left\lVert\frac{1}{2}\Tr_{\left[m\right]}\left(\bm{X}_i\right)\right\rVert_\ast^{\frac{1}{p}}\\
            &\leq\sum_{i=1}^m\left\lVert\frac{1}{2}\bm{X}_i\right\rVert_\ast^{\frac{1}{p}}+\sum_{i=m+1}^{m+n}\left\lVert\frac{1}{2}\bm{X}_i\right\rVert_\ast^{\frac{1}{p}}\\
            &=\sum_{i=1}^{m+n}\left\lVert\frac{1}{2}\bm{X}_i\right\rVert_\ast^{\frac{1}{p}},
        \end{aligned}
    \end{equation}
    with the second line following as the nuclear norm is nonincreasing under quantum channels. As $\left\{\bm{X}_i\right\}_{i=1}^{m+n}$ was arbitrary, Eq.~\eqref{eq:x_dec_tens} follows.

    We now consider Eq.~\eqref{eq:state_decomp}. Eq.~\eqref{eq:x_dec_tens} immediately implies that:
    \begin{equation}
        \left\lVert\bm{\rho}\otimes\bm{\rho'}-\bm{\sigma}\otimes\bm{\sigma'}\right\rVert_{W_p}\geq\left\lVert\bm{\rho}-\bm{\sigma}\right\rVert_{W_p}+\left\lVert\bm{\rho'}-\bm{\sigma'}\right\rVert_{W_p},
    \end{equation}
    leaving only the other direction. First, we claim that for general $\bm{Y}\in\mathcal{O}_a$ and traceless $\bm{X}\in\mathcal{O}_b$ it is the case that:
    \begin{equation}
        \left\lVert\bm{X}\otimes\bm{Y}\right\rVert_{W_p}\leq\left\lVert\bm{X}\right\rVert_{W_p}\left\lVert\bm{Y}\right\rVert_\ast^{\frac{1}{p}}
    \end{equation}
    and
    \begin{equation}
        \left\lVert\bm{X}\otimes\bm{Y}\right\rVert_{W_p}\leq\left\lVert\bm{X}\right\rVert_\ast^{\frac{1}{p}}\left\lVert\bm{Y}\right\rVert_{W_p}.
    \end{equation}
    This is because, for any $\left\{\bm{X}_i\right\}_{i=1}^b\in\mathcal{B}\left(\bm{X}\right)$, we have by the definition of the quantum Wasserstein distance:
    \begin{equation}
        \begin{aligned}
            \left\lVert\bm{X}\otimes\bm{Y}\right\rVert_{W_p}&\leq\sum_{i=1}^b\left\lVert\frac{1}{2}\bm{X}_i\otimes\bm{Y}\right\rVert_\ast^{\frac{1}{p}}\\
            &=\left\lVert\bm{Y}\right\rVert_\ast^{\frac{1}{p}}\sum_{i=1}^b\left\lVert\frac{1}{2}\bm{X}_i\right\rVert_\ast^{\frac{1}{p}};
        \end{aligned}
    \end{equation}
    the other inequality follows similarly. Therefore, by the triangle inequality,
    \begin{equation}
        \begin{aligned}
            \left\lVert\bm{\rho}\otimes\bm{\rho'}-\bm{\sigma}\otimes\bm{\sigma'}\right\rVert_{W_p}&\leq\left\lVert\left(\bm{\rho}-\bm{\sigma}\right)\otimes\bm{\rho'}\right\rVert_{W_p}+\left\lVert\bm{\sigma}\otimes\left(\bm{\rho'}-\bm{\sigma'}\right)\right\rVert_{W_p}\\
            &\leq\left\lVert\bm{\rho}-\bm{\sigma}\right\rVert_{W_p}+\left\lVert\bm{\rho'}-\bm{\sigma'}\right\rVert_{W_p},
        \end{aligned}
    \end{equation}
    and Eq.~\eqref{eq:state_decomp} follows.
\end{proof}
This immediately gives the following corollary, generalizing Corollary~1 of \revref\cite{9420734}.
\begin{corollary}[Quantum Wasserstein distance over product states]\label{cor:quantum_w1_product_states}
    For any $\bm{\rho},\bm{\sigma}\in\mathcal{S}_n^{\mathrm{m}}$,
    \begin{equation}
        \left\lVert\bm{\rho}-\bm{\sigma}\right\rVert_{W_p}\geq\sum_{i=1}^n\left\lVert\frac{1}{2}\Tr_{\left[n\right]\setminus\left\{i\right\}}\left(\bm{\rho}-\bm{\sigma}\right)\right\rVert_\ast^{\frac{1}{p}},
    \end{equation}
    with equality when $\bm{\rho}$ and $\bm{\sigma}$ are product states.
\end{corollary}
\begin{proof}
    This follows from repeatedly applying Proposition~\ref{prop:quantum_wp_tps} with $m=1$ and $n$ the remaining qubits.
\end{proof}
In particular, the quantum Wasserstein distance of order $p$ reduces to the Hamming distance over computational basis states as $\frac{1}{2}\Tr_{\left[n\right]\setminus\left\{i\right\}}\left(\bm{\rho}-\bm{\sigma}\right)$ has trace norm $1$ whenever $\bm{\rho}$ and $\bm{\sigma}$ disagree on qubit $i$ and $0$ otherwise.

We now consider the quantum Wasserstein distance over general mixtures. We first define the notion of a \emph{coupling} between two probability distributions $p$ and $q$, specializing to discrete spaces for simplicity.
\begin{definition}[Coupling on a discrete space]
    Let $p$ and $q$ be probability distributions over a set $\mathcal{X}$ of finite cardinality. A probability distribution $\pi$ on $\mathcal{X}\times\mathcal{X}$ is called a \emph{coupling} between $p$ and $q$ if:
    \begin{align}
        p\left(x\right)&=\sum_{y\in\mathcal{X}}\pi\left(x,y\right),\\
        q\left(y\right)&=\sum_{x\in\mathcal{X}}\pi\left(x,y\right).
    \end{align}
\end{definition}
Couplings are used to define the classical Wasserstein distance, summarized as follows.
\begin{definition}[Classical Wasserstein distances~{\cite[Definition~2]{9420734}}]
    The classical Wasserstein distance of order $\alpha$ between two distributions $p$ and $q$ over a discrete space $\mathcal{X}$ is defined as:
    \begin{equation}
        W_\alpha\left(p,q\right):=\inf_{\pi\in\mathcal{C}\left(p,q\right)}\left(\mathbb{E}_{\left(x,y\right)\sim\pi}d_{\mathrm{H}}\left(x,y\right)^\alpha\right)^{\frac{1}{\alpha}}.
    \end{equation}
\end{definition}
We now state our result, a strengthened version of Proposition~6 of \revref\cite{9420734} relating the quantum and classical Wasserstein distances, to match our setting. Note we use the notation $W_\alpha$ rather than $W_p$ here to avoid confusion with the distribution $p\left(\overline{s}\right)$.
\begin{proposition}[Quantum Wasserstein distance over mixtures of product states]\label{prop:quant_wass_prod_states}
    Consider quantum states $\bm{\rho}$ and $\bm{\sigma}$ mutually diagonalized by the same product state basis $\left\{\bm{s}\right\}_{\bm{s}\in\left\{0,1\right\}^{\times n}}$:
    \begin{align}
        \bm{\rho}&=\sum_{\bm{s}\in\left\{0,1\right\}^{\times n}} p\left(\bm{s}\right)\ket{\bm{s}}\bra{\bm{s}},\\
        \bm{\sigma}&=\sum_{\bm{s}\in\left\{0,1\right\}^{\times n}} q\left(\bm{s}\right)\ket{\bm{s}}\bra{\bm{s}}.
    \end{align}
    Let $\mathcal{C}\left(p,q\right)$ be the set of couplings between $p$ and $q$. Then, for any $\alpha\geq 1$,
    \begin{equation}\label{eq:quant_wass_red}
        \inf_{\pi\in\mathcal{C}\left(p,q\right)}\left(\mathbb{E}_{\left(\bm{s},\bm{t}\right)\sim\pi}\left\lVert\ket{\bm{s}}\bra{\bm{s}}-\ket{\bm{t}}\bra{\bm{t}}\right\rVert_{W_1}^\alpha\right)^{\frac{1}{\alpha}}\leq W_\alpha\left(\bm{\rho},\bm{\sigma}\right)\leq\inf_{\pi\in\mathcal{C}\left(p,q\right)}\sum_{\bm{s},\bm{t}\in\left\{0,1\right\}^{\times n}}\pi\left(\bm{s},\bm{t}\right)^{\frac{1}{\alpha}}\left\lVert\ket{\bm{s}}\bra{\bm{s}}-\ket{\bm{t}}\bra{\bm{t}}\right\rVert_{W_1}.
    \end{equation}
\end{proposition}
\begin{proof}
    By Proposition~\ref{prop:contractive_tensor_prod}, the quantum Wasserstein distance is invariant under conjugation by tensor products of $1$-local unitaries. In particular, we can consider $\left\{\ket{\bm{s}}\right\}_{\bm{s}\in\left\{0,1\right\}^{\times n}}$ to be the computational basis WLOG.

    We now prove the desired result, beginning by showing that:
    \begin{equation}\label{eq:wass_prod_st_first_dir}
        W_\alpha\left(\bm{\rho},\bm{\sigma}\right)\geq\inf_{\pi\in\mathcal{C}\left(p,q\right)}\left(\mathbb{E}_{\left(\bm{s},\bm{t}\right)\sim\pi}\left\lVert\ket{\bm{s}}\bra{\bm{s}}-\ket{\bm{t}}\bra{\bm{t}}\right\rVert_{W_1}^\alpha\right)^{\frac{1}{\alpha}}.
    \end{equation}
    First, by the definition of the quantum Wasserstein distance we have:
    \begin{equation}
        \bm{\rho}-\bm{\sigma}=\left\lVert\bm{\rho}-\bm{\sigma}\right\rVert_{W_\alpha}^\alpha\sum_{i=1}^n r_i\left(\bm{\rho}^{\left(i\right)}-\bm{\sigma}^{\left(i\right)}\right)
    \end{equation}
    for some $\bm{r}$ with $1=\left\lVert\bm{r}\right\rVert_{\frac{1}{\alpha}}^\alpha\geq\left\lVert\bm{r}\right\rVert_1$ and $\Tr_{\left\{i\right\}}\left(\bm{\rho}^{\left(i\right)}-\bm{\sigma}^{\left(i\right)}\right)=0$ for all $i\in\left[n\right]$. In particular, taking the $\bm{\rho}^{\left(i\right)}$ and $\bm{\sigma}^{\left(i\right)}$ to be diagonal WLOG and denoting their diagonals as $p^{\left(i\right)}$ and $q^{\left(i\right)}$, respectively,
    \begin{equation}
        p-q=\left\lVert\bm{\rho}-\bm{\sigma}\right\rVert_{W_\alpha}^\alpha\sum_{i=1}^n r_i\left(p^{\left(i\right)}-q^{\left(i\right)}\right).
    \end{equation}
    Noting that $p^{\left(i\right)}$ and $q^{\left(i\right)}$ differ by at most $1$ in classical Wasserstein distance $W_\alpha$ as they marginalize to the same distribution on the components $\left[n\right]\setminus\left\{i\right\}$, we then have:
    \begin{equation}
        \inf_{\pi\in\mathcal{C}\left(p,q\right)}\sum_{\bm{s},\bm{t}\in\left\{0,1\right\}^{\times n}}\pi\left(\bm{s},\bm{t}\right)d_{\mathrm{H}}\left(\bm{s},\bm{t}\right)^\alpha=W_\alpha\left(p,q\right)^\alpha\leq\left\lVert\bm{\rho}-\bm{\sigma}\right\rVert_{W_\alpha}^\alpha,
    \end{equation}
    where $d_{\mathrm{H}}$ denotes the Hamming distance. Eq.~\eqref{eq:wass_prod_st_first_dir} then follows by taking the $\alpha$th root of both sides and recalling that the Hamming distance between computational basis states equals their Wasserstein distance of order $1$.

    We now show that:
    \begin{equation}
        W_\alpha\left(\bm{\rho},\bm{\sigma}\right)\leq\inf_{\pi\in\mathcal{C}\left(p,q\right)}\sum_{\bm{s},\bm{t}\in\left\{0,1\right\}^{\times n}}\pi\left(\bm{s},\bm{t}\right)^{\frac{1}{\alpha}}\left\lVert\ket{\bm{s}}\bra{\bm{s}}-\ket{\bm{t}}\bra{\bm{t}}\right\rVert_{W_1}.
    \end{equation}
    Let $\pi\in\mathcal{C}\left(p,q\right)$ be arbitrary. We have from the triangle inequality that:
    \begin{equation}
        \begin{aligned}
            \left\lVert\bm{\rho}-\bm{\sigma}\right\rVert_{W_\alpha}&=\left\lVert\sum_{\bm{s},\bm{t}\in\left\{0,1\right\}^{\times n}}\pi\left(\bm{s},\bm{t}\right)\left(\ket{\bm{s}}\bra{\bm{s}}-\ket{\bm{t}}\bra{\bm{t}}\right)\right\rVert_{W_\alpha}\\
            &\leq\sum_{\bm{s},\bm{t}\in\left\{0,1\right\}^{\times n}}\left\lVert\pi\left(\bm{s},\bm{t}\right)\left(\ket{\bm{s}}\bra{\bm{s}}-\ket{\bm{t}}\bra{\bm{t}}\right)\right\rVert_{W_\alpha}\\
            &=\sum_{\bm{s},\bm{t}\in\left\{0,1\right\}^{\times n}}\pi\left(\bm{s},\bm{t}\right)^{\frac{1}{\alpha}}d_{\mathrm{H}}\left(\bm{s},\bm{t}\right),
        \end{aligned}
    \end{equation}
    with the final line following from Corollary~\ref{cor:quantum_w1_product_states}. The final result then follows by recalling that the Hamming distance between computational basis states equals their Wasserstein distance of order $1$.
\end{proof}

\section{Examples of Stable Quantum Algorithms}\label{sec:stable_q_algs}

We here relate the notion of stability in Wasserstein distance that we use in the main text to other natural notions of the stability of a quantum algorithm, as well as give explicit examples of standard quantum algorithms which are stable. As a tool to convert between various notions of stability, we will use the \emph{Wasserstein complexity}~\cite{li2022wasserstein}:
\begin{equation}
    \operatorname{WC}\left(\bm{U}\right):=\max_{\bm{\rho}\in\mathcal{S}_n}\left\lVert\bm{\rho}-\bm{U}\bm{\rho}\bm{U}^\dagger\right\rVert_{W_1},
\end{equation}
where $\mathcal{S}_n$ is the set of quantum pure states on $n$ qubits. We will also use as a tool the \emph{Nielsen complexity} $\operatorname{NC}\left(\cdot\right)$, a known lower bound on the gate complexity of a quantum circuit~\cite{nielsen2006}. We will often make use of the fact that the Nielsen complexity of a unitary operator
\begin{equation}
    \bm{U}=\exp\left(-\ci\sum_i c_i\bm{P}_i\right),
\end{equation}
for $\bm{P}_i$ distinct Pauli operators, is upper-bounded by:\footnote{There are many equivalent definitions of the Nielsen complexity; we here use the ``$F_p$'' definition~\cite{nielsen2006} to match the convention of \revref\cite{li2022wasserstein}, where the metric is defined in terms of the $L^1$-norm and utilizes a penalty function.}
\begin{equation}
    \operatorname{NC}\left(\bm{U}\right)\leq\left\lVert\bm{c}\right\rVert_1.
\end{equation}
Finally, we will use the fact that the Wasserstein complexity lower-bounds the Nielsen complexity, a result due to \revref\cite{li2022wasserstein}.
\begin{theorem}[{\cite[Theorem~7]{li2022wasserstein}}]\label{thm:wc_lbs_nc}
    \begin{equation}
        \operatorname{WC}\left(\bm{U}\right)\leq\frac{1}{4\sqrt{2}}\operatorname{NC}\left(\bm{U}\right).
    \end{equation}
\end{theorem}

\subsection{Lipschitz Gate Complexity}

One natural notion of stability is stability in \emph{gate complexity}; that is, small changes in the input should lead to states which differ by low-complexity quantum circuits. We show that stability under this definition implies the definition of stability we give in Definition~\ref{def:stable_qas}.
\begin{proposition}[Stability in gate complexity]\label{prop:stab_gc}
    Consider a quantum algorithm $\bm{\mathcal{A}}\left(\bm{X},\omega\right)$ as in the setting of Definition~\ref{def:stable_qas}, and let $\operatorname{GC}\left(\cdot\right)$ denote the gate complexity of a quantum circuit. Assume there exist a $\mathfrak{d}\in\mathbb{N}$ and a set $\mathcal{K}\subseteq\left[0,1\right]$ such that for all $\kappa'\in\mathcal{K}$,
    \begin{equation}\label{eq:gate_comp_stab}
        \mathbb{P}_{\left(\bm{X},\bm{Y},\omega\right)\sim\mathbb{P}_{\bm{X},\bm{Y}}^{\mathfrak{d},\kappa'}\otimes\mathbb{P}_\varOmega}\left[\inf\limits_{\bm{U}\in\operatorname{SU}\left(2^n\right):\bm{\mathcal{A}}\left(\bm{Y},\omega\right)=\bm{U}\bm{\mathcal{A}}\left(\bm{X},\omega\right)\bm{U}^\dagger}\operatorname{GC}\left(\bm{U}\right)\leq f+L\left\lVert\bm{X}-\bm{Y}\right\rVert_1\right]\geq 1-p_{\mathrm{st}}.
    \end{equation}
    Then, $\bm{\mathcal{A}}$ is $\left(\left(1+\frac{f}{4\sqrt{2}}\right)\sqrt{n},\frac{L}{4\sqrt{2}}\sqrt{n},\mathfrak{d},\mathcal{K},p_{\mathrm{st}}\right)$-stable.
\end{proposition}
\begin{proof}
    First, note that the gate complexity of a quantum circuit is lower-bounded by its Nielsen complexity~\cite{nielsen2006}:
    \begin{equation}
        \operatorname{NC}\left(\bm{U}\right)\leq\operatorname{GC}\left(\bm{U}\right).
    \end{equation}
    By Theorem~\ref{thm:wc_lbs_nc}, the Nielsen complexity of a quantum circuit is in turn bounded up to a constant by the Wasserstein complexity:
    \begin{equation}
        \operatorname{NC}\left(\bm{U}\right)\geq 4\sqrt{2}\operatorname{WC}\left(\bm{U}\right)=4\sqrt{2}\max_{\bm{\rho}\in\mathcal{S}_n}\left\lVert\bm{\rho}-\bm{U}\bm{\rho}\bm{U}^\dagger\right\rVert_{W_1}.
    \end{equation}
    For the $\bm{U}$ in the infinimum of Eq.~\eqref{eq:gate_comp_stab}, we then have that:
    \begin{equation}
        \begin{aligned}
            \operatorname{NC}\left(\bm{U}\right)\geq 4\sqrt{2}\max_{\bm{\rho}\in\mathcal{S}_n}\left\lVert\bm{\rho}-\bm{U}\bm{\rho}\bm{U}^\dagger\right\rVert_{W_1}&\geq 4\sqrt{2}\left\lVert\bm{\mathcal{A}}\left(\bm{X},\omega\right)-\bm{U}\bm{\mathcal{A}}\left(\bm{X},\omega\right)\bm{U}^\dagger\right\rVert_{W_1}\\
            &=4\sqrt{2}\left\lVert\bm{\mathcal{A}}\left(\bm{X},\omega\right)-\bm{\mathcal{A}}\left(\bm{Y},\omega\right)\right\rVert_{W_1}.
        \end{aligned}
    \end{equation}
    Finally, by the equivalence of Wasserstein norms (Proposition~\ref{prop:equiv_wass_norms}):
    \begin{equation}
        \left\lVert\bm{\mathcal{A}}\left(\bm{X},\omega\right)-\bm{\mathcal{A}}\left(\bm{Y},\omega\right)\right\rVert_{W_1}\geq\frac{1}{n}\left\lVert\bm{\mathcal{A}}\left(\bm{X},\omega\right)-\bm{\mathcal{A}}\left(\bm{Y},\omega\right)\right\rVert_{W_2}^2,
    \end{equation}
    so in particular if $\frac{1}{n}\left\lVert\bm{\mathcal{A}}\left(\bm{X},\omega\right)-\bm{\mathcal{A}}\left(\bm{Y},\omega\right)\right\rVert_{W_2}^2\geq 1$:
    \begin{equation}
        \left\lVert\bm{\mathcal{A}}\left(\bm{X},\omega\right)-\bm{\mathcal{A}}\left(\bm{Y},\omega\right)\right\rVert_{W_1}\geq\frac{1}{\sqrt{n}}\left\lVert\bm{\mathcal{A}}\left(\bm{X},\omega\right)-\bm{\mathcal{A}}\left(\bm{Y},\omega\right)\right\rVert_{W_2}.
    \end{equation}
    The result then follows by Definition~\ref{def:stable_qas} and taking an additional $\sqrt{n}$ in the first stability parameter to account for the case
    \begin{equation}
        \frac{1}{n}\left\lVert\bm{\mathcal{A}}\left(\bm{X},\omega\right)-\bm{\mathcal{A}}\left(\bm{Y},\omega\right)\right\rVert_{W_2}^2<1\implies\left\lVert\bm{\mathcal{A}}\left(\bm{X},\omega\right)-\bm{\mathcal{A}}\left(\bm{Y},\omega\right)\right\rVert_{W_2}<\sqrt{n}.
    \end{equation}
\end{proof}
Of course, this immediately implies that stability under the Nielsen complexity---upper-bounded by the gate complexity---also implies stability under the Wasserstein metric:
\begin{corollary}
    Consider a quantum algorithm $\bm{\mathcal{A}}\left(\bm{X},\omega\right)$ as in the setting of Definition~\ref{def:stable_qas}. Assume there exist a $\mathfrak{d}\in\mathbb{N}$ and a set $\mathcal{K}\subseteq\left[0,1\right]$ such that for all $\kappa'\in\mathcal{K}$,
    \begin{equation}
        \mathbb{P}_{\left(\bm{X},\bm{Y},\omega\right)\sim\mathbb{P}_{\bm{X},\bm{Y}}^{\mathfrak{d},\kappa'}\otimes\mathbb{P}_\varOmega}\left[\inf\limits_{\bm{U}\in\operatorname{SU}\left(2^n\right):\bm{\mathcal{A}}\left(\bm{Y},\omega\right)=\bm{U}\bm{\mathcal{A}}\left(\bm{X},\omega\right)\bm{U}^\dagger}\operatorname{NC}\left(\bm{U}\right)\leq f+L\left\lVert\bm{X}-\bm{Y}\right\rVert_1\right]\geq 1-p_{\mathrm{st}}.
    \end{equation}
    Then, $\bm{\mathcal{A}}$ is $\left(\left(1+\frac{f}{4\sqrt{2}}\right)\sqrt{n},\frac{L}{4\sqrt{2}}\sqrt{n},\mathfrak{d},\mathcal{K},p_{\mathrm{st}}\right)$-stable.
\end{corollary}

\subsection{Lipschitz Hamiltonian Evolution}

We can consider another natural notion of stability defined by the Lipschitzness of the Hamiltonian evolution as a function of the inputs. We show that stability in this sense also implies the notion of stability we consider in the main text.
\begin{proposition}[Stability in Hamiltonian dynamics]\label{prop:stab_ham_dyn}
    Let $\bm{\rho}_0\left(\omega\right)$ be an arbitrary quantum state depending only on a source of classical randomness $\omega\sim\mathbb{P}_\varOmega$ and consider a quantum algorithm of the form:
    \begin{equation}
        \bm{\mathcal{A}}\left(\bm{X},\omega\right)=\exp\left(-\ci\bm{H}\left(\bm{X},\omega\right)\right)\bm{\rho}_0\left(\omega\right)\exp\left(\ci\bm{H}\left(\bm{X},\omega\right)\right)
    \end{equation}
    for $\bm{H}\left(\bm{X},\omega\right)$ a Hermitian $n$-qubit operator with Pauli decomposition:
    \begin{equation}
        \bm{H}\left(\bm{X},\omega\right)=\sum_i h_i\left(\bm{X},\omega\right)\bm{P}_i.
    \end{equation}
    Assume there exist a $\mathfrak{d}\in\mathbb{N}$ and a $\mathcal{K}\subseteq\left[0,1\right]$ such that for all $\kappa'\in\left[\kappa,1\right]$,
    \begin{equation}\label{eq:lip_cont_l2_ev}
        \mathbb{P}_{\left(\bm{X},\bm{Y},\omega\right)\sim\mathbb{P}_{\bm{X},\bm{Y}}^{\mathfrak{d},\kappa'}\otimes\mathbb{P}_\varOmega}\left[\left\lVert\bm{h}\left(\bm{X},\omega\right)-\bm{h}\left(\bm{Y},\omega\right)\right\rVert_1\leq f+L\left\lVert\bm{X}-\bm{Y}\right\rVert_1\right]\geq 1-p_{\mathrm{st}}
    \end{equation}
    and
    \begin{equation}\label{eq:lip_cont_wedge_ev}
        \mathbb{P}_{\left(\bm{X},\bm{Y},\omega\right)\sim\mathbb{P}_{\bm{X},\bm{Y}}^{\mathfrak{d},\kappa'}\otimes\mathbb{P}_\varOmega}\left[\left\lVert\left[\bm{H}\left(\bm{X},\omega\right),\bm{H}\left(\bm{Y},\omega\right)\right]\right\rVert_{\mathrm{op}}\leq\tilde{f}+\tilde{L}\left\lVert\bm{X}-\bm{Y}\right\rVert_1\right]\geq 1-p_{\mathrm{st}}.
    \end{equation}
    Then, $\bm{\mathcal{A}}$ is $\left(\left(1+\frac{f}{4\sqrt{2}}+\frac{3\tilde{f}}{2}n\right)\sqrt{n},\left(\frac{L}{4\sqrt{2}}+\frac{3\tilde{L}}{2}n\right)\sqrt{n},\mathfrak{d},\mathcal{K},2p_{\mathrm{st}}\right)$-stable. Here, $\mathbb{P}_{\bm{X},\bm{Y}}^{\mathfrak{d},\kappa'}$ is as defined in Definition~\ref{def:stable_qas}.
\end{proposition}
\begin{proof}
    Recall from the equivalence of quantum Wasserstein norms (Proposition~\ref{prop:equiv_wass_norms}) that if
    \begin{equation}
        \frac{1}{n}\left\lVert\bm{\mathcal{A}}\left(\bm{X},\omega\right)-\bm{\mathcal{A}}\left(\bm{Y},\omega\right)\right\rVert_{W_2}^2\geq 1,
    \end{equation}
    it is the case that:
    \begin{equation}
        \frac{1}{\sqrt{n}}\left\lVert\bm{\mathcal{A}}\left(\bm{X},\omega\right)-\bm{\mathcal{A}}\left(\bm{Y},\omega\right)\right\rVert_{W_2}\leq\frac{1}{n}\left\lVert\bm{\mathcal{A}}\left(\bm{X},\omega\right)-\bm{\mathcal{A}}\left(\bm{Y},\omega\right)\right\rVert_{W_2}^2\leq\left\lVert\bm{\mathcal{A}}\left(\bm{X},\omega\right)-\bm{\mathcal{A}}\left(\bm{Y},\omega\right)\right\rVert_{W_1}.
    \end{equation}
    Just as in the proof of Proposition~\ref{prop:stab_gc} we now assume this, and account for the other case by taking an additional $\sqrt{n}$ in the first stability parameter at the end. By the triangle inequality:
    \begin{equation}
        \begin{aligned}
            &\left\lVert\bm{\mathcal{A}}\left(\bm{X},\omega\right)-\bm{\mathcal{A}}\left(\bm{Y},\omega\right)\right\rVert_{W_1}\\
            &\leq\left\lVert\bm{\mathcal{A}}\left(\bm{X},\omega\right)-\bm{\varDelta}\left(\bm{X},\bm{Y},\omega\right)\bm{\mathcal{A}}\left(\bm{X},\omega\right)\bm{\varDelta}\left(\bm{X},\bm{Y},\omega\right)^\dagger\right\rVert_{W_1}+\left\lVert\bm{\mathcal{A}}\left(\bm{Y},\omega\right)-\bm{\varDelta}\left(\bm{X},\bm{Y},\omega\right)\bm{\mathcal{A}}\left(\bm{X},\omega\right)\bm{\varDelta}\left(\bm{X},\bm{Y},\omega\right)^\dagger\right\rVert_{W_1}\\
            &\leq\left\lVert\bm{\mathcal{A}}\left(\bm{X},\omega\right)-\bm{\varDelta}\left(\bm{X},\bm{Y},\omega\right)\bm{\mathcal{A}}\left(\bm{X},\omega\right)\bm{\varDelta}\left(\bm{X},\bm{Y},\omega\right)^\dagger\right\rVert_{W_1}+\left\lVert\bm{\mathcal{A}}\left(\bm{Y},\omega\right)-\bm{\varUpsilon}\left(\bm{X},\bm{Y},\omega\right)\bm{\mathcal{A}}\left(\bm{Y},\omega\right)\bm{\varUpsilon}\left(\bm{X},\bm{Y},\omega\right)^\dagger\right\rVert_{W_1}\\
            &\leq\operatorname{WC}\left(\bm{\varDelta}\left(\bm{X},\bm{Y},\omega\right)\right)+\operatorname{WC}\left(\bm{\varUpsilon}\left(\bm{X},\bm{Y},\omega\right)\right),
        \end{aligned}
    \end{equation}
    where we have defined:
    \begin{align}
        \bm{\varDelta}\left(\bm{X},\bm{Y},\omega\right)&:=\exp\left(-\ci\bm{H}\left(\bm{Y},\omega\right)+\ci\bm{H}\left(\bm{X},\omega\right)\right),\\
        \bm{\varUpsilon}\left(\bm{X},\bm{Y},\omega\right)&:=\exp\left(-\ci\bm{H}\left(\bm{Y},\omega\right)+\ci\bm{H}\left(\bm{X},\omega\right)\right)\exp\left(-\ci\bm{H}\left(\bm{X},\omega\right)\right)\exp\left(\ci\bm{H}\left(\bm{Y},\omega\right)\right).
    \end{align}
    By Theorem~\ref{thm:wc_lbs_nc}, the Wasserstein complexity is related to the Nielsen complexity by a constant:
    \begin{equation}\label{eq:wc_nc_bound_lipschitz}
        \begin{aligned}
            \operatorname{WC}\left(\bm{\Delta}\left(\bm{X},\bm{Y},\omega\right)\right)&\leq\frac{1}{4\sqrt{2}}\operatorname{NC}\left(\bm{\Delta}\left(\bm{X},\bm{Y},\omega\right)\right)\\
            &\leq\frac{1}{4\sqrt{2}}\left\lVert\bm{h}\left(\bm{X},\omega\right)-\bm{h}\left(\bm{Y},\omega\right)\right\rVert_1.
        \end{aligned}
    \end{equation}
    It remains to bound $\operatorname{WC}\left(\bm{\varUpsilon}\left(\bm{X},\bm{Y},\omega\right)\right)$. We define:
    \begin{equation}\label{eq:m_def_trott_err}
        \bm{M}\left(\bm{X},\bm{Y},\omega\right):=\bm{\varUpsilon}\left(\bm{X},\bm{Y},\omega\right)-\bm{I},
    \end{equation}
    where $\bm{I}$ is the $n$-qubit identity operator. Note that:
    \begin{equation}
        \begin{aligned}
            \operatorname{WC}\left(\bm{\varUpsilon}\left(\bm{X},\bm{Y},\omega\right)\right)=&\sup_{\bm{\rho}\in\mathcal{S}_n}\left\lVert\bm{\rho}-\bm{\varUpsilon}\left(\bm{X},\bm{Y},\omega\right)\bm{\rho}\bm{\varUpsilon}\left(\bm{X},\bm{Y},\omega\right)^\dagger\right\rVert_{W_1}\\
            =&\sup_{\bm{\rho}\in\mathcal{S}_n}\left\lVert\bm{\rho}\bm{M}\left(\bm{X},\bm{Y},\omega\right)^\dagger+\bm{M}\left(\bm{X},\bm{Y},\omega\right)\bm{\rho}+\bm{M}\left(\bm{X},\bm{Y},\omega\right)\bm{\rho}\bm{M}\left(\bm{X},\bm{Y},\omega\right)^\dagger\right\rVert_{W_1}\\
            \leq&\frac{n}{2}\sup_{\bm{\rho}\in\mathcal{S}_n}\left\lVert\bm{\rho}\bm{M}\left(\bm{X},\bm{Y},\omega\right)^\dagger+\bm{M}\left(\bm{X},\bm{Y},\omega\right)\bm{\rho}+\bm{M}\left(\bm{X},\bm{Y},\omega\right)\bm{\rho}\bm{M}\left(\bm{X},\bm{Y},\omega\right)^\dagger\right\rVert_\ast\\
            \leq&\frac{n}{2}\sup_{\bm{\rho}\in\mathcal{S}_n}\left\lVert\bm{\rho}\bm{M}\left(\bm{X},\bm{Y},\omega\right)^\dagger\right\rVert_\ast+\frac{n}{2}\sup_{\bm{\rho}\in\mathcal{S}_n}\left\lVert\bm{M}\left(\bm{X},\bm{Y},\omega\right)\bm{\rho}\right\rVert_\ast\\
            &+\frac{n}{2}\sup_{\bm{\rho}\in\mathcal{S}_n}\left\lVert\bm{M}\left(\bm{X},\bm{Y},\omega\right)\bm{\rho}\bm{M}\left(\bm{X},\bm{Y},\omega\right)^\dagger\right\rVert_\ast\\
            \leq&n\left\lVert\bm{M}\left(\bm{X},\bm{Y},\omega\right)\right\rVert_{\mathrm{op}}+\frac{n}{2}\left\lVert\bm{M}\left(\bm{X},\bm{Y},\omega\right)\right\rVert_{\mathrm{op}}^2,
        \end{aligned}
    \end{equation}
    where the third line follows from Proposition~\ref{prop:w_trace_norm_rel}, the penultimate line from the triangle inequality, and the final line from H\"{o}lder's inequality. Note from the triangle inequality and Eq.~\eqref{eq:m_def_trott_err} that $\left\lVert\bm{M}\left(\bm{X},\bm{Y},\omega\right)\right\rVert_{\mathrm{op}}\leq 2$, so we can further upper bound this expression with the weaker yet simpler bound:
    \begin{equation}
        \operatorname{WC}\left(\bm{\varUpsilon}\left(\bm{X},\bm{Y},\omega\right)\right)\leq 3n\left\lVert\bm{M}\left(\bm{X},\bm{Y},\omega\right)\right\rVert_{\mathrm{op}}.
    \end{equation}

    $\bm{M}\left(\bm{X},\bm{Y},\omega\right)$ can be interpreted as the multiplicative Trotter error of implementing $\bm{\varDelta}\left(\bm{X},\bm{Y},\omega\right)$ using a first-order Trotter formula. Specifically, by the unitary invariance of the operator norm,
    \begin{equation}
        \begin{aligned}
            \left\lVert\bm{M}\left(\bm{X},\bm{Y},\omega\right)\right\rVert_{\mathrm{op}}&=\left\lVert\exp\left(-\ci\bm{H}\left(\bm{X},\omega\right)+\ci\bm{H}\left(\bm{Y},\omega\right)\right)\bm{M}\left(\bm{X},\bm{Y},\omega\right)\right\rVert_{\mathrm{op}}\\
            &=\left\lVert\exp\left(-\ci\bm{H}\left(\bm{X},\omega\right)\right)\exp\left(\ci\bm{H}\left(\bm{Y},\omega\right)\right)-\exp\left(-\ci\bm{H}\left(\bm{X},\omega\right)+\ci\bm{H}\left(\bm{Y},\omega\right)\right)\right\rVert_{\mathrm{op}};
        \end{aligned}
    \end{equation}
    Proposition~9 of \revref\cite{PhysRevX.11.011020} then gives an upper bound for this error in terms of the operator norm of a commutator:
    \begin{equation}
        \begin{aligned}
            \left\lVert\bm{M}\left(\bm{X},\bm{Y},\omega\right)\right\rVert_{\mathrm{op}}&=\left\lVert\exp\left(-\ci\bm{H}\left(\bm{X},\omega\right)\right)\exp\left(\ci\bm{H}\left(\bm{Y},\omega\right)\right)-\exp\left(-\ci\bm{H}\left(\bm{X},\omega\right)+\ci\bm{H}\left(\bm{Y},\omega\right)\right)\right\rVert_{\mathrm{op}}\\
            &\leq\frac{1}{2}\left\lVert\left[\bm{H}\left(\bm{X},\omega\right),\bm{H}\left(\bm{Y},\omega\right)\right]\right\rVert_{\mathrm{op}}.
        \end{aligned}
    \end{equation}
    This bound taken in combination with Eq.~\eqref{eq:wc_nc_bound_lipschitz} gives the final result.
\end{proof}

We can generalize this statement to depth-$p$ algorithms by taking into account the operator growth induced by each layer of the circuit. This can be formalized using the Wasserstein contraction norm, given as Definition~\ref{def:wass_cont_norm}.
\begin{proposition}[Stability in layered Hamiltonian dynamics]\label{prop:stab_layered_ham_dyn}
    Let $\bm{\rho}_0\left(\omega\right)$ be an arbitrary quantum state depending only on a source of classical randomness $\omega\sim\mathbb{P}_\varOmega$ and consider a quantum algorithm of the form:
    \begin{equation}
        \bm{\mathcal{A}}_p\left(\bm{X},\omega\right)=\prod_{\beta=p}^1\exp\left(-\ci\bm{H}^{\left(\beta\right)}\left(\bm{X},\omega\right)\right)\bm{\rho}_0\left(\omega\right)\prod_{\beta=1}^p\exp\left(\ci\bm{H}^{\left(\beta\right)}\left(\bm{X},\omega\right)\right),
    \end{equation}
    for $\bm{H}^{\left(\alpha\right)}\left(\bm{X},\omega\right)$ Hermitian $n$-qubit operators with Pauli decompositions:
    \begin{equation}
        \bm{H}^{\left(\beta\right)}\left(\bm{X},\omega\right)=\sum_i h_i^{\left(\beta\right)}\left(\bm{X},\omega\right)\bm{P}_i.
    \end{equation}
    If there exists $\mathcal{K}\subseteq\left[0,1\right]$ such that for all $\kappa'\in\mathcal{K}$,
    \begin{equation}
        \mathbb{P}_{\left(\bm{X},\bm{Y},\omega\right)\sim\mathbb{P}_{\bm{X},\bm{Y}}^{\mathfrak{d},\kappa'}\otimes\mathbb{P}_\varOmega}\left[\left\lVert\bm{h}^{\left(\beta\right)}\left(\bm{X},\omega\right)-\bm{h}^{\left(\beta\right)}\left(\bm{Y},\omega\right)\right\rVert_1\leq f^{\left(\beta\right)}+L^{\left(\beta\right)}\left\lVert\bm{X}-\bm{Y}\right\rVert_1\right]\geq 1-p_{\mathrm{st}}^{\left(\beta\right)}
    \end{equation}
    and
    \begin{equation}
        \mathbb{P}_{\left(\bm{X},\bm{Y},\omega\right)\sim\mathbb{P}_{\bm{X},\bm{Y}}^{\mathfrak{d},\kappa'}\otimes\mathbb{P}_\varOmega}\left[\left\lVert\left[\bm{H}^{\left(\beta\right)}\left(\bm{X},\omega\right),\bm{H}^{\left(\beta\right)}\left(\bm{Y},\omega\right)\right]\right\rVert_{\mathrm{op}}\leq\tilde{f}^{\left(\beta\right)}+\tilde{L}^{\left(\beta\right)}\left\lVert\bm{X}-\bm{Y}\right\rVert_1\right]\geq 1-p_{\mathrm{st}}^{\left(\beta\right)}
    \end{equation}
    for all $\beta\in\left[p\right]$, and
    \begin{equation}\label{eq:stab_w1_cont_norm}
        \mathbb{P}_{\left(\bm{X},\omega\right)\sim\mathbb{P}_{\bm{X}}\otimes\mathbb{P}_\varOmega}\left[\left\lVert\exp\left(-\ci\bm{H}^{\left(\beta\right)}\left(\bm{X},\omega\right)\right)\right\rVert_{W_1\to W_1}\leq W^{\left(\beta\right)}\right]\geq 1-p_{\mathrm{st}}^{\left(\beta\right)}
    \end{equation}
    for all integer $1<\beta\leq p$, then $\bm{\mathcal{A}}_p$ is $\left(\left(1+\frac{f_p}{4\sqrt{2}}+\frac{3\tilde{f}_p}{2}n\right)\sqrt{n},\left(\frac{L_p}{4\sqrt{2}}+\frac{3\tilde{L}_p}{2}n\right)\sqrt{n},\mathfrak{d},\mathcal{K},3p_{\mathrm{st},p}\right)$-stable. Here, $\mathbb{P}_{\bm{X},\bm{Y}}^{\mathfrak{d},\kappa'}$ is as defined in Definition~\ref{def:stable_qas}, and defining:
    \begin{equation}\label{eq:v_alpha_def}
        V_\beta:=\prod_{\gamma=\beta+1}^p W^{\left(\gamma\right)}
    \end{equation}
    with the convention $V_p=1$, we have:
    \begin{align}
        f_p:=\sum_{\beta=1}^p V_\beta f^{\left(\beta\right)},&\quad\tilde{f}_p:=\sum_{\beta=1}^p V_\beta\tilde{f}^{\left(\beta\right)},\\
        L_p:=\sum_{\beta=1}^p V_\beta L^{\left(\beta\right)},&\quad\tilde{L}_p:=\sum_{\beta=1}^p V_\beta\tilde{L}^{\left(\beta\right)},\\
        p_{\mathrm{st},p}&:=\min\left(\sum_{\beta=1}^p p_{\mathrm{st}}^{\left(\beta\right)},\frac{1}{3}\right).
    \end{align}
\end{proposition}
\begin{proof}
    As in the proof of Proposition~\ref{prop:stab_ham_dyn}, we assume that $\frac{1}{n}\left\lVert\bm{\mathcal{A}}_p\left(\bm{X},\omega\right)-\bm{\mathcal{A}}_p\left(\bm{Y},\omega\right)\right\rVert_{W_2}^2\geq 1$ and take $f\to f+\sqrt{n}$ at the end to account for the other case. We proceed inductively in $p$, assuming the inductive hypothesis that, with probability at least $1-3p_{\mathrm{st},p-1}$,
    \begin{equation}
        \left\lVert\bm{\mathcal{A}}\left(\bm{X},\omega\right)-\bm{\mathcal{A}}\left(\bm{Y},\omega\right)\right\rVert_{W_1}\leq\frac{f_{p-1}}{4\sqrt{2}}+\frac{3\tilde{f}_{p-1}}{2}n+\left(\frac{L_{p-1}}{4\sqrt{2}}+\frac{3\tilde{L}_{p-1}}{2}n\right)\left\lVert\bm{X}-\bm{Y}\right\rVert_1.
    \end{equation}
    The proof of Proposition~\ref{prop:stab_ham_dyn} implies the desired result in the base case $p=1$. For $p>1$, we first define:
    \begin{equation}
        \bm{\varGamma}_{p-1}\left(\bm{X},\bm{Y},\omega\right):=\bm{\mathcal{A}}_{p-1}\left(\bm{X},\omega\right)-\bm{\mathcal{A}}_{p-1}\left(\bm{Y},\omega\right),
    \end{equation}
    which by the inductive hypothesis has (with probability at least $1-p_{\mathrm{st},p-1}$) bounded Wasserstein norm. We also define:
    \begin{align}
        \bm{\varDelta}_p\left(\bm{X},\bm{Y},\omega\right)&:=\exp\left(-\ci\bm{H}^{\left(p\right)}\left(\bm{Y},\omega\right)+\ci\bm{H}^{\left(p\right)}\left(\bm{X},\omega\right)\right),\\
        \bm{\varUpsilon}_p\left(\bm{X},\bm{Y},\omega\right)&:=\exp\left(-\ci\bm{H}^{\left(p\right)}\left(\bm{Y},\omega\right)+\ci\bm{H}^{\left(p\right)}\left(\bm{X},\omega\right)\right)\exp\left(-\ci\bm{H}^{\left(p\right)}\left(\bm{X},\omega\right)\right)\exp\left(\ci\bm{H}^{\left(p\right)}\left(\bm{Y},\omega\right)\right),\\
        \bm{\tilde{\mathcal{A}}}_p\left(\bm{X},\bm{Y},\omega\right)&:=\exp\left(-\ci\bm{H}^{\left(p\right)}\left(\bm{Y},\omega\right)\right)\bm{\mathcal{A}}_{p-1}\left(\bm{X},\omega\right)\exp\left(\ci\bm{H}^{\left(p\right)}\left(\bm{Y},\omega\right)\right).
    \end{align}
    We then have by the triangle inequality:
    \begin{equation}
        \begin{aligned}
            &\left\lVert\bm{\mathcal{A}}_p\left(\bm{X},\omega\right)-\bm{\mathcal{A}}_p\left(\bm{Y},\omega\right)\right\rVert_{W_1}\\
            \leq&\left\lVert\bm{\mathcal{A}}_p\left(\bm{X},\omega\right)-\bm{\varDelta}_p\left(\bm{X},\bm{Y},\omega\right)\bm{\mathcal{A}}_p\left(\bm{X},\omega\right)\bm{\varDelta}_p\left(\bm{X},\bm{Y},\omega\right)^\dagger\right\rVert_{W_1}\\
            &+\left\lVert\bm{\tilde{\mathcal{A}}}_p\left(\bm{X},\bm{Y},\omega\right)-\bm{\varDelta}_p\left(\bm{X},\bm{Y},\omega\right)\bm{\mathcal{A}}_p\left(\bm{X},\omega\right)\bm{\varDelta}_p\left(\bm{X},\bm{Y},\omega\right)^\dagger\right\rVert_{W_1}\\
            &+\left\lVert\bm{\mathcal{A}}_p\left(\bm{Y},\omega\right)-\bm{\tilde{\mathcal{A}}}_p\left(\bm{X},\bm{Y},\omega\right)\right\rVert_{W_1}\\
            \leq&\operatorname{WC}\left(\bm{\varDelta}_p\left(\bm{X},\bm{Y},\omega\right)\right)+\left\lVert\bm{\tilde{\mathcal{A}}}_p\left(\bm{Y},\omega\right)-\bm{\varUpsilon}_p\left(\bm{X},\bm{Y},\omega\right)\bm{\tilde{\mathcal{A}}}_p\left(\bm{X},\bm{Y},\omega\right)\bm{\varUpsilon}_p\left(\bm{X},\bm{Y},\omega\right)^\dagger\right\rVert\\
            &+\left\lVert\exp\left(-\ci\bm{H}^{\left(p\right)}\left(\bm{Y},\omega\right)\right)\bm{\varGamma}_{p-1}\left(\bm{X},\bm{Y},\omega\right)\exp\left(\ci\bm{H}^{\left(p\right)}\left(\bm{Y},\omega\right)\right)\right\rVert_{W_1}\\
            \leq&\operatorname{WC}\left(\bm{\varDelta}_p\left(\bm{X},\bm{Y},\omega\right)\right)+\operatorname{WC}\left(\bm{\varUpsilon}_p\left(\bm{X},\bm{Y},\omega\right)\right)+\left\lVert\exp\left(-\ci\bm{H}^{\left(p\right)}\left(\bm{Y},\omega\right)\right)\right\rVert_{W_1\to W_1}\left\lVert\bm{\varGamma}_{p-1}\left(\bm{X},\bm{Y},\omega\right)\right\rVert_{W_1},
        \end{aligned}
    \end{equation}
    where in the final line we recall the Wasserstein contraction norm of Definition~\ref{def:wass_cont_norm}.
    The first two terms in the final line are identical to those considered in the proof of Proposition~\ref{prop:stab_ham_dyn}, and it is the case that
    \begin{equation}
        \left\lVert\exp\left(-\ci\bm{H}^{\left(p\right)}\left(\bm{Y},\omega\right)\right)\right\rVert_{W_1\to W_1}\leq W^{\left(p\right)}
    \end{equation}
    conditioned on the event given as Eq.~\eqref{eq:stab_w1_cont_norm} occurring. Finally, with probability at least $1-3p_{\mathrm{st},p-1}$,
    \begin{equation}
        \left\lVert\bm{\varGamma}_{p-1}\left(\bm{X},\bm{Y},\omega\right)\right\rVert_{W_1}\leq\frac{f_{p-1}}{4\sqrt{2}}+\frac{3\tilde{f}_{p-1}}{2}n+\left(\frac{L_{p-1}}{4\sqrt{2}}+\frac{3\tilde{L}_{p-1}}{2}n\right)\left\lVert\bm{X}-\bm{Y}\right\rVert_1
    \end{equation}
    by the inductive hypothesis. The final result follows by the union bound, the equivalence of quantum Wasserstein norms (Proposition~\ref{prop:equiv_wass_norms}), and noting:
    \begin{align}
        f_p=f^{\left(p\right)}+W^{\left(p\right)} f_{p-1}&=\sum_{\beta=1}^p V_\beta f^{\left(\beta\right)},\quad\tilde{f}_p=\tilde{f}^{\left(p\right)}+W^{\left(p\right)}\tilde{f}_{p-1}=\sum_{\beta=1}^p V_\beta\tilde{f}^{\left(\beta\right)},\\
        L_p=L^{\left(p\right)}+W^{\left(p\right)} L_{p-1}&=\sum_{\beta=1}^p V_\beta L^{\left(\beta\right)},\quad\tilde{L}_p=\tilde{L}^{\left(p\right)}+W^{\left(p\right)}\tilde{L}_{p-1}=\sum_{\beta=1}^p V_\beta\tilde{L}^{\left(\beta\right)},\\
        p_{\mathrm{st},p}&=\min\left(p_{\mathrm{st}}^{\left(p\right)}+p_{\mathrm{st},p-1},\frac{1}{3}\right)=\min\left(\sum_{\beta=1}^p p_{\mathrm{st}}^{\left(\beta\right)},\frac{1}{3}\right).
    \end{align}
\end{proof}

\subsection{Trotterized Quantum Annealing}

We now show that a class of algorithms based on the popular \emph{quantum annealing} optimization algorithm~\cite{farhi2000quantum} is stable. While first proposed using time-dependent Hamiltonian evolution---which is typically difficult to implement in practice---one can perform a variant with only time-independent Hamiltonian evolution~\cite{farhi2014quantumapproximate,PhysRevA.92.042303}. This is the variant we consider here, which we call \emph{$p$-Trotterized quantum annealing}.
\begin{definition}[$p$-Trotterized quantum annealing]\label{def:p_disc_quantum_anneal}
    Consider an $n$-qubit Hamiltonian $\bm{H}_C\left(\bm{X}\right)$ for which one wishes to find a maximal energy state, and consider a partitioning:
    \begin{equation}\label{eq:cost_ham_part}
        \bm{H}_C\left(\bm{X}\right)=\sum_{i=1}^K\bm{H}_C^{\left(i\right)}\left(\bm{X}\right)
    \end{equation}
    where (at fixed $i$) $\left\{\bm{H}_C^{\left(i\right)}\left(\bm{X}\right)\right\}_{\bm{X}}$ is mutually commuting. We call the quantum algorithm:
    \begin{equation}
        \bm{\mathcal{A}}_p\left(\bm{X},\omega\right)=\ket{\psi_p\left(\bm{X},\omega\right)}\bra{\psi_p\left(\bm{X},\omega\right)}
    \end{equation}
    the \emph{$p$-Trotterized quantum annealing} algorithm, where $\ket{\psi_p\left(\bm{X},\omega\right)}$ is defined as the $n$-qubit state:
    \begin{equation}
        \ket{\psi_p\left(\bm{X},\omega\right)}:=\prod_{l=p}^1\left(\exp\left(-\ci\beta_l\bm{H}_M\left(\omega\right)\right)\left(\prod_{i=K}^1\exp\left(-\ci\frac{\gamma_l^{\left(i\right)}}{\sqrt{n}}\bm{H}_C^{\left(i\right)}\left(\bm{X}\right)\right)\right)\right)\ket{\psi_0\left(\omega\right)}
    \end{equation}
    for some choice of initial state $\ket{\psi_0\left(\omega\right)}$ and ``mixing Hamiltonian'' $\bm{H}_M\left(\omega\right)$.
\end{definition}
The scaling of the parameters $\gamma_l^{\left(i\right)}$ by $\frac{1}{\sqrt{n}}$ is motivated by the fact that, for the problems we consider, $\left\lVert\bm{H}_C\left(\bm{X}\right)\right\rVert_{\mathrm{op}}=\operatorname{\Theta}\left(\sqrt{n}\right)$ w.h.p. Typically, $\ket{\psi_0\left(\omega\right)}$ is chosen to be the maximal-energy eigenstate of $\bm{H}_M\left(\omega\right)$ so that the success of the algorithm is guaranteed for sufficiently large $p$~\cite{farhi2014quantumapproximate,PhysRevA.92.042303}, though we do not require that here. In what follows we are also agnostic as to how parameters $\bm{\theta}=\left(\beta_l,\gamma_l^{\left(i\right)}\right)_{l,i}$ are chosen. If the parameters are optimized over, this gives an optimization algorithm known as the \emph{Hamiltonian variational ansatz} (HVA)~\cite{PhysRevA.92.042303}. If further $\bm{H}_P$ is diagonal in the computational basis, the algorithm is typically known as the \emph{quantum approximate optimization algorithm} (QAOA)~\cite{farhi2014quantumapproximate}.

In the case of QAOA, it was known that $p$-Trotterized quantum annealing for $p\leq\operatorname{O}\left(\log\left(n\right)\right)$ was ``stable'' in a way that implied algorithmic hardness in optimizing certain classical combinatorial optimization problems~\cite{farhi2020stable,PRXQuantum.4.010309,anshu2023concentrationbounds}. Via Proposition~\ref{prop:stab_layered_ham_dyn}, we show that $p$-Trotterized quantum annealing algorithms are stable in the more general sense of our Definition~\ref{def:stable_qas}, implying hardness even for finding the ground states of \emph{quantum}, noncommuting spin glass models. Our work also generalizes previous studies of algorithmic hardness in optimizing low-depth circuits via gradient descent~\cite{anschuetz2022critical,anschuetzkiani2022,anschuetz2025unified}. Here, we are agnostic to the training algorithm used---the parameters can be chosen however one wishes---and we do not require that the circuit is drawn according to some distribution approximating the Haar distribution.
\begin{corollary}[$p$-Trotterized quantum annealing is stable]\label{cor:p_trott_qa}
    Assume that $\bm{H}_M\left(\omega\right)$ is $d$-local and that
    \begin{equation}\label{eq:cost_ham_lin_form}
        \bm{H}_C^{\left(i\right)}\left(\bm{X}\right)=\frac{1}{\sqrt{Z}}\sum_{j\in\mathcal{I}_i}X_j\bm{P}_j
    \end{equation}
    for some index set $\mathcal{I}_i$ labeling $d$-local Pauli operators $\bm{P}_j$, where $Z=\sum_{i=1}^K\left\lvert\mathcal{I}_i\right\rvert\geq n$. Assume both are supported on interaction hypergraphs of degree at most $\mathfrak{d}$. Let $\bm{\theta}=\left(\beta_l,\gamma_l^{\left(i\right)}\right)_{l,i}$ be the vector of all parameters of the algorithm. Then, the $p$-Trotterized quantum annealing algorithm is $\left(\sqrt{n},\lambda_p,\mathfrak{d},0,0\right)$-stable, where:
    \begin{equation}
        \lambda_p=\frac{\left\lVert\bm{\theta}\right\rVert_\infty}{4\sqrt{2n}}\left(\frac{3}{2}d\mathfrak{d}\right)^{\left(K+1\right)p}.
    \end{equation}
\end{corollary}
\begin{proof}
    We aim to show that the assumptions of Proposition~\ref{prop:stab_layered_ham_dyn} are satisfied. First, note using Proposition~\ref{prop:contractive_gen_channel} that the $W^{\left(\beta\right)}$ are bounded: 
    \begin{equation}
        W^{\left(\beta\right)}\leq\frac{3}{2}d\mathfrak{d}.
    \end{equation}
    As $\left\{\bm{H}_C^{\left(i\right)}\left(\bm{X}\right)\right\}_{\bm{X}}$ is a set of mutually commuting operators by assumption, we also have that $\tilde{f}^{\left(\beta\right)}=\tilde{L}^{\left(\beta\right)}=0$ for all $\beta$. Finally,
    \begin{equation}
        L^{\left(\beta\right)}\sqrt{n}\leq\frac{\left\lVert\bm{\theta}\right\rVert_\infty}{\sqrt{n}\sqrt{Z}}\sqrt{n}\leq\frac{\left\lVert\bm{\theta}\right\rVert_\infty}{\sqrt{n}}.
    \end{equation}
    The result then follows directly from Proposition~\ref{prop:stab_layered_ham_dyn}.
\end{proof}

\subsection{Generalized Phase Estimation}

We can also consider a class of algorithms defined by the phase estimation primitive~\cite{Nielsen_Chuang_2010_fourier}. Here, some $\bm{X}$-independent (though potentially problem class-dependent) initial state $\bm{\rho}_0$ is chosen, and $\bm{H}_C\left(\bm{X}\right)$ then measured via phase estimation in the hopes of achieving a high energy with high probability. This occurs if $\bm{\rho}_0$ has good initial fidelity with the high-energy space of the $\bm{H}_C\left(\bm{X}\right)$. We here consider a very general class of phase-estimation like algorithms.
\begin{definition}[Generalized phase estimation]
    Consider an $n$-qubit Hamiltonian $\bm{H}_C\left(\bm{X}\right)$ for which one wishes to find a maximal energy state. Consider some initial $n$-qubit pure state $\bm{\rho}_0\left(\omega\right)$ and $A$-qubit ancillary register initialized in the pure state $\bm{\sigma}_0\left(\omega\right)$, where $\omega$ is a source of classical randomness $\omega\sim\mathbb{P}_\varOmega$. Let $\bm{H}_A\left(\omega\right)$ be a Hermitian operator acting only on the ancillary register with Pauli decomposition:
    \begin{equation}
        \bm{H}_A\left(\omega\right)=\sum_i h_{A,i}\left(\omega\right)\bm{P}_i,
    \end{equation}
    and let $\bm{\mathcal{M}}_A$ be an arbitrary quantum channel acting only on the ancillary register. We call the quantum algorithm:
    \begin{equation}
        \bm{\mathcal{A}}\left(\bm{X},\omega\right)=\bm{\mathcal{M}}_A\left(\left(\exp\left(-\ci\bm{H}_C\left(\bm{X}\right)\otimes\bm{H}_A\left(\omega\right)\right)\left(\bm{\rho}_0\left(\omega\right)\otimes\bm{\sigma}_0\left(\omega\right)\right)\exp\left(\ci\bm{H}_C\left(\bm{X}\right)\otimes\bm{H}_A\left(\omega\right)\right)\right)\right)
    \end{equation}
    \emph{generalized phase estimation} with $A$ ancillary qubits.
\end{definition}
For instance, if $\bm{\mathcal{M}}_A$ is measurement in the Fourier basis, $\bm{\sigma}_0$ is the state $\ket{+}^{\otimes A}$, and
\begin{equation}
    \bm{H}_A=t\sum_{i=0}^{A-1}2^i\ket{1}\bra{1}_i
\end{equation}
for some choice of $t>0$, then this reduces to the traditional phase estimation algorithm~\cite{Nielsen_Chuang_2010_fourier}.

We now show that generalized phase estimation is a stable quantum algorithm as a result of Proposition~\ref{prop:stab_ham_dyn}.
\begin{corollary}[Generalized phase estimation is stable]\label{cor:gen_phase_est}
    Assume $\bm{H}_C\left(\bm{X}\right)$ is of the form:
    \begin{equation}
        \bm{H}_C\left(\bm{X}\right)=\frac{1}{\sqrt{Z}}\sum_{i\in\mathcal{I}}X_i\bm{P}_i
    \end{equation}
    for some index set $\mathcal{I}$ labeling a set of Pauli operators $\left\{\bm{P}_i\right\}_{i\in\mathcal{I}}$, where $Z=\left\lvert\mathcal{I}\right\rvert\geq n$. Let $\tilde{L}$ be such that:
    \begin{equation}
        \left\lVert\left[\bm{H}_C\left(\bm{X}\right),\bm{H}_C\left(\bm{Y}\right)\right]\right\rVert_{\mathrm{op}}\leq\tilde{L}\left\lVert\bm{X}-\bm{Y}\right\rVert_1
    \end{equation}
    \changetwo{with probability $1-p_{\mathrm{st}}$ over} $\bm{X},\bm{Y}$. Then, generalized phase estimation with $A$ ancillary qubits is $\left(\sqrt{n},\lambda_A,\infty,0,\changetwo{\min\left(3p_{\mathrm{st}},1\right)}\right)$-stable, where:
    \begin{equation}
        \lambda_A=\frac{3}{4}A\left(\frac{\left\lVert\bm{h}_A\right\rVert_1}{2\sqrt{2}}+3\tilde{L}\left\lVert\bm{H}_A\right\rVert_{\mathrm{op}}^2 n^{\frac{3}{2}}\right).
    \end{equation}
\end{corollary}
\begin{proof}
    By Proposition~\ref{prop:stab_ham_dyn}, the algorithm on $n+A$ qubits:
    \begin{equation}
        \bm{\tilde{\mathcal{A}}}\left(\bm{X}\right)=\left(\exp\left(-\ci\bm{H}_C\left(\bm{X}\right)\otimes\bm{H}_A\right)\left(\bm{\rho}_0\otimes\bm{\sigma}_0\right)\exp\left(\ci\bm{H}_C\left(\bm{X}\right)\otimes\bm{H}_A\right)\right)
    \end{equation}
    is $\left(0,\lambda,\infty,0,\changetwo{\min\left(3p_{\mathrm{st}},1\right)}\right)$-stable, with:
    \begin{equation}
        \lambda=\frac{1}{4\sqrt{2}}\left(\frac{\left\lVert\bm{h}_A\right\rVert_1}{\sqrt{Z}}\sqrt{n}\right)+\frac{3\tilde{L}}{2}\left\lVert\bm{H}_A\right\rVert_{\mathrm{op}}^2n^{\frac{3}{2}}\leq\frac{\left\lVert\bm{h}_A\right\rVert_1}{4\sqrt{2}}+\frac{3\tilde{L}}{2}\left\lVert\bm{H}_A\right\rVert_{\mathrm{op}}^2n^{\frac{3}{2}}.
    \end{equation}
    Now, by Proposition~\ref{prop:contractive_gen_channel},
    \begin{equation}
        \left\lVert\bm{\mathcal{A}}\left(\bm{X}\right)-\bm{\mathcal{A}}\left(\bm{Y}\right)\right\rVert_{W_1}=\left\lVert\bm{\mathcal{M}}_A\left(\bm{\tilde{\mathcal{A}}}\left(\bm{X}\right)\right)-\bm{\mathcal{M}}_A\left(\bm{\tilde{\mathcal{A}}}\left(\bm{Y}\right)\right)\right\rVert_{W_1}\leq\frac{3}{2}A\left\lVert\bm{\tilde{\mathcal{A}}}\left(\bm{X}\right)-\bm{\tilde{\mathcal{A}}}\left(\bm{Y}\right)\right\rVert_{W_1};
    \end{equation}
    in particular, $\bm{\mathcal{A}}_A$ shares (or improves upon) the stability parameters of $\bm{\tilde{\mathcal{A}}}_A$ up to a factor of $\frac{3}{2}A$. This gives the final result.
\end{proof}
Using the choice of $\bm{H}_A$ as in the traditional phase estimation algorithm~\cite{Nielsen_Chuang_2010_fourier} gives the following.
\begin{corollary}[Phase estimation is stable]\label{cor:phase_est_stab}
    Consider the setting of Corollary~\ref{cor:gen_phase_est}, with the specific choice:
    \begin{equation}
        \bm{H}_A=t\sum_{i=0}^{A-1}2^i\ket{1}\bra{1}_i=t\sum_{i=0}^{A-1}2^i\left(\frac{\bm{I}-\bm{Z}_{i+1}}{2}\right)
    \end{equation}
    for some $t>0$, where $\bm{Z}_i$ denotes the Pauli operator $\bm{\sigma}^{\left(3\right)}$ acting on the $i$th ancillary qubit. Then, phase estimation with $A$ ancillary qubits is $\left(\sqrt{n},\lambda_t,\infty,0,\changetwo{\min\left(3p_{\mathrm{st}},1\right)}\right)$-stable, where:
    \begin{equation}
        \lambda_t=\frac{3}{4}At\left(2^{A-\frac{5}{2}}+3\times 2^{2A}t\tilde{L}n^{\frac{3}{2}}\right).
    \end{equation}
\end{corollary}
\begin{proof}
    We have that:
    \begin{equation}
        \left\lVert\bm{H}_A\right\rVert_{\mathrm{op}}=t\sum_{i=0}^{A-1}2^i=t\left(2^A-1\right)\leq 2^A t.
    \end{equation}
    We also calculate the $L^1$-norm of the Pauli coefficients:
    \begin{equation}
        \left\lVert\bm{h}_A\right\rVert_1=\frac{t}{2}\sum_{i=0}^{A-1}2^i=\frac{t}{2}\left(2^A-1\right)\leq 2^{A-1}t,
    \end{equation}
    giving the final result by Corollary~\ref{cor:gen_phase_est}.
\end{proof}
\change{$A$ dictates the number of bits of precision the energy is measured to, and parameterizes the complexity of the algorithm.} For models with $\left\lVert\bm{H}_C\left(\bm{X}\right)\right\rVert_{\mathrm{op}}=\operatorname{\Theta}\left(\sqrt{n}\right)$---such as the typical case for the quantum spin glasses we study in the main text---then one would take $t=\operatorname{\Theta}\left(\frac{1}{\sqrt{n}}\right)$, giving a stability parameter of:
\begin{equation}
    \lambda_t\leq\operatorname{O}\left(n^{-\frac{1}{2}}+\tilde{L}\sqrt{n}\right)\change{2^{\operatorname{O}\left(A\right)}}.
\end{equation}

\changetwo{We now consider $\tilde{L}$ for the same concrete example we considered in Corollary~\ref{cor:trott_algs_fail} of the main text: the sparsified quantum $k$-spin model with sparsity parameter $p=\operatorname{\Theta}\left(n^{-\left(k-1\right)}\right)$. We have:
\begin{equation}
    \begin{aligned}
        \left\lVert\left[\bm{H}_C\left(\bm{X}\right),\bm{H}_C\left(\bm{Y}\right)\right]\right\rVert_{\mathrm{op}}&=\left\lVert\left[\bm{H}_C\left(\bm{X}\right)-\bm{H}_C\left(\bm{Y}\right),\bm{H}_C\left(\bm{Y}\right)\right]\right\rVert_{\mathrm{op}}\\
        &=\left\lVert\left[\bm{H}_C\left(\bm{X}-\bm{Y}\right),\bm{H}_C\left(\bm{Y}\right)\right]\right\rVert_{\mathrm{op}}\\
        &=\frac{2}{Z}\left\lVert\sum_{i\in\mathcal{I}}\left(X_i-Y_i\right)\sum_{j:\left[\bm{P}_i,\bm{P}_j\right]\neq\bm{0}}Y_j\frac{\left[\bm{P}_i,\bm{P}_j\right]}{2}\right\rVert_{\mathrm{op}}\\
        &\leq\frac{2}{Z}\sum_{i\in\mathcal{I}}\left\lvert X_i-Y_i\right\rvert\left\lVert\frac{1}{2}\sum_{j:\left[\bm{P}_i,\bm{P}_j\right]\neq\bm{0}}Y_j\left[\bm{P}_i,\bm{P}_j\right]\right\rVert_{\mathrm{op}}\\
        &\leq\frac{2}{Z}\left\lVert\bm{X}-\bm{Y}\right\rVert_1\max_{i\in\mathcal{I}}\left\lVert\frac{1}{2}\sum_{j:\left[\bm{P}_i,\bm{P}_j\right]\neq\bm{0}}Y_j\left[\bm{P}_i,\bm{P}_j\right]\right\rVert_{\mathrm{op}}.
    \end{aligned}
\end{equation}
For $p=\operatorname{\Theta}\left(n^{-\left(k-1\right)}\right)$, by Bernstein's inequality $Z=\operatorname{\Theta}\left(n\right)$ with probability $1-\exp\left(-\operatorname{\Omega}\left(n\right)\right)$ over the sparsification. In particular, with probability exponentially close to $1$ over the sparsification, for each $i$ there are only a constant number of $j$ indices such that $\left[\bm{P}_i,\bm{P}_j\right]\neq\bm{0}$ as the $\bm{P}_i$ are $k$-local. We therefore have that with high probability:
\begin{equation}
    \left\lVert\frac{1}{2}\sum_{j:\left[\bm{P}_i,\bm{P}_j\right]\neq\bm{0}}Y_j\left[\bm{P}_i,\bm{P}_j\right]\right\rVert_{\mathrm{op}}=\operatorname{O}\left(1\right)
\end{equation}
for all $i$. This gives:
\begin{equation}
    \tilde{L}=\operatorname{O}\left(n^{-1}\right),
\end{equation}
and in particular phase estimation for this model is $\left(\sqrt{n},n^{-1/2}2^{\operatorname{O}\left(A\right)},\infty,0,\exp\left(-\operatorname{\Omega}\left(n\right)\right)\right)$-stable.
}

\subsection{Lindbladian Evolution}

We finally consider quantum Metropolis-like algorithms, which simulate the natural thermalization process of a system interacting with a bath~\cite{chen2023efficientexactnoncommutativequantum,jiang2024quantummetropolissamplingweak}. We generally write these algorithms as Trotterized Hamiltonian interactions with a bath system.
\begin{definition}[Lindbladian evolution algorithms]\label{def:therm_algs}
    Consider an $n$-qubit Hamiltonian $\bm{H}_C\left(\bm{X}\right)$ for which one wishes to find a maximal energy state, an $N$-qubit bath Hamiltonian $\bm{H}_B\left(\omega\right)$, and a $2$-local $\left(n+N\right)$-qubit interaction Hamiltonian $\bm{H}_I\left(\omega\right)$. Consider a partitioning:
    \begin{equation}
        \bm{H}_C\left(\bm{X}\right)=\sum_{i=1}^K\bm{H}_C^{\left(i\right)}\left(\bm{X}\right)
    \end{equation}
    where (at fixed $i$) $\left\{\bm{H}_C^{\left(i\right)}\left(\bm{X}\right)\right\}_{\bm{X}}$ is mutually commuting. We call the quantum algorithm:
    \begin{equation}
        \bm{\mathcal{A}}_p\left(\bm{X},\omega\right)=\Tr_{\left\{i\right\}_{i=n+1}^{n+N}}\left(\ket{\psi_p\left(\bm{X},\omega\right)}\bra{\psi_p\left(\bm{X},\omega\right)}\right)
    \end{equation}
    a \emph{Lindbladian evolution} algorithm, where $\ket{\psi_p\left(\bm{X},\omega\right)}$ is defined as the $\left(n+N\right)$-qubit state:
    \begin{equation}
        \ket{\psi_p\left(\bm{X},\omega\right)}:=\prod_{l=p}^1\left(\exp\left(-\ci\beta_l\bm{H}_B\left(\omega\right)\right)\exp\left(-\ci\delta_l\bm{H}_I\left(\omega\right)\right)\left(\prod_{i=K}^1\exp\left(-\ci\frac{\gamma_l^{\left(i\right)}}{\sqrt{n}}\bm{H}_C^{\left(i\right)}\left(\bm{X}\right)\right)\right)\right)\ket{\psi_0\left(\omega\right)}
    \end{equation}
    for some choice of initial state $\ket{\psi_0\left(\omega\right)}$ and parameters $\left\{\beta_l,\delta_l,\gamma_l^{\left(i\right)}\right\}_{l\in\left[p\right],i\in\left[K\right]}$.
\end{definition}
In an almost identical fashion to Corollary~\ref{cor:p_trott_qa}, one can prove that this class of algorithms is stable.
\begin{corollary}[Lindbladian evolution is stable]\label{cor:lind_ev_stab}
    Assume $\bm{H}_B\left(\omega\right)$ is $d$-local and that
    \begin{equation}
        \bm{H}_C^{\left(i\right)}\left(\bm{X}\right)=\frac{1}{\sqrt{Z}}\sum_{j\in\mathcal{I}_i}X_j\bm{P}_j
    \end{equation}
    for some index set $\mathcal{I}_i$ labeling $d$-local Pauli operators $\bm{P}_j$, where $Z=\sum_{i=1}^K\left\lvert\mathcal{I}_i\right\rvert\geq n$. Assume both are supported on interaction hypergraphs of degree at most $\mathfrak{d}$. Then, the Lindbladian evolution algorithm described in Definition~\ref{def:therm_algs} is $\left(\sqrt{n},\lambda_p,\mathfrak{d},0,0\right)$-stable, where:
    \begin{equation}
        \lambda_p=\frac{\left\lVert\bm{\theta}\right\rVert_\infty}{4\sqrt{2n}}\left(\frac{3}{2}\max\left(2,d\right)\mathfrak{d}\right)^{\left(K+2\right)p}.
    \end{equation}
\end{corollary}
\begin{proof}
    This follows identically to Corollary~\ref{cor:p_trott_qa} by taking $K+1\to K+2$ and recalling that the locality of the terms composing $\bm{H}_I\left(\omega\right)$ is $2$.
\end{proof}

\section{Examples of Local Shadows Estimators}\label{sec:classical_shadows}

\subsection{Preliminaries}

We begin by recalling notation used in the main text. We use $\mathcal{S}_n^{\mathrm{m}}$ to denote the space of states on $n$ qubits, $\mathcal{O}_n$ to denote the space of Hermitian observables on $n$ qubits, and $\mathcal{B}_6$ to denote the set of $n$-$d$it strings with $d=6$ that are classical representations of Pauli basis states. We use the notation $\ket{\bm{b};\bm{s}}$ to represent elements of $\mathcal{B}_6$, where $\bm{b}\in\left\{1,2,3\right\}^{\times n}$ labels an $n$-qubit Pauli operator and $\bm{s}\in\left\{0,1\right\}^{\times n}$ the eigenstates of the operator labeled by $\bm{b}$. As we will only be interested in expectation values of observables in states $\ket{\bm{b};\bm{s}}\in\mathcal{B}_6$, we will often abuse notation and write an expectation value as:
\begin{equation}
    \bra{\bm{b};\bm{s}}\bm{O}\ket{\bm{b};\bm{s}}
\end{equation}
for $\bm{O}\in\mathcal{O}_n$; this should be understood as an expectation value of $\bm{O}$ in the Pauli basis state $\ket{\psi}\in\mathbb{C}^{2^n}$ labeled by $\ket{\bm{b};\bm{s}}$. We will similarly ``equate'' operators in the $d$it representation with operators in the qubit representation $\mathbb{C}^{2^n\times 2^n}$, and this should be considered as equating expectation values of the two under this correspondence.

We now restate our definition of a local shadows estimator, given as Definition~\ref{def:eff_loc_shad_est} in the main text. 

\locshad*

Informally, we say that an $\left(\delta,p_{\mathrm{est}},p_{\mathrm{b}}\right)$-efficient local shadows estimator of a class of random Hamiltonians $\mathcal{H}$ exists if, given a state $\bm{\rho}$, one can:
\begin{enumerate}
    \item Construct an instance-independent description of $\bm{\rho}$ out of $n$-qubit Pauli basis states using (convex combinations of) tensor product channels.
    \item Use these Pauli basis states in a linear estimator of the ground state energy, achieving one-sided multiplicative error $\delta$ with probability at least $1-p_{\mathrm{est}}$.
\end{enumerate}
Such a description of $\bm{\rho}$ is what is known as a \emph{classical shadows representation} of $\bm{\rho}$~\cite{huang2020predicting}.

\subsection{Pauli Shadows}

We first discuss a simple case of the Pauli shadows algorithm of \revref\cite{huang2020predicting}. We specialize to the setting where the $\bm{H}_i$ are $k$-local Pauli operators, i.e.,
\begin{equation}\label{eq:k_local_ham_shadows}
    \bm{H}_{\bm{X}}=\frac{1}{\sqrt{Z\left(p,n\right)}}\sum_{i=1}^D X_i\bm{P}_i.
\end{equation}
We assume that $Z\left(p,n\right)=\operatorname{\Omega}\left(n\right)$ and that $\bm{H}_{\bm{X}}$ exhibits the self-averaging property described in Proposition~\ref{prop:self_averaging}.

In this setting, the superoperator associated with the algorithm is:
\begin{equation}
    \bm{\mathcal{M}}\left(\bm{\rho}\right)=3^{-n}\sum_{\bm{b}\in\left\{1,2,3\right\}^{\times n}}\sum_{\bm{s}\in\left\{0,1\right\}^{\times n}}\bra{\bm{b};\bm{s}}\bm{\rho}\ket{\bm{b};\bm{s}}\ket{\bm{b};\bm{s}}\bra{\bm{b};\bm{s}}.
\end{equation}
The energy estimator as given in Eq.~\eqref{eq:estimator} is then just a scaling of the original observable:
\begin{equation}\label{eq:rescaling_r}
    \bm{\mathcal{R}}\left(\bm{H}_{\bm{X}}\right)=3^k\bm{H}_{\bm{X}}.
\end{equation}
The Pauli shadows estimator has variance given by the square of the \emph{shadow norm}~\cite{huang2020predicting}:
\begin{equation}
    \left\lVert\bm{O}\right\rVert_{\mathrm{shadow}}:=\sup_{\bm{\sigma}\in\mathcal{S}_n^{\mathrm{m}}}\left(3^{-n}\sum_{\bm{b}\in\left\{1,2,3\right\}^{\times n}}\sum_{\bm{s}\in\left\{0,1\right\}^{\times n}}\bra{\bm{b};\bm{s}}\bm{\sigma}\ket{\bm{b};\bm{s}}\bra{\bm{b};\bm{s}}\left(\bm{\mathcal{D}}^{\otimes n}\right)^{-1}\left(\bm{O}\right)\ket{\bm{b};\bm{s}}^2\right)^{\frac{1}{2}},
\end{equation}
where $\bm{\mathcal{D}}^{-1}$ is the inverse of the single-qubit depolarizing channel with loss parameter $\frac{1}{3}$:
\begin{equation}
    \bm{\mathcal{D}}^{-1}\left(\bm{A}\right)=3\bm{A}-\Tr\left(\bm{A}\right)\bm{I}.
\end{equation}
We explicitly bound the shadow norm for the special case of $k$-local traceless observables.
\begin{lemma}[Shadow norm bound]\label{lem:shadow_norm_bound}
    Consider the traceless:
    \begin{equation}
        \bm{O}=\sum_{i=1}^m c_i\bm{P}_i,
    \end{equation}
    where each $\bm{P}_i$ is a Pauli operator that has support on $k$ qubits and $c_i\in\mathbb{R}$. Then:
    \begin{equation}
        \left\lVert\bm{O}\right\rVert_{\mathrm{shadow}}\leq 3^k\sup_{\substack{\bm{b}\in\left\{1,2,3\right\}^{\times n}\\\bm{s}\in\left\{0,1\right\}^{\times n}}}\left\lvert\bra{\bm{b};\bm{s}}\bm{O}\ket{\bm{b};\bm{s}}\right\rvert\leq 3^k\left\lVert\bm{O}\right\rVert_{\mathrm{op}}.
    \end{equation}
\end{lemma}
\begin{proof}
    This follows from direct calculation:
    \begin{equation}
        \begin{aligned}
            \left\lVert\bm{O}\right\rVert_{\mathrm{shadow}}^2=&\sup_{\bm{\sigma}\in\mathcal{S}_n^{\mathrm{m}}}3^{-n}\sum_{\bm{b}\in\left\{1,2,3\right\}^{\times n}}\sum_{\bm{s}\in\left\{0,1\right\}^{\times n}}\bra{\bm{b};\bm{s}}\bm{\sigma}\ket{\bm{b};\bm{s}}\bra{\bm{b};\bm{s}}\left(\bm{\mathcal{D}}^{\otimes n}\right)^{-1}\left(\bm{O}\right)\ket{\bm{b};\bm{s}}^2\\
            =&\sup_{\bm{\sigma}\in\mathcal{S}_n^{\mathrm{m}}}3^{-n}\sum_{\bm{b}\in\left\{1,2,3\right\}^{\times n}}\sum_{\bm{s}\in\left\{0,1\right\}^{\times n}}\\
            &\sum_{i,j=1}^m c_i c_j\bra{\bm{b};\bm{s}}\bm{\sigma}\ket{\bm{b};\bm{s}}\bra{\bm{b};\bm{s}}\left(\bm{\mathcal{D}}^{\otimes n}\right)^{-1}\left(\bm{P}_i\right)\ket{\bm{b};\bm{s}}\bra{\bm{b};\bm{s}}\left(\bm{\mathcal{D}}^{\otimes n}\right)^{-1}\left(\bm{P}_j\right)\ket{\bm{b};\bm{s}}\\
            =&9^k\sup_{\bm{\sigma}\in\mathcal{S}_n^{\mathrm{m}}}3^{-n}\sum_{\bm{b}\in\left\{1,2,3\right\}^{\times n}}\sum_{\bm{s}\in\left\{0,1\right\}^{\times n}}\sum_{i,j=1}^m c_i c_j\bra{\bm{b};\bm{s}}\bm{\sigma}\ket{\bm{b};\bm{s}}\bra{\bm{b};\bm{s}}\bm{P}_i\ket{\bm{b};\bm{s}}\bra{\bm{b};\bm{s}}\bm{P}_j\ket{\bm{b};\bm{s}}\\
            =&9^k\sup_{\bm{\sigma}\in\mathcal{S}_n^{\mathrm{m}}}3^{-n}\sum_{\bm{b}\in\left\{1,2,3\right\}^{\times n}}\sum_{\bm{s}\in\left\{0,1\right\}^{\times n}}\bra{\bm{b};\bm{s}}\bm{\sigma}\ket{\bm{b};\bm{s}}\bra{\bm{b};\bm{s}}\bm{O}\ket{\bm{b};\bm{s}}^2\\
            \leq&9^k\sup_{\substack{\bm{b}\in\left\{1,2,3\right\}^{\times n}\\\bm{s}\in\left\{0,1\right\}^{\times n}}}\bra{\bm{b};\bm{s}}\bm{O}\ket{\bm{b};\bm{s}}^2.
        \end{aligned}
    \end{equation}
\end{proof}
In particular, by the self-averaging property of $\bm{H}_{\bm{X}}$ (see Proposition~\ref{prop:self_averaging}), we have for any constant $t>0$ and sufficiently large $n$ that:
\begin{equation}
    \begin{aligned}
        \mathbb{P}_{\bm{X}}\left[\mathcal{V}\right]&:=\mathbb{P}_{\bm{X}}\left[\left\lVert\bm{H}_{\bm{X}}\right\rVert_{\mathrm{shadow}}\geq 3^k E^\ast\sqrt{n}+t\sqrt{n}\right]\\
        &\leq\mathbb{P}_{\bm{X}}\left[\left\lVert\bm{H}_{\bm{X}}\right\rVert_{\mathrm{shadow}}\geq 3^k E_n^\ast\sqrt{n}+t\sqrt{n}\right]\\
        &\leq\exp\left(-\operatorname{\Omega}\left(n\right)\right).
    \end{aligned}
\end{equation}
This immediately gives a bound on the failure probability $p_{\mathrm{est}}$ of the Pauli shadows algorithm. From Cantelli's inequality and the union bound, we have for any constant $\delta>0$ and $t>0$ that, conditioned on $\mathcal{V}$,
\begin{equation}\label{eq:pauli_shadows_err_bound}
    \mathbb{P}_{\left(\bm{X},\omega\right)}\left[\Tr\left(\bm{\mathcal{R}}\left(\bm{H}_{\bm{X}}\right)\bm{\tilde{\mathcal{M}}}\left(\bm{\rho},\omega\right)\right)-\Tr\left(\bm{H}_{\bm{X}}\bm{\rho}\right)\geq-\delta E^\ast\sqrt{n}\mid\mathcal{V}\right]\geq 1-\frac{\left(3^k E^\ast+t\right)^2}{\left(3^k E^\ast+t\right)^2+\delta^2 E^{\ast 2}}.
\end{equation}
This implies that the Pauli shadows algorithm is an efficient local shadows estimator for the class $\mathcal{H}$ of $k$-local disordered Hamiltonians of the form of Eq.~\eqref{eq:k_local_ham_shadows}.
\begin{proposition}[The Pauli shadows algorithm is an efficient local shadows estimator]\label{prop:pauli_shadows_params}
    For any choice of $\delta>0$ and sufficiently large $n$, the Pauli shadows algorithm~\cite{huang2020predicting} is an $\left(\delta,p_{\mathrm{est}},\exp\left(-\operatorname{\Omega}\left(n\right)\right)\right)$-efficient local shadows estimator for the class of $k$-local Hamiltonians given in Eq.~\eqref{eq:k_local_ham_shadows}, where:
    \begin{equation}
        p_{\mathrm{est}}=\frac{1}{1+0.99\times 9^{-k}\delta^2}.
    \end{equation}
\end{proposition}
\begin{proof}
    Given $\delta$, this follows from Eq.~\eqref{eq:pauli_shadows_err_bound} by taking $n$ sufficiently large and $t$ sufficiently small.
\end{proof}
This bound can be further improved for the quantum $k$-spin model:
\begin{equation}\label{eq:k_loc_qsgm_app}
    \bm{H}_{k\mathrm{-spin}}=\frac{1}{\sqrt{p\binom{n}{k}}}\sum_{\overline{i}\in\binom{\left[n\right]}{k}}\sum_{\bm{b}\in\left\{1,2,3\right\}^{\times k}}S_{\overline{i},\bm{b}}J_{\overline{i},\bm{b}}\prod_{j=1}^k\bm{\sigma}_{i_j}^{\left(b_j\right)}
\end{equation}
with $p\geq\operatorname{\Omega}\left(n^{-\left(k-1\right)}\right)$, where the $S_{\overline{i},\bm{b}}$ are chosen i.i.d. from the Bernoulli distribution with sparsity parameter $p$ and the $J_{\overline{i},\bm{b}}$ are i.i.d. standard normal random variables. This is due to a known result that the maximal expectation value of $n^{-\frac{1}{2}}\bm{H}_{k\mathrm{-spin}}$ with respect to a Pauli basis state is at most $\operatorname{O}_k\left(1\right)$~\cite[Theorem~2]{anschuetz2024productstates}. As $n^{-\frac{1}{2}}\left\lVert\bm{H}_{k\mathrm{-spin}}\right\rVert_{\mathrm{op}}\geq\operatorname{\Omega}_k\left(k^{-1}3^{\frac{k}{2}}\right)$~\cite[Corollary~D.2]{anschuetz2024strongly}, we have by Lemma~\ref{lem:shadow_norm_bound}:
\begin{equation}
    \frac{1}{\sqrt{n}}\left\lVert\bm{H}_{k\mathrm{-spin}}\right\rVert_{\mathrm{shadow}}\leq\operatorname{O}_k\left(3^k\right)\leq\operatorname{O}_k\left(k 3^{\frac{k}{2}}E^\ast\right).
\end{equation}
This immediately strengthens Proposition~\ref{prop:pauli_shadows_params} in this specific case.
\begin{proposition}[The Pauli shadows algorithm is an efficient local shadows estimator for the quantum $k$-spin model]\label{prop:pauli_shadows_params_k_loc_spin}
    For any choice of $\delta>0$ and sufficiently large $n$, the Pauli shadows algorithm~\cite{huang2020predicting} is an $\left(\delta,p_{\mathrm{est}},\exp\left(-\operatorname{\Omega}\left(n\right)\right)\right)$-efficient local shadows estimator for the class of quantum $k$-spin model Hamiltonians (Eq.~\eqref{eq:k_loc_qsgm_app}), where:
    \begin{equation}
        p_{\mathrm{est}}=\frac{1}{1+0.99\times k^{-2}3^{-k}\delta^2}.
    \end{equation}
\end{proposition}

\subsection{Derandomized Pauli Shadows}

One may also consider a derandomized variant of the Pauli shadows estimator~\cite{PhysRevLett.127.030503}. We here assume a setting where the Hamiltonian is able to be grouped into sums of terms mutually diagonalized by some Pauli frame $\bm{b}\in\mathcal{P}\subseteq\left\{1,2,3\right\}^{\times n}$:
\begin{equation}
    \bm{H}_{\bm{X}}=\frac{1}{\sqrt{\left\lvert\mathcal{P}\right\rvert Z\left(p,n\right)}}\sum_{\bm{b}\in\mathcal{P}}\sum_{i\in\mathcal{I}_{\bm{b}}}X_{\bm{b},i}\bm{P}_{\bm{b},i}=:\frac{1}{\sqrt{\left\lvert\mathcal{P}\right\rvert}}\sum_{\bm{b}\in\mathcal{P}}\bm{H}_{\bm{b},\bm{X}}.
\end{equation}
Here, $\mathcal{I}_{\bm{b}}$ is some index set associated with the Pauli frame $\bm{b}$. We assume for simplicity that $\left\lvert\mathcal{P}\right\rvert$ is independent of $n$, and once again assume that $Z\left(p,n\right)=\operatorname{\Omega}\left(D\right)$ and that $\bm{H}_{\bm{X}}$ exhibits the self-averaging property described in Proposition~\ref{prop:self_averaging}.

In this setting, the derandomized variant of the Pauli shadows estimator has lower variance than the original estimator. This is because one only considers bases compatible with $\bm{H}_{\bm{X}}$. Specifically, one uses the energy estimator~\cite{PhysRevLett.127.030503}:
\begin{equation}
    \bm{\mathcal{R}}\left(\bm{H}_{\bm{X}}\right)=\left\lvert\mathcal{P}\right\rvert\bm{H}_{\bm{X}},
\end{equation}
and the superoperator $\bm{\mathcal{M}}$ associated with the algorithm is the dephasing channel in the Pauli frames labeled by $\bm{b}\in\mathcal{P}$:
\begin{equation}
    \bm{\mathcal{M}}\left(\bm{\rho}\right)=\frac{1}{\left\lvert\mathcal{P}\right\rvert}\sum_{\bm{b}\in\mathcal{P}}\sum_{\bm{s}\in\left\{0,1\right\}^{\times n}}\bra{\bm{b};\bm{s}}\bm{\rho}\ket{\bm{b};\bm{s}}\ket{\bm{b};\bm{s}}\bra{\bm{b};\bm{s}}.
\end{equation}
The estimator has variance upper-bounded by $\left\lvert\mathcal{P}\right\rvert^2\left\lVert\bm{H}_{\bm{X}}\right\rVert_{\mathrm{op}}^2$, giving the following result again by Cantelli's inequality.
\begin{proposition}[The derandomized Pauli shadows algorithm is an efficient local shadows estimator]\label{prop:derand_pauli_shadows_params}
    For any choice of $\delta>0$ and sufficiently large $n$, the derandomized Pauli shadows algorithm~\cite{PhysRevLett.127.030503} is an $\left(\delta,p_{\mathrm{est}},\exp\left(-\operatorname{\Omega}\left(n\right)\right)\right)$-efficient local shadows estimator for the class of $k$-local Hamiltonians given in Eq.~\eqref{eq:k_local_ham_shadows}, where:
    \begin{equation}
        p_{\mathrm{est}}=\frac{1}{1+0.99\left\lvert\mathcal{P}\right\rvert^{-1}\delta^2}.
    \end{equation}
\end{proposition}
\begin{proof}
    The result follows identically to Proposition~\ref{prop:pauli_shadows_params}.
\end{proof}

\bibliographystyle{IEEEtran}
\bibliography{main}

\end{document}